\documentclass[12pt,letterpaper]{article}
\usepackage[T1]{fontenc}
\usepackage{lmodern}
\usepackage[utf8]{inputenc}
\usepackage[margin = 1.25in]{geometry}
\usepackage{booktabs}
\usepackage{amsmath,amssymb,amsfonts,amsthm}
\usepackage{thmtools}

\usepackage{cancel}
\usepackage{bm}
\usepackage{mathtools}
\usepackage{changepage}
\usepackage{verbatim}
\usepackage{graphicx}
\usepackage{emptypage}

\usepackage[all,cmtip]{xy}
\usepackage[bottom]{footmisc}
\usepackage{fnpct}
\usepackage[titletoc,title,page]{appendix}
\usepackage[nottoc,numbib]{tocbibind}
\usepackage{chngcntr}
\usepackage{apptools}
\usepackage{listings}
\usepackage{color}
\usepackage{sgame, tikz} 
\usepackage{pgfplots}

\usepackage{microtype}    
\usepackage[normalem]{ulem}    

\AtAppendix{\counterwithin{lem}{section}}
\usepackage{caption}
\usepackage{subcaption}
\usepackage{float}
\usepackage{xr-hyper}
\definecolor{GmailBlue}{RGB}{42, 93, 176}
\definecolor{myblue}{HTML}{0072B2}
\PassOptionsToPackage{hyphens}{url}\usepackage[colorlinks=true,linkcolor=GmailBlue,allcolors=GmailBlue]{hyperref}
\usepackage{sgamevar}
\usepackage{lipsum}
\usepackage{xcolor}
\theoremstyle{plain} 
\newtheorem{theorem}{Theorem}
\newtheorem*{theorem*}{Theorem}
\newtheorem{corollary}{Corollary}
\newtheorem{lemma}{Lemma}
\newtheorem*{lemma*}{Lemma}

\newtheorem{proposition}{Proposition}

\theoremstyle{definition} 
\newtheorem{definition}{Definition}
\newtheorem*{definition*}{Definition}

\newtheorem*{prop*}{Properties}

\newtheorem*{ex*}{Example}

\newtheorem{remark}{Remark}
\newtheorem{assumption}{Assumption}

\newcommand{\argmax}{\operatornamewithlimits{argmax}}

\usepackage[nodisplayskipstretch]{setspace}
\reversemarginpar

\makeatletter
\renewcommand{\paragraph}{%
	\@startsection{paragraph}{4}%
	{\z@}{1.5ex \@plus 1ex \@minus .2ex}{-0.7em}%
	{\normalfont\normalsize\bfseries}%
}
\makeatother

\let\originalleft\left
\let\originalright\right
\renewcommand{\left}{\mathopen{}\mathclose\bgroup\originalleft}
\renewcommand{\right}{\aftergroup\egroup\originalright}

\usepackage[nameinlink]{cleveref}
\crefname{manualasm}{assumption}{assumptions}
\crefname{innercustomthm}{theorem}{theorems}
\crefname{prop}{proposition}{propositions}
\crefname{ex}{example}{examples}
\crefname{defn}{definition}{definitions}
\crefname{claim}{claim}{claims}
\crefname{lem}{lemma}{lemmas}
\crefname{thm}{theorem}{theorems}
\Crefname{assumption}{Assumption}{Assumptions}
\Crefname{lemma}{Lemma}{Lemmata}

\usepackage{etoolbox}
\newcommand{\zerodisplayskips}{%
  \setlength{\abovedisplayskip}{1.5ex}%
  \setlength{\belowdisplayskip}{1.5ex}%
  \setlength{\abovedisplayshortskip}{1.5ex}%
  \setlength{\belowdisplayshortskip}{1.5ex}}
\appto{\normalsize}{\zerodisplayskips}
\appto{\small}{\zerodisplayskips}
\appto{\footnotesize}{\zerodisplayskips}

\newcommand{\headercref}[2]{\texorpdfstring{\Cref{#2}}{#1 \ref{#2}}} 

\allowdisplaybreaks
\usepackage[shortlabels]{enumitem}
    \setlist[itemize]{noitemsep,nolistsep}
    \setlist[enumerate,1]{noitemsep,nolistsep,label=(\arabic*)}
    \setlist[enumerate,2]{noitemsep,nolistsep,label=(\alph*)}
    \setlist[enumerate,3]{noitemsep,nolistsep,label=(\roman*)}

\DeclareMathOperator*{\MPC}{MPC}
\DeclareMathOperator*{\MPS}{MPS}
\DeclareMathOperator*{\supp}{supp}
\DeclareMathOperator*{\BR}{BR}
\DeclareMathOperator*{\IC}{IC}

\DeclareMathOperator*{\ccv}{ccv}
\DeclareMathOperator*{\cvx}{cvx}
\DeclareMathOperator*{\interior}{int}

\DeclareMathOperator*{\sign}{sgn}
\newcommand{\Ccv}{\mathcal{I}_{\ccv}}
\newcommand{\Cvx}{\mathcal{I}_{\cvx}}

\newcommand{\de}{\mathop{}\!\mathrm{d}}

\usepackage{accents}

\usepackage{natbib}
\pgfplotsset{compat=1.17}

\title{Delegated Information Provision\thanks{\protect We thank 
Martino Banchio,
Gregorio Curello,
Alkis Georgiadis-Harris,
Nicola Pavoni,
Marco Ottaviani,
Mark Whitmeyer,
and seminar audiences at 
Bielefeld, Bocconi, Essex, CSEF-IGIER Symposium 2026, EEA-ESEM 2026, EWMD 2026, NET, Theros 2026, and TTW 2026
for their comments.
Preusser acknowledges financial support by the European Research Council (HEUROPE 2022 ADG, GA No. 101055295 – InfoEcoScience).
}

}
\author{Francesco Bilotta\thanks{Università Bocconi. Email: \href{mailto:francesco.bilotta2@phd.unibocconi.it}{francesco.bilotta2@phd.unibocconi.it}}
\and
Christoph Carnehl\thanks{Università Bocconi and IGIER. Email: \href{mailto:christoph.carnehl@unibocconi.it}{christoph.carnehl@unibocconi.it}}
\and
Justus Preusser\thanks{Università Bocconi and IGIER. Email: \href{mailto:justus.preusser@unibocconi.it}{justus.preusser@unibocconi.it}}
}

\date{\today}
\pgfplotsset{compat=1.15}
\begin{document}
\maketitle

\begin{abstract}

\noindent A designer relies on an experimenter to provide information to a decision maker. The experimenter has incentives to persuade rather than merely transmit information. Anticipating this motive, the designer can restrict the set of admissible experiments, but cannot prevent the experimenter from garbling any admissible experiment. We model this situation as delegation over experiments. The optimal delegation set is obtained by comparing maximally informative experiments among those the experimenter has no incentive to garble. When the experimenter's preferences are $S$-shaped, we characterize these experiments as double censorship. Relative to the full-delegation benchmark, double censorship features an intermediate pooling region, inducing a smaller pooling region for the highest states. We show that the designer strictly benefits from imposing a nontrivial delegation set that constrains persuasion while retaining information provision. Applying our results to recommender systems, we show that privacy constraints can arise endogenously to protect consumers against persuasion.

\end{abstract}

\pagenumbering{arabic}

\section{Introduction}\label{sec:intro}

In many environments, delegation separates the production of information from its use in decision making. In courts, prosecutors generate evidence that informs judges; in regulatory settings, firms conduct tests that guide approval decisions; and on platforms, algorithms recommend products to consumers.
In each case, the party who produces information has incentives to influence the decision maker rather than merely transmit facts. While a large literature on Bayesian persuasion characterizes how the experimenter optimally designs information to exploit this influence, much less attention has been paid to the complementary problem of a designer who anticipates persuasion and seeks to limit its force. To do so, a designer may restrict the experimenter's discretion by specifying which experiments are admissible. Yet even under such restrictions, delegation remains subject to moral hazard: the experimenter can garble any admissible experiment before it reaches the decision maker. We study the optimal design of delegated information provision, which trades off the informativeness of admissible experiments against the experimenter's ability to persuade.

In practice, such control over informational discretion often takes the form of restrictions on admissible experiments rather than direct control of actions. In the examples above, legislators restrict which evidence prosecutors may present, regulators specify testing procedures for firms, and privacy protection limits which data algorithms can use. These restrictions impose an upper bound on informativeness while leaving substantial discretion to the experimenter. As a result, the experimenter may still suppress available evidence, implement tests imperfectly, or communicate only coarse information.

We therefore cast the design of restrictions as a delegation problem over experiments. The designer specifies a delegation set: a set of admissible Blackwell-experiments. The experimenter can choose an arbitrary garbling of any admissible experiment. Therefore, the designer can only impose an upper bound on information. Given a delegation set, the experimenter faces a standard Bayesian persuasion problem with a one-dimensional state and posterior-mean preferences. The experimenter picks an experiment to persuade a decision maker (DM) to take a binary action. The DM's optimal choice depends on a privately known outside option. 
We assume that the designer maximizes a weighted sum of the experimenter's payoff and a component that is convex in the posterior mean (such as the DM's posterior-mean payoff). 
Due to this assumption, the designer need not necessarily be distinct from the DM but can also be interpreted as the DM themselves who commits ex ante to an admissibility rule as a form of self-discipline against persuasion.

Our central insight is that effective limits on the experimenter's discretion cannot deliver experiments that are uniformly more informative than the full-delegation benchmark---the outcome of the traditional Bayesian persuasion problem. Any uniformly more informative restriction will be garbled by the experimenter, while any uniformly less informative restriction is dominated by the persuasion outcome. The designer's only viable strategy is thus to reshape the admissible information structures by selecting an experiment that is Blackwell-incomparable to persuasion.
Such restrictions redistribute informativeness across states: by deliberately inducing pooling of some regions of the state space, the designer dampens incentives for concealment and induces finer disclosure where persuasion would otherwise be most severe. We show that optimally designed restrictions strictly dominate full delegation, adding value not via more information per se, but with a different shape of information revelation. 

As a first step, we show that the designer's problem reduces to choosing among \emph{maximal incentive-compatible} (MIC) experiments. An experiment $F$ is incentive compatible if it is self-enforcing: when the designer sets $F$ as upper bound, the experimenter optimally implements $F$ rather than garble it. An incentive-compatible experiment is maximal if no strictly more informative experiment is also incentive compatible. Since the designer values both the experimenter's payoff and information, the designer never imposes a restriction that is incentive compatible but not maximal. 

While incentive compatibility and maximality are intuitive properties of optimal delegation, they already yield significant structure on the designer's problem. A key implication of MIC is that any profitable restriction must be Blackwell-incomparable to the persuasion outcome under full delegation $F^*$. Blackwell-more informative restrictions are garbled back to $F^*$, while Blackwell-less informative restrictions are dominated by $F^*$. Therefore, any strict gain from restricting discretion must arise from Blackwell-incomparable experiments that reallocate informativeness across regions of the state space.

This fact is particularly evident when the experimenter's payoff is $S$-shaped, which is our baseline assumption. In this case, we fully characterize maximal incentive-compatible (MIC) experiments. They take the form of \emph{double censorship}: there exist two thresholds such that all states below the first threshold are fully revealed, while states between the thresholds and above the second threshold are pooled into one intermediate and one high signal, respectively. Full delegation is itself MIC and corresponds to the unrestricted $S$-shaped persuasion outcome; it arises as the limiting case in which the two thresholds coincide, yielding \emph{upper censorship}.

Double censorship makes the designer's informational trade-off transparent. Incentive compatibility implies that shrinking the top pooling region---i.e., improving information precisely in the high states the experimenter would otherwise censor---requires expanding the intermediate pooling region. Thus, better information in high states comes at the cost of worse information in intermediate ones. MIC experiments are therefore ordered by the size of the top pooling interval. However, they are not Blackwell-comparable: each MIC experiment redistributes informativeness between high and intermediate states rather than uniformly increasing or decreasing it. This redistribution logic and the ensuing incomparability are the key factors generating scope for improvements beyond full delegation.

In fact, leveraging the structure of MIC experiments, we show that the designer strictly prefers a nontrivial restriction: full delegation is never optimal. Intuitively, starting from full delegation, introducing an arbitrarily small amount of pooling around the censorship cutoff entails only a second-order informational loss, yet it disciplines the experimenter's incentives and reduces the pooling region at the highest states. The resulting gain is of first order, as it yields a benefit proportional to the mass of states pooled under persuasion. Hence, improving information in high states strictly outweighs the loss on intermediate states, yielding a higher payoff for the designer.

To yield practical insights and to better understand how optimal MIC experiments trade off information across states, we turn to comparative statics varying the designer's payoff. We begin by showing that MIC experiments are totally ordered according to a novel \emph{discipline order}. 
An MIC double censorship experiment is more disciplined if it admits a larger intermediate pooling region but a smaller top pooling region.
We show that the designer optimally implements a more disciplined experiment, for example, if the designer's objective places higher weight on types of the DM with high outside options, or a lower weight on the experimenter's utility.

We illustrate our results in the context of platforms employing recommendation algorithms to inform and to persuade consumers to purchase. Platforms have access to vast data about consumer demographics and their browsing behavior that they can use to estimate consumers' valuation. While recommendation algorithms frequently take revenue maximization as objective, platform operators may care not only about revenues but also about consumers.\footnote{Alternatively, a regulatory authority may want to regulate a platform's recommendation algorithm to protect consumers.} We show that the tension between persuasion and information provision can lead to endogenously arising privacy constraints---even in the absence of intrinsic preferences for privacy---, as our optimal restrictions can be implemented through privacy sets as defined in \citet{strack2024privacy}. Our comparative static results imply that the platform optimally imposes more stringent privacy constraints for the recommendation algorithm as their weight on consumers relative to revenues increases.

Finally, we argue that our main results extend beyond the $S$-shaped case, providing general conditions for the suboptimality of full delegation. Two simple examples, an uninformed decision-maker and $M$-shaped experimenter preferences, illustrate these general conditions, which relate to the experimenter's posterior-mean payoff and the full-delegation solution.

\hypertarget{par:related_literature}{\paragraph*{Related Literature.}}
Most closely related are \citet{ichihashi2019limiting}, \citet{arieli2022bayesianpersuasionmediators}, and \citet{mylovanov2024constructive}. 
\citet{ichihashi2019limiting} also studies a designer limiting the information available to an experimenter, characterizing the set of attainable payoff profiles. We differ in both setting and focus: By imposing additional structure on the persuasion problem, we derive the structure of optimal restrictions, their implementation and comparative statics. In contrast to \citet{ichihashi2019limiting}, the designer strictly benefits from imposing a restriction.
\citet{arieli2022bayesianpersuasionmediators} study persuasion with strategic mediators who sequentially garble information. They provide geometric characterizations for abstract settings. We focus on a single experimenter and leverage the tractability of the posterior-mean environment to characterize optimal restrictions and their economic properties.
\citet{mylovanov2024constructive} compare sequential obfuscation to sequential disclosure, corresponding to Blackwell upper- and lower bounds by parties with opposing preferences. 
In our environment, the designer can impose upper bounds but faces an experimenter with partially aligned preferences and the ability to garble any experiment.

Our model connects to the literature on attention management. In attention-management models, a receiver rationally ignores a sender's information due to attention costs. Our model nests a class of such problems,\footnote{Namely, by identifying the designer with the sender and the experimenter with the receiver. The sender has convex posterior-mean payoffs, and the receiver's costs are posterior-mean separable.} and our characterization implies that full disclosure is not optimal. A close precedent is an example by \citet{lipnowski2020attention} that also emphasizes informational trade-offs across states.\footnote{Other related papers in this strand of literature are \citet{wei2021persuasion,BloedelSegal2021,lipnowski2022optimal,song2025robust,jain2026competitive,dallara2026persuadinginattentiveprivatelyinformed}.}

A related strand studies information design under exogenously given constraints on information structures, such as limits on the number of messages or monotonicity requirements.\footnote{See \citet{le2019persuasion,babichenko2021bayesianpersuasionexante,ivanov2021optimal,mensch2021monotone,tsakas2021noisy,ball2022experimental,onuchic2023conveying,antic2024selected,aybas2024persuasioncoarsecommunication,lyu2025coarseinformationdesign}. \citet{antic2024selected} also consider how a DM can benefit from particular restrictions on information structures, namely by ruling out ex-ante which state-contingent pieces of evidence may be presented.}
We instead make the restriction itself the design object.

The persuasion literature has studied other reasons why information transmission may improve in environments with persuasion incentives.
\citet{tsakas2021resisting} ask whether the DM can gain by committing to burning money when taking certain actions.
\citet{curello2024comparativestaticspersuasion} ask which changes to the experimenter's payoffs induce the experimenter to choose a more informative experiment.\footnote{\citet{curello2024comparativestaticspersuasion} also provide a result concerning changes to the prior. We comment on how this result relates to our problem in \Cref{footnote:curello_sinander} further ahead.}
We instead focus on informational restrictions on admissible experiments.

\citet{kleiner2021extreme} and \citet{kolotilin2025persuasion} establish equivalences between a class of delegation and persuasion problems. These equivalences do not apply to our setting since the designer delegates persuasion itself to the experimenter.

Finally, we contribute to the literature on information design subject to privacy constraints.
\Cref{sec:privacy_application} explains how our restrictions can capture privacy sets, as proposed by \citet{strack2024privacy}.\footnote{As we explain in \Cref{footnote:differential_privacy}, we can also capture the notion of differential privacy (\citet{dwork2006calibrating}). \citet{schmutte2025information} and \citet{pan2025differentially} study information design subject to differential privacy constraints. See also \citet{he2025privateprivateinformation} for the more distant notion of private private information.}
This and other privacy notions are motivated by an intrinsic interest in protecting individual-specific characteristics.
In contrast, we study the optimal design of such constraints by a designer who has no intrinsic preference for privacy but seeks to steer the incentives of a biased experimenter.
We show that privacy constraints may arise nevertheless to protect individuals from persuasion.

\section{Model} \label{sec:model}

Our model features a designer, an experimenter, and a decision maker (DM). 

\paragraph*{State, outside option, and payoffs.}
There is a state $\omega \in [0,1]$ which is the realization of a random variable $\bm{\omega}$ with prior CDF $H$ and strictly positive, continuous density $h$. We denote the prior mean of the state by $\mu = \mathbb{E}_{H}[\bm{\omega}]$. 
There is also an outside option $r\in [0, 1]$ which is the realization of a random variable $\bm{r}$ with prior CDF $G$ and strictly positive, continuously differentiable density $g$.
The outside option is privately observed by the DM and independent of the state.

The DM decides whether or not to act, $a\in \{0,1\}$, where $a=0$ represents inaction and $a=1$ represents action. 
The DM's payoff from inaction is normalized to zero.
When the state is $\omega$ and the outside option is $r$, the DM's payoff from action is $\omega - r$.
Accordingly, when the DM's belief about the state has mean $m$, the DM takes the action if and only if $m \geq r$.\footnote{The DM's tie-breaking will be irrelevant as the CDF $G$ of the outside option is atomless.}

The experimenter always prefers action: the payoff from action is one, and that from inaction is zero. Given posterior mean $m$, the DM acts with probability $G(m)$, so the experimenter's posterior-mean payoff is $G(m)$.

We describe the designer's payoffs in reduced form given the DM's posterior mean. If the posterior mean is $m$, the designer's interim expected payoff $u_{D}(m)$ is given by
\begin{equation*}
    u_{D}(m) = \alpha \gamma(m) + (1 - \alpha) G(m),
\end{equation*}
where $\alpha\in(0,1]$ and $\gamma$ is continuously differentiable and strictly convex. Thus, the designer places positive weight on an information-valuing component and may also place weight on the experimenter's payoff. In the leading welfare case, $\gamma(m) = \mathbb{E}_{G}[\max\lbrace 0, m - \bm{r}\rbrace]$ is the DM's posterior-mean payoff.

We make two assumptions on the distribution of the outside option $G$ and the prior mean of the state $\mu$.
\begin{assumption}\label{assumption:S-shape}
The density $g$ is strictly quasiconcave with an interior maximum $r_{0} \in (0, 1)$.
Accordingly, the CDF $G$ is \emph{S-shaped}: strictly convex on $[0, r_{0}]$, and strictly concave on $[r_{0}, 1]$.
\end{assumption}
\begin{assumption}\label{assumption:informativeness}
    It holds $g(\mu) \mu > G(\mu)$.
\end{assumption}

\Cref{assumption:S-shape} is a common assumption in the persuasion literature \citep[see, for example, ][]{kolotilin2022censorship}. Importantly, the experimenter's payoff $G$ is $S$-shaped, which yields tractability in the experimenter's problem as described momentarily. We discuss more general preferences in \Cref{sec:discussion}.

\Cref{assumption:informativeness} states that the prior mean $\mu$ lies in the convex part of $G$ or close to the inflection point $r_{0}$ in the concave part. The assumption rules out trivial cases: if \Cref{assumption:informativeness} fails, no information transmission would arise; see \Cref{sec:unrestricted_persuasion}.

\paragraph*{Information.}
While the DM privately observes the realized outside option $r$, the DM does not observe the realized state $\omega$. Instead, the experimenter provides information about the state to the DM via an \emph{experiment}. Due to the assumed posterior-mean preferences, we may identify each experiment with its induced CDF $F$ of posterior means. 
It is convenient to consider \emph{integrated CDFs (ICDFs)} and \emph{mean-preserving contractions (MPCs)}.
Given a CDF $F$, we define its ICDF $I_{F}$ by $I_{F}(m) = \int_{0}^{m} F(m^{\prime}) \de m^{\prime}$ for all $m\in [0, 1]$.
A CDF $F$ is an MPC of a CDF $\bar{F}$ if and only if $I_{F} \leq I_{\bar{F}}$ and $I_{F}(1) = I_{\bar{F}}(1)$; analogously, $\bar{F}$ is a mean-preserving spread (MPS) of $F$.
As is well-known, the set of experiments coincides with the set of MPCs of the prior $H$, denoted $\MPC(H)$. 
We equip $\MPC(H)$ with the $L^{1}$-norm.

Let $\delta_{\mu}$ be the degenerate (uninformative) experiment, i.e., the Dirac measure on the prior mean $\mu$. 
Then, we can express $\MPC(H)$ as the set of CDFs $F$ satisfying $I_{\delta_{\mu}} \leq I_{F} \leq I_{H}$.
\Cref{fig:ICDF_example} depicts the ICDF representation and constraint.
\begin{figure}[t!]
    \centering
    \includegraphics[width=0.75\linewidth]{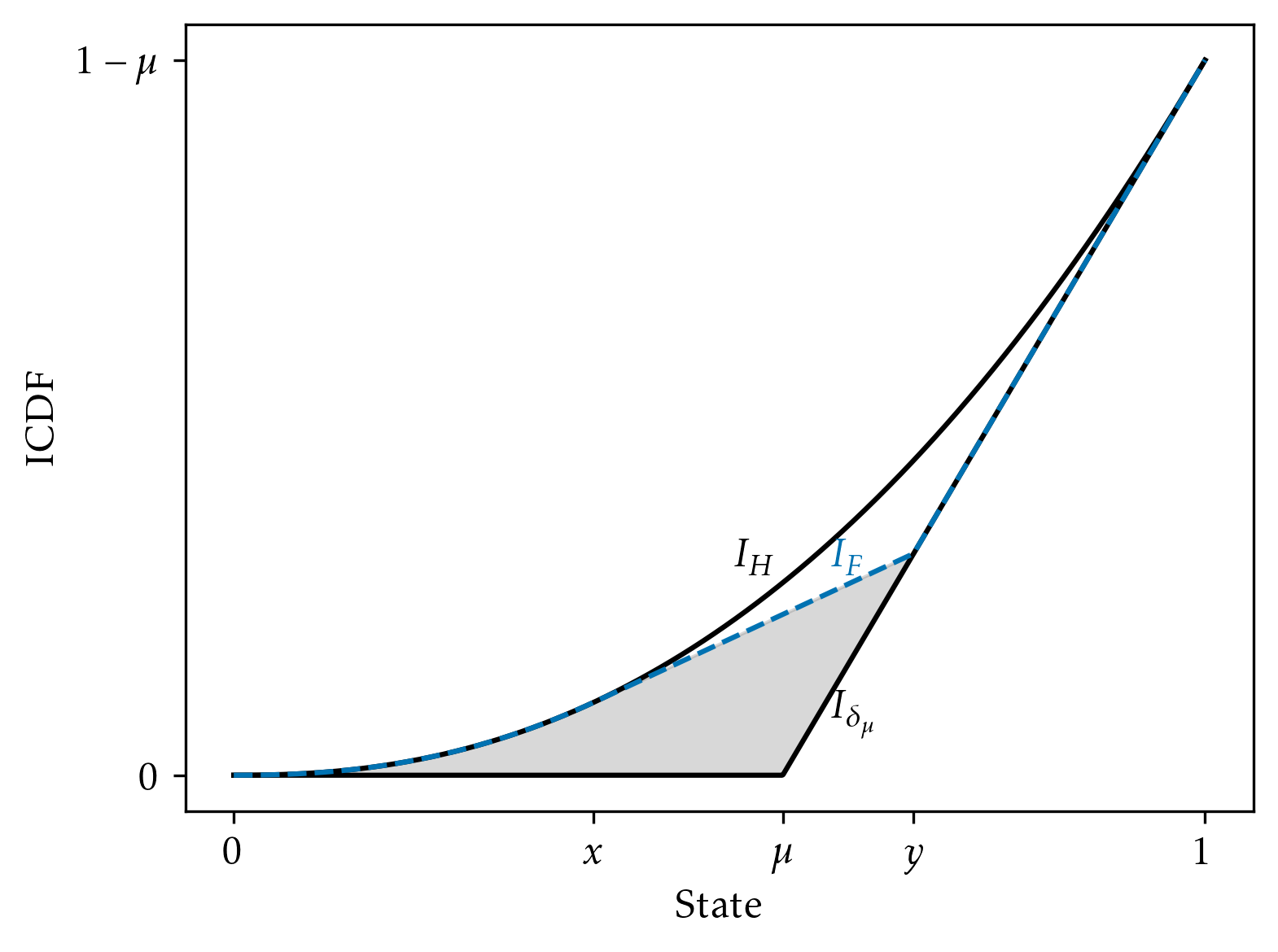}
    \caption{\emph{The ICDFs of the prior $H$, the degenerate experiment $\delta_{\mu}$, and an experiment $F$ that pools $H$ above a threshold $x$ to a point $y$, and else coincides with $H$.}}
    \label{fig:ICDF_example}
\end{figure}

\paragraph*{Restrictions.}
Before the experimenter chooses an experiment, the designer imposes a \emph{restriction} $\bar{F}\in\MPC(H)$: given restriction $\bar{F}$, the experimenter must choose an experiment from the set $\MPC(\bar{F})$ rather than from $\MPC(H)$.
Note, a restriction is an MPS-upper bound on admissible experiments, not an MPS-lower bound: the experimenter can always garble information.
In \Cref{fig:ICDF_example}, if the restriction is $F$, the experimenter can choose among experiments whose ICDFs lie between $I_{F}$ and $I_{\delta_{\mu}}$, as indicated by the shaded area. For example, $F$ is admissible, but $H$ is not.

We explain in \Cref{remark:abstract_restrictions} further ahead that our model also captures general restrictions on information provision, not only restrictions in the form of a single distribution of posterior-means.
We discuss a real-world example of such restrictions in \Cref{sec:privacy_application}.

\paragraph*{Incentive-compatibility and optimality.}
Recall that the experimenter's interim expected payoff from inducing a posterior mean $m$ equals $G(m)$.
Thus, the experimenter's ex-ante payoff from an experiment $F$ equals $\int G(m) \de F(m)$.

Given restriction $\bar{F}$, let $\BR(\bar{F})$ be the set of the experimenter's \emph{best replies}:\footnote{Note that $\BR(\bar{F})$ is non-empty as $\MPC(\bar{F})$ is compact and $F\mapsto \int G(m) \de F(m)$ is continuous.}
\begin{equation*}
    \BR(\bar{F}) = \argmax_{F\in\MPC(\bar{F})} \int G(m)\de F(m).
\end{equation*}
An alternative interpretation of $\MPC(\bar{F})$ is simply the set of all experiments when $\bar{F}$ is the underlying prior distribution of the state.
Accordingly, $F\in\BR(\bar{F})$ means that the experimenter finds $F$ optimal given full delegation and the fictitious prior $\bar{F}$.\footnote{Interpreting the restriction $\bar{F}$ as a fictitious prior, we explain how our analysis relates to a result of \citet{curello2024comparativestaticspersuasion}. Their Proposition 5 implies that there is no shift of the prior $\bar{F}$ that makes $\BR(\bar{F})$ more informative (in the weak set order) for all $S$-shaped $G$. This result has no direct bearing on our analysis since we fix $G$. Moreover, as we will highlight in \Cref{sec:mic-characterization}, optimizing over restrictions $\bar{F}$ entails optimizing over a certain set of incomparable experiments.\label{footnote:curello_sinander}}

We define the set of incentive-compatible experiments as those that are best replies to some restriction.
\begin{definition}
An experiment $F\in\MPC(H)$ is \emph{incentive compatible} if there exists a restriction $\bar{F}\in\MPC(H)$ such that $F\in\BR(\bar{F})$. We denote the set of incentive-compatible experiments by $\mathcal{F}^{\IC}$.
\end{definition}

A useful observation is that an experiment $F$ is incentive compatible if and only if $F\in\BR(F)$; that is, when $F$ is the imposed restriction, the experimenter finds it optimal not to garble the restriction. 
To see why, suppose $F\in\BR(\bar{F})$ for some restriction $\bar{F}$, implying that $F$ is optimal for the experimenter among the admissible experiments $\MPC(\bar{F})$. 
Since $F\in \MPC(\bar{F})$, all mean-preserving contractions of $F$ are admissible too, that is, $\MPC(F)\subseteq \MPC(\bar{F})$. 
Thus, imposing $F$ as a restriction instead of $\bar{F}$ possibly shrinks the set of admissible experiments while $F$ remains admissible. Naturally, $F$ remains optimal among $\MPC(F)$, implying $F\in\BR(F)$.

The designer's utility from an experiment $F$ is given by 
\begin{equation*}
    U_{D}(F) = \int u_{D}(m) \de F(m).
\end{equation*}
An incentive-compatible experiment is \emph{optimal} if it maximizes the designer's utility across the set of incentive-compatible experiments $\mathcal{F}^{\IC}$.\footnote{This optimality notion assumes that the experimenter breaks ties in favor of the designer. An optimal incentive-compatible experiment exists; see \Cref{appendix:optimal_existence}.}

\begin{remark}[Beyond reduced-form restrictions]\label{remark:abstract_restrictions}
Our formulation captures general restrictions on information production or communication. To see this, suppose we model the experimenter's information explicitly through a signal $(\mathcal{S},\pi)$, where $\mathcal{S}$ is a space of messages and $\pi(\omega,z)$ maps the state $\omega$ and an independent randomization device $z$ into a message. The designer restricts the experimenter to a set of signals that is closed under garblings. To relate this to our reduced-form formulation, note that each signal induces a distribution $F$ of posterior means. Conversely, any mean-preserving contraction of $F$ can be induced by a signal that is a garbling of the signal that induces $F$.\footnote{We verify this fact in \Cref{appendix:ordered_experiments}.} Therefore, if $F$ is induced by the experimenter's preferred signal under the given restriction, then $F$ must be optimal among all its mean-preserving contractions, i.e. $F\in\BR(F)$, and is therefore incentive-compatible in our~reduced-form~sense.

For instance, rules of evidence may restrict prosecutors to particular forensic procedures---such as validated DNA, toxicology, or other certified testing methods---while excluding other methods or limiting the distinctions that admissible tests may generate. Such rules therefore place an upper bound on informativeness. The prosecutor may nevertheless use less information, for example by not commissioning an admissible test, not introducing its result at trial, or relying on a coarser admissible conclusion. Privacy constraints provide another natural example for restrictions. In \Cref{sec:privacy_application}, we show how prominent privacy notions map into our framework and can implement the restrictions characterized below.
\end{remark}

\section{Benchmark: Full Delegation}\label{sec:unrestricted_persuasion}
In this section, we discuss the benchmark in which the designer imposes no restriction---$\bar{F} = H$---and the experimenter chooses among all experiments.  
The results in \citet{kolotilin2022censorship} imply that upper censorship of the prior is the unique best reply of the experimenter: the experimenter reveals all states below a threshold $x$, and pools all remaining states to an atom $y=\mathbb{E}_H[\bm{\omega} \mid \bm{\omega} \geq x]$.
The dashed blue experiment $F$ depicted in \Cref{fig:ICDF_example} is an instance of an upper censorship experiment.

\begin{definition}\label{def:upper_censorship}
An experiment $F$ is \emph{upper censorship} (of the prior $H$) with threshold $x$ and atom $y$ if $F$ coincides with the prior $H$ on the interval $[0, x]$ and assigns mass $1 - H(x)$ to the point $y=\mathbb{E}_H[\bm{\omega} \mid \bm{\omega} \geq x]$.
\end{definition}

We restate the characterization from \citet{kolotilin2022censorship}.

\begin{proposition}[\citet{kolotilin2022censorship}]\label{prop:unrestricted_persuasion}
There is a unique best reply $F^\ast$ to $H$. In particular, $F^\ast$ is an upper censorship experiment with a threshold $x^\ast$ and an atom $y^\ast = \mathbb{E}_H[\bm{\omega} \mid \bm{\omega} \geq x^{\ast}]$ 
satisfying $0 < x^{\ast} < r_{0} < y^{\ast}< 1$ and
    \begin{equation}\label{eq:unrestricted_persuasion_FOC}
        G(x^{\ast}) + g(y^{\ast})(y^{\ast} - x^{\ast}) = G(y^{\ast}).
    \end{equation}
\end{proposition}

We recall the intuition from \citet{kolotilin2022censorship}.
The atom $y^{\ast}$ lies in the concave region of the experimenter's payoff $G$. Due to this concavity, the experimenter optimally pools some states around $y^{\ast}$ to $y^{\ast}$. As pooling additional states into $y^{\ast}$ raises the probability of the pooling outcome, the experimenter also adds some states from the convex part into the pooling region.
\Cref{eq:unrestricted_persuasion_FOC} is the first-order condition that characterizes the optimal pooling region, trading off the value of revelation in the convex region against the likelihood of the pooling outcome. \Cref{fig:price_FOC} depicts this optimality condition and illustrates the optimal upper censorship experiment.

\begin{figure}[ht]
    \centering
    \includegraphics[width=0.75\linewidth]{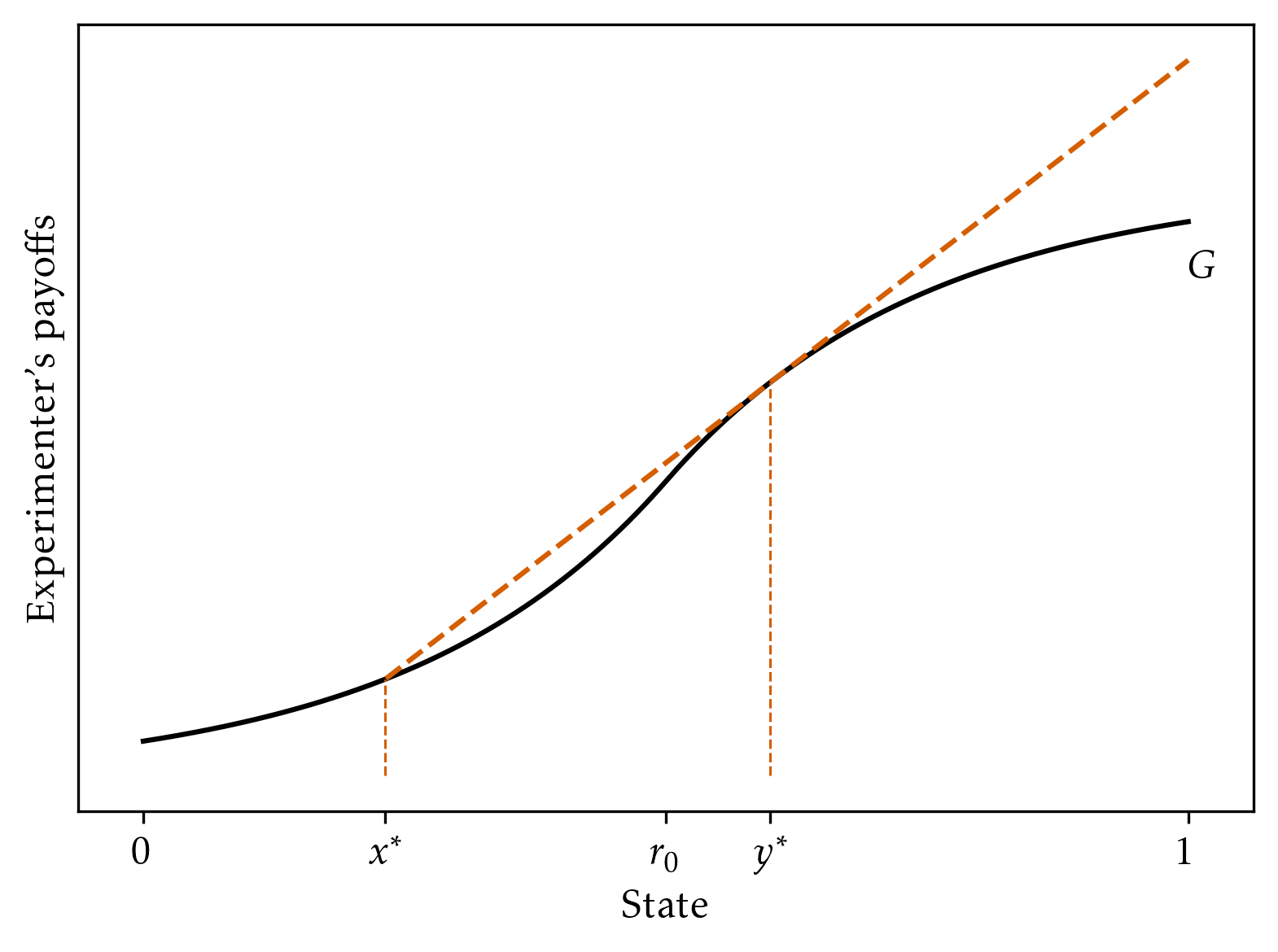}
    \caption{\emph{\Cref{eq:unrestricted_persuasion_FOC} defines the optimal upper censorship experiment, illustrated by the red dotted line. The tangent to $G$ at $y^{\ast}$ intersects $G$ at $x^{\ast}$. Since $x^{\ast} < r_{0}$, some states in the convex region of the experimenter's payoff are also pooled into the high atom $y^{\ast}$.}}
    \label{fig:price_FOC}
\end{figure}

Under \Cref{assumption:informativeness}, the threshold $x^{\ast}$ is interior: the experimenter reveals a non-trivial interval of states.
If \Cref{assumption:informativeness} fails, the uninformative experiment $\delta_{\mu}$ is the unique best reply to $H$.
When the designer can pick a restriction, \Cref{assumption:informativeness} remains necessary for information transmission since any restriction only shrinks the set of experiments while leaving $\delta_{\mu}$ admissible.

The full delegation outcome is the benchmark that the designer can guarantee by imposing no restriction.
But the resulting upper censorship pools an interval of states and, hence, substantial information is lost. We next study which alternative experiments the designer can induce via restrictions.

\section{Maximal Incentive-Compatible Experiments} \label{sec:mic-characterization}

In this section, we show that it suffices to consider a subset of incentive-compatible experiments that we call \emph{maximal}. Our first main result characterizes all maximal incentive-compatible experiments.

\subsection{Motivation and Definition}\label{sec:MIC:motivation_and_definition}
Incentive-compatible experiments represent all experiments that the designer can implement via a restriction. 
However, some incentive-compatible experiments are clearly suboptimal. 
Specifically, suppose an incentive-compatible $F$ is a proper mean-preserving contraction of another incentive-compatible experiment $\tilde F$.
The experimenter is better off under $\tilde F$ than under $F$, since $\tilde F$ represents a more lenient restriction.
The convex part of the designer's payoff is also higher---$\int \gamma \de \tilde F > \int \gamma \de F$---since $\tilde F$ is a proper mean-preserving spread of $F$.
Thus, $F$ cannot be optimal.
As a concrete example, the uninformative experiment $\delta_{\mu}$ is incentive-compatible since, when imposed as a restriction, the experimenter can only choose $\delta_{\mu}$.
However, the uninformative experiment $\delta_{\mu}$ is a proper mean-preserving contraction of the full delegation outcome $F^{\ast}$ and, thus, $\delta_{\mu}$ is not optimal.

This observation suggests to focus on those incentive-compatible experiments that are not mean-preserving contractions of other incentive-compatible experiments.
\begin{definition}
An incentive-compatible experiment $F$ is \emph{maximal} (MIC) if there does not exist an incentive-compatible experiment $\tilde{F}$ of which $F$ is a proper mean-preserving contraction; that is, there does not exist an incentive-compatible experiment $\tilde{F}$ such that $F\in\MPC(\tilde{F})$ and $\tilde{F}\notin \MPC(F)$.
\end{definition}
Economically, an MIC experiment represents a restriction that the experimenter has no incentive to garble, whereas the experimenter would garble every strictly more informative restriction.

The next lemma verifies that maximality is indeed without loss for optimality.
\begin{lemma}\label{lemma:MIC_wlog}
    For all incentive-compatible experiments $F$ there exists a maximal incentive-compatible experiment $\tilde{F}$ such that $U_{D}(\tilde F) \geq U_{D}(F)$.
\end{lemma}

The notion and optimality of MIC experiments highlight how the designer may benefit from imposing a restriction.
All MIC experiments are MPS-incomparable to one another, by definition of maximality.
Since MIC experiments are without loss, the designer thus trades-off informativeness about different regions of the state space when choosing among MIC experiments.

\subsection{Double Censorship}
In our main characterization result, we show that all MIC experiments take the form of \emph{double censorship} experiments.
A double-censorship experiment divides the state space into three intervals: a fully revealing interval, an intermediate pooling interval, and a high pooling interval. 
\Cref{fig:double_censorship_ICDF} depicts the ICDF of a double-censorship experiment.

\begin{definition}
An experiment $F$ is \emph{double censorship} with thresholds $(s, t)$ and atoms $(x, y)$ if $s \leq t$ and $x = \mathbb{E}_{H}[\bm{\omega}\mid \bm{\omega} \in [s, t]]$ and $y = \mathbb{E}_{H}[\bm{\omega}\mid \bm{\omega} \in [t, 1]]$ and
\begin{equation*}
    \forall m\in [0, 1],\quad
    F(m) = 
    \begin{cases}
        H(m)\quad&\mbox{if } m\in [0, s),\\
        H(s)\quad&\mbox{if } m\in [s, x),\\
        H(t)\quad&\mbox{if } m\in [x, y),\\
        1\quad&\mbox{if } m\in [y, 1].
    \end{cases}
\end{equation*}
\end{definition}
Each upper censorship experiment, including the full-delegation outcome $F^\ast$, is a special double censorship experiment with $s = x = t$, and such that all states below $x$ are fully revealed while all states above $x$ are pooled to $y=\mathbb{E}_H[\bm{\omega}\mid \bm{\omega} \in [x, 1]]$.

\begin{figure}[t!]
    \centering
    \includegraphics[width=0.75\linewidth]{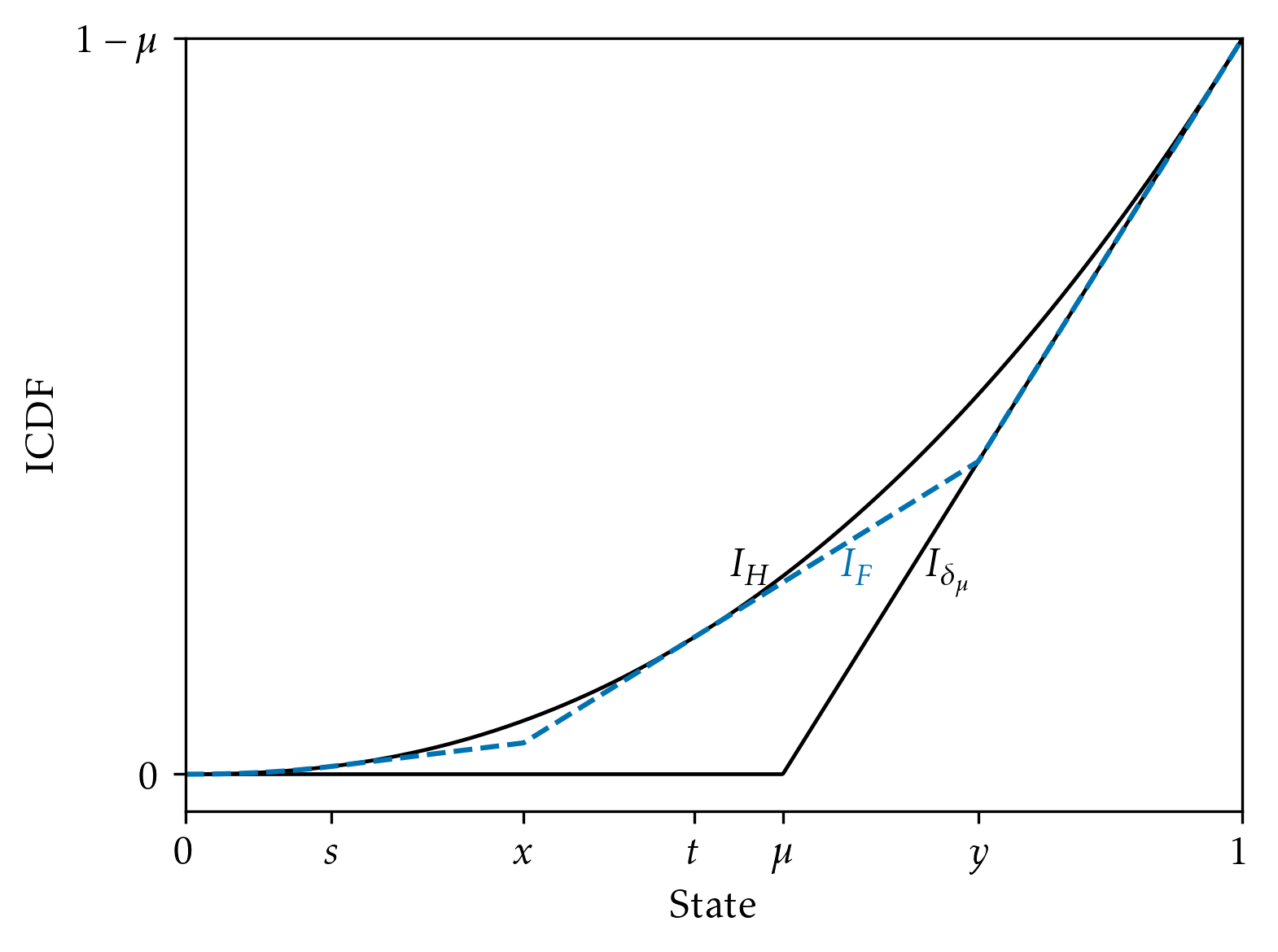}
    \caption{\emph{The ICDF of double censorship $F$ with thresholds $(s, t)$ and atoms $(x, y)$.}}
    \label{fig:double_censorship_ICDF}
\end{figure}

\subsection{Characterization}\label{sec:mic-characterization:characterization}

Our first main result shows that all MIC experiments take the form of double censorship. To do so, we introduce the set $P$, which relates the two atoms of a double censorship experiment to each other through a tangency condition similar to the one in \Cref{prop:unrestricted_persuasion} and depicted in \Cref{fig:price_FOC}. 
We will see that this tangency condition characterizes incentive-compatibility of MIC double censorship experiments.
Formally, let
\begin{equation}\label{eq:P_def}
    P = \lbrace (x, y)\colon 0 \leq x \leq r_{0} \leq y\leq 1 \mbox{ and } G(x) + g(y)(y - x) = G(y)\rbrace.
\end{equation}

\begin{theorem}\label{thm:MIC_characterization}
    An experiment $F$ is MIC if and only if $F$ is a double censorship experiment with thresholds $(s, t)$ and atoms $(x, y)$ satisfying
    \begin{equation*}
        (x, y) \in P \quad\mbox{and}\quad  s \leq x^{\ast} \leq t 
        \quad\mbox{and}\quad 0 < x < y < 1
        ,
    \end{equation*}
    where $x^\ast$ denotes the threshold of the upper censorship experiment under full delegation.
\end{theorem}
\Cref{thm:MIC_characterization} shows that double censorship experiments are the only candidates for MIC experiments. MIC experiments are characterized by interpretable properties of the thresholds $s$ and $t$, and the atoms $x$ and $y$.
The first property, $(x, y)\in P$, characterizes incentive-compatibility for double censorship experiments, as we explain in more detail in the discussion of the proof.
The second property, $s \leq x^{\ast} \leq t$, reflects maximality and highlights how a given MIC experiment $F$ compares to the full-delegation outcome $F^{\ast}$. 
On the one hand, the pooling region $[t, 1]$ is contained in the full-delegation pooling region $[x^{\ast}, 1]$.
On the other hand, there is a new pooling region $[s, t]$ vis-{\`a}-vis full revelation of states below $x^{\ast}$ under full delegation.
Finally, the inequalities $0 < x < y < 1$ assert that $x$ and $y$ are indeed distinct atoms and imply that $F$ is non-degenerate.

To gain intuition for the MIC characterization, consider the experimenter's trade-off when choosing an upper censorship threshold $x^{\ast}$ under the prior $H$.
This choice trades-off the location of the atom $y^{\ast} = \mathbb{E}_{H}[\bm{\omega}\mid\bm{\omega}\in [x^{\ast}, 1]]$ versus the probability assigned to the atom, $1 - H(x^{\ast})$.
This trade-off is captured by the first-order condition \eqref{eq:unrestricted_persuasion_FOC} that defines the set $P$.
If the designer imposes some restriction $\bar{F}$, the designer effectively imposes $\bar{F}$ as the prior, thereby affecting the trade-off between the location of the atom and the probability assigned to it.
To be more concrete, suppose $\bar{F}$ obtains from $H$ by pooling a small interval $(s, t]$ around $x^{\ast}$ to $x$, where $x$ lies \emph{below} $x^{\ast}$, i.e. $s < x < x^{\ast} < t < y^{\ast}$.
Suppose the experimenter would continue to reveal all states weakly below $x^{\ast}$ while pooling all states strictly above $x^{\ast}$, and when states are drawn from the fictitious prior $\bar{F}$. Since $\bar{F}$ pools states in $[x^{\ast}, t]$ to a point $x$ below $x^{\ast}$, the new atom lies above $y^{\ast}$---i.e. $y = \mathbb{E}_{\bar{F}}[\bm{\omega}\mid \bm{\omega} \in [x^{\ast}, 1]] > y^{\ast}$---, and the probability assigned to $y$ is below the probability assigned to $y^{\ast}$---i.e. $1 - \bar{F}(x^{\ast}) = 1 - H(t) < 1 - H(x^{\ast})$.
Therefore, from the perspective of the experimenter's trade-off, the atom $y$ is now too high and the probability assigned to $y$ is too low.
Therefore, the experimenter's best reply is to reduce the upper censorship threshold from $x^{\ast}$.
In fact, we can show that, if the points $x$ and $y$ are precisely such that $(x, y)\in P$, as in the experiments described by \Cref{thm:MIC_characterization}, then a new optimal upper censorship threshold is $x^{\ast}$ itself (and the resulting experiment is optimal among all experiments that are feasible given $\bar{F}$).
Given the construction of $\bar{F}$, revealing all states weakly below $x^{\ast}$ while pooling all states strictly above $x^{\ast}$ effectively means that, when states are drawn according to $H$, the experimenter reveals all states weakly below $s$, while pooling all states in the respective intervals $(s, t]$ and $(t, 1]$.
We have obtained a double censorship experiment.

\begin{remark}
    \Cref{thm:MIC_characterization} implies that MIC experiments are (particular) extreme points of the set of all experiments, $\MPC(H)$. This is no accident: we show that even without the $S$-shapedness assumption all MIC experiments are extreme points of $\MPC(H)$; see \Cref{appendix:extreme_points}.
\end{remark}

\subsection{Sketch of the Proof}
In the following, we outline the proof of \Cref{thm:MIC_characterization}, focusing on the characterization of an arbitrary MIC experiment as a double censorship experiment. 
The rest of the proof is found in the appendix.

First, we characterize incentive-compatible experiments in our environment building on the price-theoretic approach of \citet{dworczak2019simple}.
\begin{lemma}\label{lemma:price_function_IC}
    An experiment $F$ is IC if and only if there exists a convex function $p\colon [\min\supp F, \max\supp F]\to\mathbb{R}$ such that 
    $p \geq G$
    and $\supp F\subseteq\lbrace m\colon p(m) = G(m)\rbrace$.
\end{lemma}
\begin{proof}[Proof of \Cref{lemma:price_function_IC}]
\citet[Theorem 1 and Corollary 1]{dworczak2019simple} imply the following necessary and sufficient condition for an experiment $F$ to be optimal in an unconstrained persuasion problem when the prior is some distribution $\bar{F}$:
\emph{
Given a prior $\bar{F}$ and $F\in\MPC(\bar{F})$, it holds $F\in\BR(\bar{F})$ if and only if there exists a convex function $p\colon [\min\supp\bar{F},\max\supp\bar{F}]\to\mathbb{R}$ such that 
$p\geq G$
and $\supp F\subseteq\lbrace m\colon p(m) = G(m)\rbrace$ and $\int p(m)\de F(m) = \int p(m) \de \bar{F}(m)$.
}\footnote{In fact, their Corollary 1 shows that if $F\in\BR(\bar{F})$, then $p$ can be chosen to be continuous also at the boundary.}

Our lemma follows from this result and recalling that an experiment $F$ is incentive-compatible if and only if $F$ is a best reply to itself, that is, if $F\in \BR(F)$. Thus, $F$ would be optimal in the full-delegation problem if $F$ was the prior. Indeed, if $F = \bar{F}$, then the integral condition $\int p(m)\de F(m) = \int p(m) \de \bar{F}(m)$ trivially holds.\footnote{\Cref{lemma:price_function_IC} does not rely on $S$-shapedness of $G$. Continuity of $G$ suffices to apply the results of \citet{dworczak2019simple}. Specifically, their regularity condition (i) allows for interior discontinuities, but requires $G$ to be continuous at the boundary points of the state space. The proof of \Cref{lemma:price_function_IC} regards $F$ as a fictitious prior, so that the relevant boundary points are the minimum and maximum of the support of $F$. Hence, continuity of $G$ guarantees that Dworczak and Martini's results are applicable for any $F$.}
\end{proof}

To further characterize incentive-compatible experiments, we combine \Cref{lemma:price_function_IC} with $S$-shapedness of the experimenter's payoff.

\begin{lemma}\label{lemma:s_shaped_IC}
    A non-degenerate experiment $F$ is incentive-compatible if and only if there exist $(x, y)\in P$ such that $F$ assigns no mass to the set $(x, y)\cup (y, 1]$.
\end{lemma}

To gain intuition, suppose $F$ is an incentive-compatible experiment.
Then, $F$ cannot assign mass to multiple states in the concave part $[r_{0}, 1]$ of the experimenter's payoff $G$ since the experimenter would strictly prefer to pool such states.
Thus, there is at most one point $y$ in support of $F$ above $r_{0}$.
In fact, as discussed in the context of full delegation in \Cref{sec:unrestricted_persuasion}, the experimenter would also pool some states just below the inflection point $r_{0}$ to shift more probability mass to the pooling outcome.
The point $x < r_{0}$ represents the cutoff value where the experimenter finds it optimal to pool no further, as reflected in the condition $G(x) + g(y)(y - x) = G(y)$ that defines the set $P$ (recall  \eqref{eq:P_def}).
We conclude that $F$ cannot assign mass to $(x, y)\cup (y, 1]$.
Note, however, that we cannot yet conclude that, as in double censorship, $y$ is obtained by pooling the prior above some threshold, or that $x$ or $y$ are even in the support of $F$.

\Cref{lemma:s_shaped_IC} shows that incentive-compatibility is fully characterized by regions of the state space on which an experiment does not assign any mass. Thus, any MIC experiment should reveal as much information as possible while respecting these forbidden regions. We next show how double censorship emerges from combining incentive-compatibility with maximality considerations. To that end, we first prove that the set $P$, which characterizes the forbidden regions, is ordered by set inclusion.

\begin{lemma}\label{lemma:p_order}
        If $(x, y)\in P$ and $(x^{\prime}, y^{\prime})\in P$, then $[x, y] \subseteq [x^{\prime}, y^{\prime}]$ or $[x^{\prime}, y^{\prime}] \subseteq [x, y]$.
\end{lemma}
This lemma directly follows from the $S$-shapedness of $G$, as illustrated in \Cref{fig:p_order_illustration}. 
The condition $(x, y)\in P$ means that the tangent of $G$ at $y$ intersects $G$ at $x$. 
As $y\geq r_{0}$, the point $y$ lies in the concave part of $G$.
Now, consider a marginal increase in $y$. By concavity, the tangent of $G$ becomes flatter. Hence, the intersection $x$ with $G$ must occur at a lower value. It follows that the intervals $[x,y]$ defined by the elements of the set $P$ are nested.

\begin{figure}[t!]
    \centering
    \includegraphics[width=0.75\linewidth]{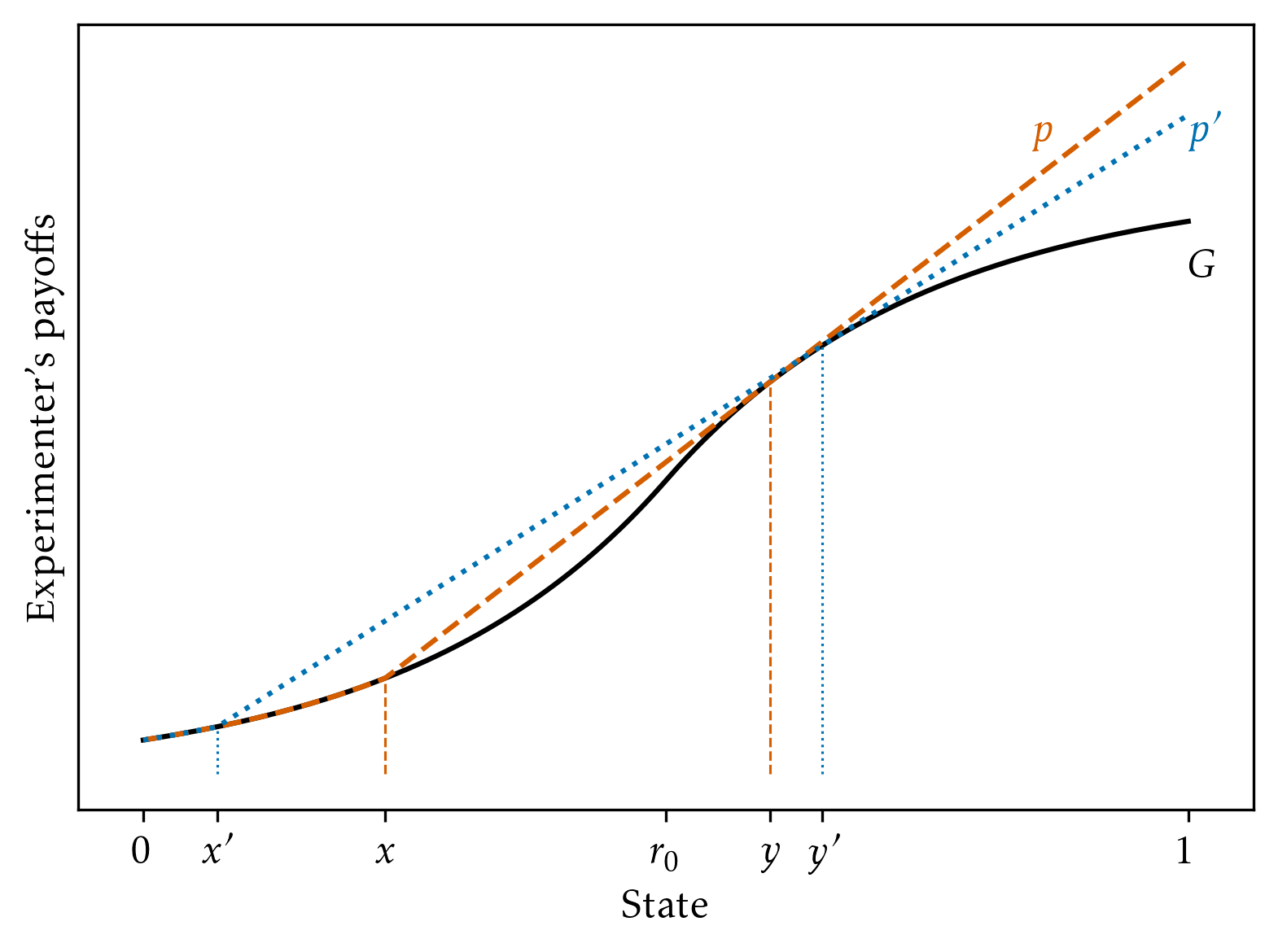}
    \caption{\emph{The order of the set $P$. For a higher value of $y$, the slope of $G$ at $y$ is lower, meaning the tangent to $G$ at $y$ intersects $G$ below $r_{0}$ at a lower point $x$.}}
    \label{fig:p_order_illustration}
\end{figure}

Next, we derive the remaining properties of the characterization of MIC experiments that result from maximality. Consider an arbitrary MIC experiment $F$. Incentive-compatibility implies, by \Cref{lemma:s_shaped_IC}, that $F$ does not allocate any mass to $(x,y)\cup(y,1]$, which implies that the ICDF $I_F$ of the experiment $F$ is affine on $[x,y]$ and $[y,1]$. Maximality implies then that $F$ is as informative as possible subject to not assigning any mass to the forbidden regions. 
Moreover, due to \Cref{lemma:p_order}, we only have to distinguish two cases for how any MIC experiment $F$ relates to the full-delegation outcome $F^\ast$: the experiment $F$ is such that $[x,y]\subseteq [x^\ast,y^\ast]$ or it is such that $[x^\ast,y^\ast]\subseteq[x,y]$.

First, suppose $[x, y] \subseteq [x^{\ast}, y^{\ast}]$. In this case, we show that $F$ must be an MPC of the full-delegation outcome $F^{\ast}$ and therefore not maximal unless $F = F^{\ast}$. 
We illustrate graphically in \Cref{fig:MIC_characterization_dominated_case}. As the experiments $F$ and $F^{\ast}$ assign no mass to $(y,1]$ and $(y^\ast,1]$, respectively, their ICDFs both hit the ICDF of the uninformative experiment $\delta_{\mu}$ at $y$ and $y^\ast$, respectively. In particular, since $y^\ast \geq y$, we find $I_F(y^\ast)=I_{F^\ast}(y^\ast)$. In this sense, the experiment $F$ pools weakly more states at the top than $F^\ast$. Moreover, $I_F$ lies weakly below $I_{F^\ast}$ on $[0,x^\ast]$ since the full-delegation $F^{\ast}$ outcome fully reveals all states in this region. Therefore, $F^\ast$ is weakly more informative at the bottom as well. Finally, note that $I_{F^\ast}$ is affine on $[x^\ast,y^\ast]$, while $I_F$ is convex. Since $I_F(x^\ast)\leq I_{F^\ast}(x^\ast)$ and $I_F(y^\ast)=I_{F^\ast}(y^\ast)$, we obtain $I_F(m)\leq I_{F^\ast}(m)$ for all $m \in [x^{\ast}, y^{\ast}]$.
Summarizing, we find that $I_{F}$ lies below $I_{F^{\ast}}$ on all of $[0, 1]$, implying that $F$ is an MPC of $F^\ast$.
The experiment $F^{\ast}$ is incentive-compatible.
Thus, for $F$ to be maximal, we must have $F=F^\ast$, implying that $F$ is upper censorship (a special case of double censorship).

\begin{figure}[t!]
     \centering
     \begin{subfigure}[t]{0.49\linewidth}
         \centering
         \includegraphics[width=\linewidth]{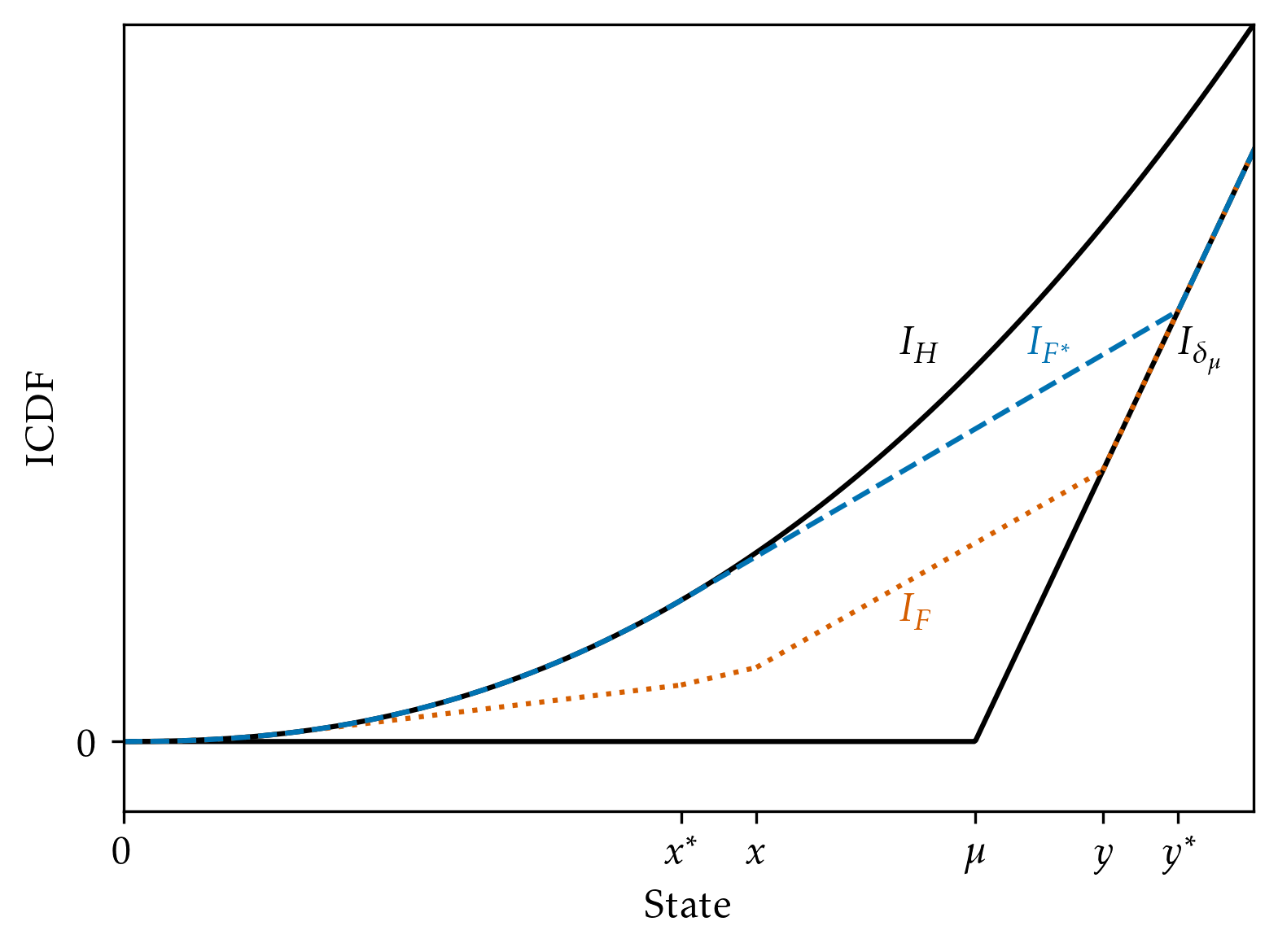}
         \caption{\emph{An incentive-compatible experiment $F$ that is an MPC of $F^{\ast}$.}}
         \label{fig:MIC_characterization_dominated_case}
     \end{subfigure}
     \hfill
     \begin{subfigure}[t]{0.49\linewidth}
         \centering
         \includegraphics[width=\linewidth]{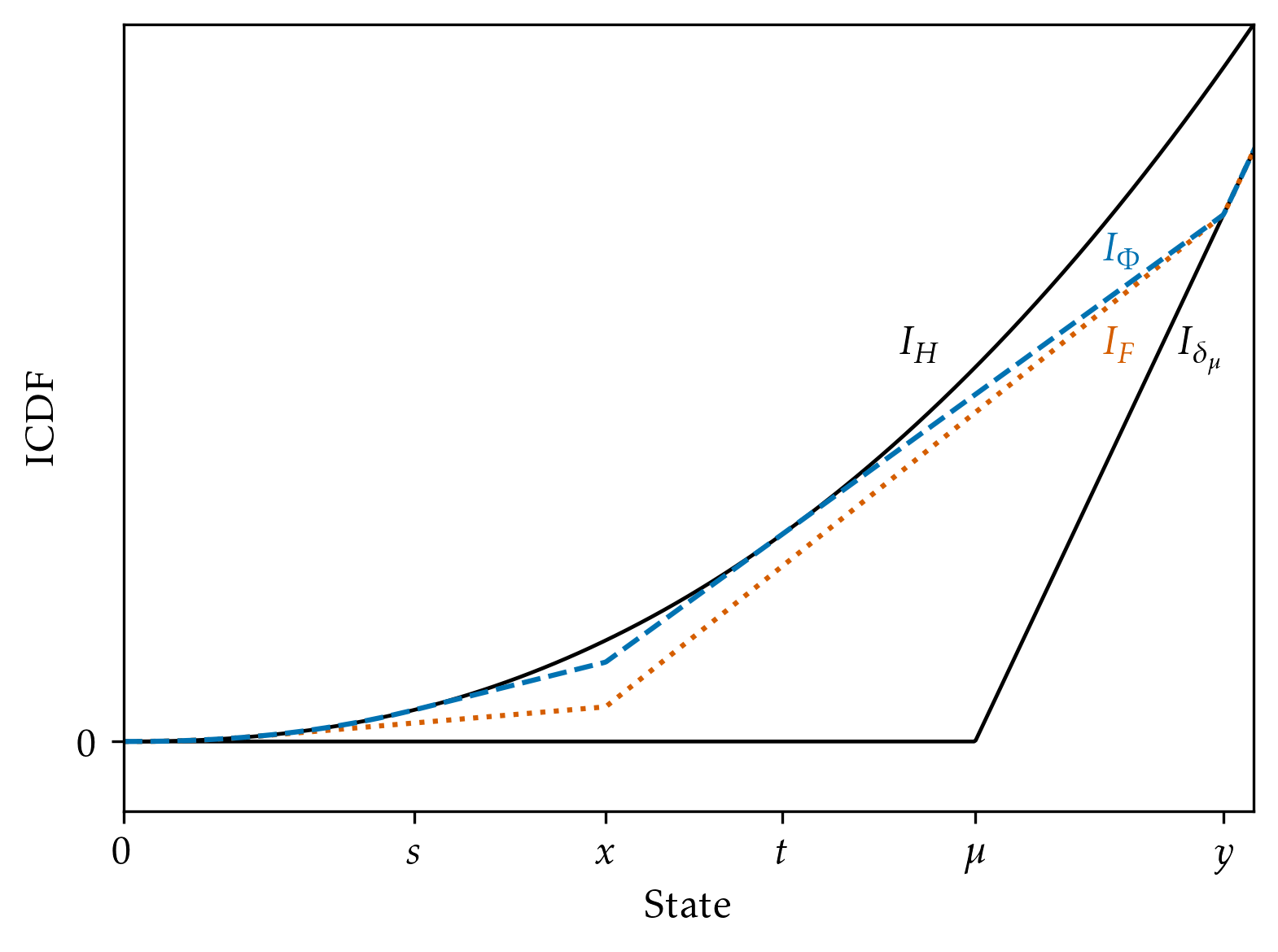}
         \caption{\emph{An incentive-compatible experiment $F$ and a double censorship experiment $\Phi$ that is an MPS of $F$.}}
         \label{fig:MIC_characterization_dominating_double_censorship}
     \end{subfigure}
     
     \caption{\emph{Characterization of MIC experiments. For readability, the horizontal axes are truncated at the top.}}
     \label{fig:combined_MIC_characterization}
\end{figure}

Second, suppose $[x^{\ast}, y^{\ast}] \subseteq [x, y]$. 
We illustrate graphically in \Cref{fig:MIC_characterization_dominating_double_censorship}.
Now, $I_F$ hits the ICDF of the uninformative experiment at $y\geq y^\ast$ and, hence, after $I_{F^\ast}$. It follows from incentive-compatibility that $I_F$ is affine on $[x,y]$ and coincides with $I_H$ at most once on this interval. For the purpose of this sketch, suppose $I_F$ lies strictly below $I_H$ on $[x,y]$. In this case, we can obtain a new experiment $\Phi$ by rotating the affine segment of $I_F$ on $[x,y]$ clockwise, while anchoring it at $I_{F}(y)$, until it is tangent to $I_H$ at a point $t$. Then, we can find a point $s$ such that the tangent to $I_{H}$ at $s$ intersects the new affine segment at $x$. Finally, we let $\Phi$ reveal all states below $s$. In the proof, we show that $[x^{\ast}, y^{\ast}] \subseteq [x, y]$ guarantees the existence of such points $s$ and $t$. The constructed $\Phi$ is incentive-compatible, a mean-preserving spread of $F$, and a double-censorship experiment. Incentive-compatibility follows from \Cref{lemma:s_shaped_IC} as $\Phi$ assigns no mass to the set $(x,y)\cup(y,1]$, while $\Phi\in\MPS(F)$ follows since $I_\Phi  \geq I_F$. Hence, for $F$ to be maximal, we must have $F=\Phi$, implying that $F$ is double censorship.

\subsection{Implementation}\label{sec:MIC_characterization:implementation}

Our MIC characterization concerns experiments that the experimenter has no incentive to garble.
We next observe that any MIC double censorship experiment can also be implemented via a simple restriction.
\begin{corollary}\label{cor:double_censorship_implementation}
    If $F$ is MIC double censorship with parameters $(s, t, x, y)$, then $F$ is a best reply to the restriction $\bar{F}$ that is obtained by pooling $H$ on $[s, t]$ to $x$ and otherwise revealing the state.
\end{corollary}
In \Cref{sec:privacy_application} we shall apply this simple implementation of MIC experiments to the design of optimal privacy constraints in the context of recommender systems.

\section{Optimal Delegation} \label{sec:optimal-restriction}
In this section, we first derive our central economic result: it is never optimal for the designer to fully delegate information provision to the experimenter. We then use the properties of optimal restrictions to derive comparative statics of the optimal policy with respect to the designer's payoff. We illustrate the economic content of the results in this section in \Cref{sec:privacy_application} when applying our model in the context of recommender systems and privacy constraints.

\subsection{Full delegation is not optimal}
We now revisit the designer's problem that is at the core of our environment. Anticipating the persuasion motive of the experimenter, can the designer gain from restricting the set of experiments available to the experimenter even though the designer places weight on an information-valuing component? 
To answer this, \Cref{lemma:MIC_wlog} implies that it suffices to consider restrictions that are MIC experiments, i.e. the experimenter has no incentive to garble the restriction but would garble any more informative restriction.
We now apply our characterization of MIC experiments as double censorship experiments.
The following theorem is our second main result. 

\begin{theorem}\label{thm:profitable_restriction}
        The full delegation outcome $F^{\ast}$ is not optimal.
\end{theorem}

We prove \Cref{thm:profitable_restriction} using a local perturbation argument at the full-delegation solution.
We construct an MIC double censorship experiment featuring a small pooling interval $[s, t]$ around the upper censorship threshold $x^{\ast}$, i.e. $s < x^{\ast} < t$.
The associated atom $y$ lies slightly above the upper censorship atom $y^{\ast}$.
We find that the marginal effect on the designer's profit from this perturbation equals:
\begin{equation}\label{eq:perturbation_derivative_text}
    \alpha \left(\gamma^{\prime}(y^{\ast}) - \frac{\gamma(y^{\ast}) - \gamma(x^{\ast})}{y^{\ast} - x^{\ast}}\right)(1 - H(x^{\ast}))>0,
\end{equation}
which is strictly positive since $\gamma$ is strictly convex and $x^{\ast} < y^{\ast}$.

To understand this marginal effect, recall that the designer maximizes a weighted sum of a convex part and the experimenter's payoff, i.e. $u_{D}(m) = \alpha \gamma(m) + (1 - \alpha) G(m)$.
The marginal effect on the experimenter's utility equals zero since we are perturbing the experiment $F^{\ast}$ that is optimal for the experimenter (envelope theorem).
It remains to understand the effect on the convex part of the designer's payoff.

Consider first the following hypothetical scenario: the designer is presented with a binary experiment with support $\lbrace x^{\ast}, y^{\ast}\rbrace$ and mean $\mu$, and can costlessly perturb the support subject to the mean constraint. 
The derivative \eqref{eq:perturbation_derivative_text} equals (up to a multiplicative factor) the change in the designer's payoff when increasing $y^{\ast}$ marginally while holding $x^{\ast}$ fixed and adjusting the probability of the realization $y^{\ast}$ to satisfy the mean constraint. 
The term $\gamma^{\prime}(y^{\ast})$ reflects the marginal increase in the value of the atom $y^{\ast}$, while $\frac{\gamma(y^{\ast}) - \gamma(x^{\ast})}{y^{\ast} - x^{\ast}}$ reflects the mass that is shifted from $y^{\ast}$ to $x^{\ast}$ to maintain the mean.
The designer gains from this perturbation since raising $y^{\ast}$ while fixing $x^{\ast}$ makes the binary experiment more informative.
In our actual model, the designer cannot freely pick the experiment but must additionally satisfy incentive compatibility. To achieve this, the designer introduces the double-censorship thresholds $s$ and $t$ and lets the lower atom $x$ adjust with $y$.
This local adjustment entails an informational loss around $x^{\ast}$. The proof shows, however, that this loss is of smaller order than the first-order gain in \eqref{eq:perturbation_derivative_text}. As $y$ approaches $y^{\ast}$, the thresholds $s$ and $t$ and the atom $x$ all converge to $x^{\ast}$, so the contribution of the newly created intermediate pool is $o(y-y^{\ast})$.
Thus, incentive compatibility can be restored locally without overturning the strict first-order gain generated by improving information in the high states.

\subsection{Comparative Statics}

We have seen that MIC experiments are Blackwell-incomparable and thus trade off informativeness across states.
In this section, we study this trade-off by developing comparative statics on optimal MIC experiments.

To describe the results, we first show that MIC double censorship experiments are totally ordered.
Let $F$ and $F'$ be two MIC double censorship experiments, with respective thresholds $s, t$ and $s', t'$.
Say $F$ is \emph{more disciplined} than $F'$, written $F\succeq F'$, if $s \leq s'$ and $t' \leq t$.
The name recalls the implementation \Cref{cor:double_censorship_implementation}: $F$ is implemented by pooling $[s, t]$, while $F'$ by pooling $[s', t']$. Thus, $F$ is implemented by imposing more discipline on the experimenter.\footnote{Indeed, if $\bar{F}$ and $\bar{F}'$ denote the restrictions that pool $H$ on $[s, t]$ and $[s', t']$ respectively, then $\bar{F}$ is an MPC of $\bar{F}'$. Consequently, all experiments that the experimenter can choose under restriction $\bar{F}$ can also be chosen under $\bar{F}'$.} The full delegation outcome $F^{\ast}$ is the least disciplined MIC experiment.

By definition, the more disciplined experiment $F$ features a larger bottom pool $[s, t]$, but a smaller top pool $[t, 1]$, compared to $F'$. Intuitively, $F$ provides better information about high states, but worse information about intermediate states. 

\begin{corollary}\label{cor:MIC_total_order}
    MIC double censorship experiments are totally ordered by discipline $\succeq$. The experimenter's utility decreases with respect to discipline.
\end{corollary}
Thus, the designer can shrink the upper pooling interval $[t, 1]$ only by expanding the lower interval region $[s, t]$. 
This trade-off is driven by the need to maintain the constraint $(x, y)\in P$ that reflects incentive compatibility.
The impact on the experimenter's utility is shown using the fact that a more disciplined experiment $F$ is implemented by pooling a larger intermediate interval $[s, t]$. 

For the first comparative statics result, we consider strictly convex designer-payoffs $u_{D}$. For smooth $u_{D, 1}$ and $u_{D, 2}$ with strictly positive second-order derivatives $u_{D, 1}''$ and $u_{D, 2}''$, say $u_{D, 2}$ is \emph{more information-sensitive at high states} than $u_{D, 1}$ if 
\begin{equation*}
    \frac{u_{D, 2}^{\prime\prime}(m)}{u_{D, 1}^{\prime\prime}(m)} \quad\mbox{is weakly increasing in } m.
\end{equation*}
The second derivative measures the value of a local mean-preserving spread. Thus, the condition says that, relative to $u_{D, 1}$, the payoff $u_{D, 2}$ places relatively more weight on information at higher posterior means.\footnote{Our terminology is borrowed directly from \citet{dworczak2026robusttrust}, who define this notion for beliefs about a binary state and in a different economic environment.}

As a leading example, suppose the designer's posterior-mean payoff is given by the DM's expected payoff, $u_{D, 1}(m) = \mathbb{E}_{G}[\lambda_{1}(\bm{r})\max\lbrace m - \bm{r}, 0\rbrace]$, where $\lambda_{1}$ is a strictly positive, type-dependent welfare weight. Here, $u_{D, 1}^{\prime\prime}(m) = \lambda_{1}(m) g(m)$, so $u_{D, 2}$ is more information-sensitive at high states if $\lambda_{2} / \lambda_{1}$ is weakly increasing.

\begin{proposition}\label{prop:convex_comparative_statics}
    Suppose the designer's objective $u_{D}$ is smooth with $u_{D}(m) > 0$ for all $m$.
    If $u_{D}$ becomes more information-sensitive at high states, then the set of optimal MIC double censorship experiments increases in the strong set order induced by discipline.
\end{proposition}
To gain intuition, suppose the designer maximizes the DM's utility, $u_{D, 1}(m) = \mathbb{E}_{\bm{r}\sim G}[\lambda_{1}(\bm{r})\max\lbrace m - \bm{r}, 0\rbrace]$. For each type $r$ of the DM, the decision-relevant information is whether the state is above $r$. Thus, high $r$-types benefit from information about high states. If the designer assigns these types relatively higher welfare weight ($\lambda_{2} / \lambda_{1}$ increases), the designer distorts the experiment to provide better information about high states, which, under MIC, means a more disciplined experiment.

Formally, we establish a single-crossing property.
Suppose that the designer prefers some $F$ to $F'$ under $u_{D}$, where $F \succ F'$.
Integrating by parts, one can show:
\begin{equation}\label{eq:atom_preferences}
    U_{D}(F) - U_{D}(F^{\prime}) = \int_{[0, 1]} \left(I_{F}(m) - I_{F'}(m)\right) u_{D}^{\prime\prime}(m) \de m \geq 0.
\end{equation}
We show that if $F$ is more disciplined than $F'$, then $I_{F} - I_{F'}$ is weakly single-crossing from below,\footnote{That is, there is $m_{0}$ such that $I_{F} - I_{F'}$ is weakly negative on the interval $[0, m_{0}]$ and weakly positive on $[m_{0}, 1]$.} which reflects that discipline trades information away at intermediate states and towards higher ones.
If $u_{D}$ becomes more sensitive to information at high states, then states where $I_{F} - I_{F'}$ is positive receive higher relative weight.
Consequently, the inequality $U_{D}(F) - U_{D}(F^{\prime}) \geq 0$ is preserved by making $u_{D}$ more sensitive to information at high states.\footnote{This is a simple instance of Variation Diminishing Property (\citet{karlin1956theory}).}
Thus, the designer continues to prefer $F$ to $F'$.

Our second result concerns the weight $\alpha$ that the designer assigns to the information-valuing component $\gamma$ instead of the experimenter's payoff $G$. Recall the designer's payoff
$u_{D}(m) = \alpha \gamma(m) + (1 - \alpha) G(m)$
for all $m$.
\begin{proposition}\label{prop:welfare_comparative_statics}
    If $\alpha$ increases, then the set of optimal MIC double-censorship experiments weakly increases in the strong-set order induced by discipline.
\end{proposition}
The proof is intuitive.
As mentioned, the experimenter's utility is lower under a more disciplined experiment. 
Therefore, if the designer prefers a more disciplined experiment, the contribution from the information-valuing component $\gamma$ must exceed the experimenter's loss. Raising the weight $\alpha$ on the information-valuing component preserves this comparison, meaning the designer continues to prefer the more disciplined experiment.
For $\alpha = 0$, the full delegation outcome $F^{\ast}$ is uniquely optimal.

\Cref{prop:convex_comparative_statics,prop:welfare_comparative_statics} identify changes to the environment that lead the designer to impose more discipline. In the next section, we interpret these comparative statics in the context of recommender systems.

\section{Application: Endogenous Privacy Constraints}\label{sec:privacy_application}

In this section, we show how our framework gives rise to endogenous privacy constraints in the context of recommender systems. We keep our model of recommender systems deliberately stylized to highlight the novel economic effects deriving from optimally delegated information provision.

Algorithms are becoming increasingly important in making recommendations to decision makers; for example, in suggesting products to consumers.
Recommendation algorithms often take as input user-specific data, such as their location, age, gender, or browsing history.
Policy discussions frequently raise privacy concerns regarding the use of such data, suggesting a tension between information transmission and privacy protection. Surprisingly, our results imply the opposite: a user who has no intrinsic preference for privacy may be strictly better off under optimally designed privacy constraints than without privacy constraints; the reason is protection against persuasion by an algorithm. \footnote{\citet{liang2026algorithm} consider the distinct but related issue of regulating algorithmic inputs to manage a trade-off between fairness and accuracy.}

In the following, we consider the privacy notion of \citet{strack2024privacy} and link it to restrictions in our model.\footnote{We can also accommodate the notion of differential privacy (\citet{dwork2006calibrating}). The key observation for both privacy notions is that they induce admissible sets of signals that are closed with respect to garblings. Hence, we can accommodate these sets as explained in \Cref{remark:abstract_restrictions}.\label{footnote:differential_privacy}}
Then, we provide a stylized model of a recommender system persuading a consumer to purchase a product; we show how the consumer can benefit from well-designed privacy constraints, even without intrinsic preferences for privacy. Finally, we illustrate how privacy constraints can endogenously arise as a form of self-regulation of a hierarchical organization with misaligned preferences over information provision and persuasion.

\subsection{Privacy Constraints}\label{ssec:privacy-constraints}

To formulate Strack and Yang's \citeyearpar{strack2024privacy} notion of privacy-preserving signals, we slightly extend our model.
There is a random variable $\bm\theta$ taking values in a measurable space $\Theta$, which is correlated with the payoff-relevant state $\bm\omega$. In particular, we assume that $\omega=v(\theta)$ with $v$ a strictly increasing function. We give an interpretation of this condition in \Cref{ssec:recommender}, in which we apply our discussion of privacy constraints to recommender systems. 
The variable $\bm\theta$ represents privacy-protected characteristics but is itself payoff-irrelevant.
A signal is a pair $(\mathcal{S}, \pi)$, where $\mathcal{S}$ is a measurable space and $\pi\colon [0, 1]\times\Theta\times [0, 1]\to \mathcal{S}$ maps the realizations of the state $\omega$, the privacy-protected characteristic $\theta$, and a randomization device $z$ to signals.
As before, an experiment refers to a distribution of posterior means about the state.
Here, we regard the experimenter as picking a signal rather than directly picking an experiment.

Let $\mathcal{P}$ be a family of events $P$ in $\Theta$ that is closed with respect to finite intersections. The sets in $\mathcal{P}$ are called \emph{privacy sets}. 
A signal $(\mathcal{S}, \pi)$ is \emph{privacy preserving} if for all $P\in\mathcal{P}$ it holds $\Pr(\bm{\theta} \in P\mid \pi) = \Pr(\bm{\theta} \in P)$ almost surely; that is, the signal cannot reveal information about any privacy set. Moreover, given some event $P_{0}$, a signal $(\mathcal{S}, \pi)$ is \emph{conditionally privacy-preserving}, if for all $P\in\mathcal{P}$ it holds $\Pr(\bm{\theta} \in P\mid \pi, P_0) = \Pr(\bm{\theta} \in P\mid P_0)$ almost surely; i.e. the privacy-constraint is required only conditional on the event $P_0$. 

\paragraph{Privacy-preserving signals as restrictions.} In the following, we show that the restrictions that we study are closely linked to privacy-preserving signals.

Fixing privacy sets $\mathcal{P}$, suppose that the experimenter can pick any privacy-preserving signal. Importantly, note that if a signal $(\mathcal{S}, \pi)$ is privacy-preserving, then all garblings of this signal are also privacy-preserving.\footnote{In fact, \citet{strack2024privacy} show that the experiments that can be induced via privacy-preserving signals correspond exactly to the MPCs of a certain experiment.} 
As explained in \Cref{remark:abstract_restrictions}, our reduced-form model captures such garbling-closed sets of restrictions on signals. That is, given $\mathcal{P}$, if the experimenter optimally picks a signal $(\mathcal{S}, \pi)$ that induces an experiment $F$, then this experiment $F$ is a best reply to itself, $F\in \BR(F)$, and is therefore incentive-compatible in our reduced-form sense.

\paragraph{Implementation via privacy constraints.}
Next, we argue that optimal restrictions can be implemented via privacy constraints under appropriate conditions on the relation between the payoff-relevant state $\bm{\omega}$ and the random variable $\bm{\theta}$ on which the privacy constraints are imposed. Since the restrictions in our model concern the payoff-relevant state $\bm{\omega}$, the implementation is easiest when $\bm{\omega}$ is strongly correlated with $\bm{\theta}$. Given the strictly monotonic mapping between the two random variables $\bm{\theta}$ and $\bm{\omega}$ that is known to the experimenter, imposing a privacy constraint on the individual-specific data $\theta$ is formally equivalent to imposing a privacy constraint on the estimated match payoff-relevant state $\omega$.

Recall from \Cref{cor:double_censorship_implementation} that any MIC double censorship $F$ with parameters $(s, t, x, y)$ is a best reply to the restriction $\bar{F}$ that is obtained by pooling $H$ on $[s, t]$ to $x$ and otherwise revealing the state. 
Using \Cref{cor:double_censorship_implementation}, we can implement $F$ via privacy sets and conditionally privacy-preserving signals as follows.
We choose as conditioning event $P_0$ the event that the payoff-relevant state $\omega$ falls into the interval $[s,t]$, and we let $\mathcal{P}$ be the family of all Borel subsets of $[s, t]$. Conditional on $\omega\notin [s, t]$, the privacy constraint holds trivially, since for all subsets $P$ of $[s, t]$ the probability of the event $\omega \in P$ equals $0$. Conversely, conditional on $\omega\in [s, t]$, the privacy constraint requires any signal to reveal no additional information beyond the state's falling into $[s,t]$. Therefore, every conditionally privacy-preserving signal induces an experiment that is a mean-preserving contraction of the experiment that reveals the state whenever the state is outside $[s, t]$, and otherwise only reveals that the state is in $[s, t]$. As before, it follows that the experimenter chooses the desired MIC experiment $F$ when constrained to use signals that are conditionally privacy-preserving for $\mathcal{P}$.

\subsection{Recommender Systems and Privacy Constraints}\label{ssec:recommender}
We now illustrate how consumers can strictly benefit from privacy constraints in the context of recommender systems, even in the absence of intrinsic preferences for privacy. A platform receives a commission for each sale and wishes to persuade consumers to purchase a product. As products are priced by sellers, the platform regards the price $p \in [0,1]$ as exogenous. The platform has access to consumer-specific data that is generated, for example, from the consumer's demographic characteristics and their browsing and purchase history. We assume that the consumer-specific information is the realization of a random variable $\bm{\theta}$ which is distributed according to a full-support distribution $T$ on $[0,1]$. For a consumer described by $\theta$, the platform can estimate her match value $\omega(\theta)$ from purchasing the product $\omega=v(\theta)$, where $v(\cdot)$ is a strictly increasing function.\footnote{Note that this assumes that the platform has sufficient information to predict the consumer's match value perfectly. It is straightforward to relax this assumption by assuming that $\omega=v(\theta)+\eta$ where $\eta$ is unobservable mean-zero noise orthogonal to the consumer type $\theta$.} The consumer has an idiosyncratic preference component $\varepsilon$ distributed according to a full-support distribution $E$, which she observes privately prior to her purchase decision. We assume that her value from not purchasing the object is zero.  Hence, the consumer purchases the product given a belief $m$ about her match value and price $p$ if
\begin{equation}
    m + \varepsilon -p \geq 0 \quad \Longleftrightarrow \quad m \geq p-\varepsilon.
\end{equation}
We can map this environment back into our framework by regarding $p-\varepsilon$ as the outside option. Then, the random outside option $r\equiv p-\varepsilon$ is distributed according to $G$, where $G$ is supported on $[0,1]$ and satisfies \Cref{assumption:S-shape} whenever $E$ has full support on $[p-1, p]$ and is $S$-shaped on that interval. The distribution $T$ of the random variable $\bm{\theta}$ describing the consumer data gives rise to the distribution $H$ of estimated match values when setting $T(t)=H(v(t))$ for $t \in [v^{-1}(0),v^{-1}(1)]$. 

\paragraph{Optimal recommendations.} An unregulated platform that maximizes its commissions will use the consumer data to persuade consumers to buy. Thus, it will implement an information policy as in \Cref{prop:unrestricted_persuasion}: send a simple \emph{buy} recommendation for the highest consumer types, and give a more nuanced recommendation to the lower types revealing the estimated match value for the consumer.

\paragraph{Optimal privacy constraints.} A regulator who is worried about consumers' being persuaded by the platform's recommendation system may want to impose a privacy constraint, limiting how the platform can use consumer-specific data. \Cref{thm:profitable_restriction} shows that a designer aligned with consumers always imposes a strict restriction to increase the consumer's payoff. One way to implement this restriction is to introduce a privacy constraint corresponding to the optimal restriction, as derived above. Hence, our framework reveals that such privacy constraints may arise endogenously not due to intrinsic privacy preferences but as a means to balance the tension between information provision and persuasion.

Our comparative statics show that the regulator optimally imposes more stringent privacy constraints when they assign relatively higher welfare weight to consumers with a high outside option $r$, i.e. those with low idiosyncratic preference $\varepsilon$, or when they assign relatively higher weight to consumers overall than the platform's revenue.

\section{Profitable Restrictions Beyond \textit{S}-shaped Persuasion}\label{sec:discussion}

In this section, we investigate whether the designer can benefit from imposing a restriction when the experimenter's payoff $G$ (i.e. the outside option distribution) is not necessarily $S$-shaped.
Equivalently, when is full delegation suboptimal?
We argue that two conditions are important for the profitability of restrictions and we formalize the discussion in \Cref{app:general}.
First, the experimenter has \emph{locally strict curvature preferences}: $G$ is locally strictly convex or strictly concave, and satisfies a technical smoothness condition.
Second, the full delegation outcome admits \emph{incomplete full revelation}: under full delegation, the experimenter fully reveals a non-trivial interval of states, but does not reveal all states.

Here, we provide two examples to show that full delegation may be optimal if these conditions fail.
In \Cref{app:general}, we show that full delegation is suboptimal if the conditions hold.
In this discussion, the definitions of incentive compatibility and optimality are as in the baseline model.
To isolate the extension on the experimenter’s side, let us assume that the designer's payoff $u_{D}$ is smooth and strictly convex.

\paragraph{Uninformed DM.}
Suppose the DM has no private information: the outside option distribution $G$ is a step function with step at $r_{0}$.
In particular, the experimenter does not have locally strict preferences, being indifferent between revealing and obfuscating the states within the subintervals below and above $r_{0}$, respectively.

In this example, the experimenter maximizes the probability of drawing a posterior mean above $r_{0}$.
If $\mu > r_{0}$, the experimenter can induce the DM to take action with probability one regardless of the restriction by choosing the experiment $\delta_{\mu}$; it is then easy to see that the designer cannot gain by imposing a restriction.
If $\mu \leq r_{0}$, then, for an arbitrary restriction $\bar{F}$ and best reply $F$, the largest point in the support of $F$ is at most $r_{0}$.
This implies that $F$ is an MPC of the experiment that reveals all states up to a threshold $s$ and pools the rest to $r_{0}$ (i.e. $r_{0} = \mathbb{E}_{H}[\bm{\omega}\mid \bm{\omega} \in [s, 1]]$, which is feasible since $\mu\leq r_{0}$).
This experiment, in turn, is a best reply to the prior $H$.
Thus, full delegation is optimal.
\citet[Corollary 1]{ichihashi2019limiting} makes the same observation.

\paragraph{$M$-shaped experimenter preferences.}
Suppose the experimenter's payoff $G$ is $M$-shaped as in \Cref{fig:M_shaped}:\footnote{While the depicted $G$ is not a CDF, transformations of the form $m \mapsto a G(m) + bm$ with $a > 0$ do not alter the experimenter's preferences over experiments since all experiments have mean $\mu$.}
concave on an interval of the lowest states, then convex, and then concave again.
In this example, $G$ has locally strict preferences.
We now discuss incomplete full revelation, which also involves the prior $H$ of the state.

\Cref{fig:M_shaped_no_restriction} depicts a situation in which the binary experiment $F^{\ast}$ that is supported on the peaks of the $M$ is a valid MPC of $H$.
Then, $F^{\ast}$ is also the experimenter's unique best reply under full delegation (apply Corollary 1 of \citet{dworczak2019simple} with the price function in \Cref{fig:M_shaped_no_restriction}). 
The full delegation outcome $F^{\ast}$ does not admit incomplete full revelation since it does not fully reveal the state on any non-trivial interval.
Let us argue why no profitable restriction exists.
\Cref{lemma:price_function_IC} implies that an arbitrary incentive-compatible experiment $F$ must be supported between the two peaks, i.e. $\supp F\subseteq [y_{1}, y_{2}]$.\footnote{Indeed, since $F^{\ast}$ is feasible, the prior mean $\mu$ lies between the two peaks. Thus, the support of $F$ cannot be entirely to the left of the left peak or to the right of the right peak. Since all price functions are convex, and since $F$ is supported on a set where a price function touches $G$, the support of $F$ lies between the peaks.} 
Among all such experiments, the binary experiment $F^{\ast}$ with support $\{y_1,y_2\}$ (and mean $\mu$) is the most informative.
Thus, full delegation is optimal for the designer.

Conversely, suppose the prior $H$ is such that the full delegation outcome $F^{\ast}$ reveals states in an interval $[x_{1}, x_{2}]$ in the valley between the peaks, and pools the remaining states to two points, $y_{1}$ and $y_{2}$, on the left and right slopes, as in \Cref{fig:M_shaped_prof_restriction}.
In particular, $F^{\ast}$ admits incomplete full revelation, and $F^{\ast}$ roughly looks like a lower censorship experiment and an upper censorship experiment patched together.
Following our perturbation argument for upper censorship experiments from the $S$-shaped case, the designer can strictly improve on full delegation by pooling a small interval of states around each of $x_{1}$ and $x_{2}$. 

\begin{figure}[t!]
    \centering
    \begin{subfigure}{0.48\textwidth}
        \includegraphics[width=\linewidth]{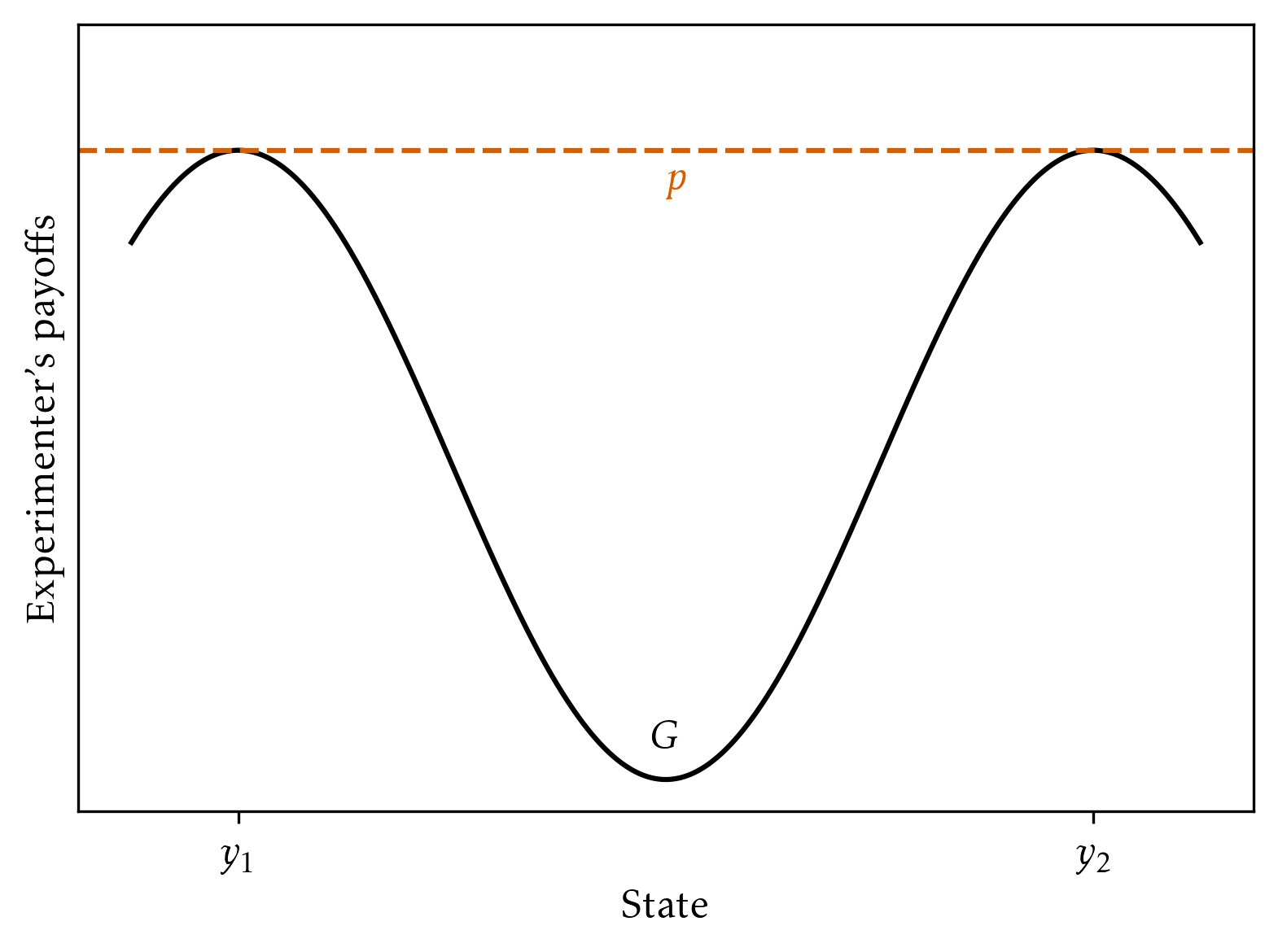}
        \caption{Full delegation is optimal.}
        \label{fig:M_shaped_no_restriction}
    \end{subfigure}
    \hfill
    \begin{subfigure}{0.48\textwidth}
        \includegraphics[width=\linewidth]{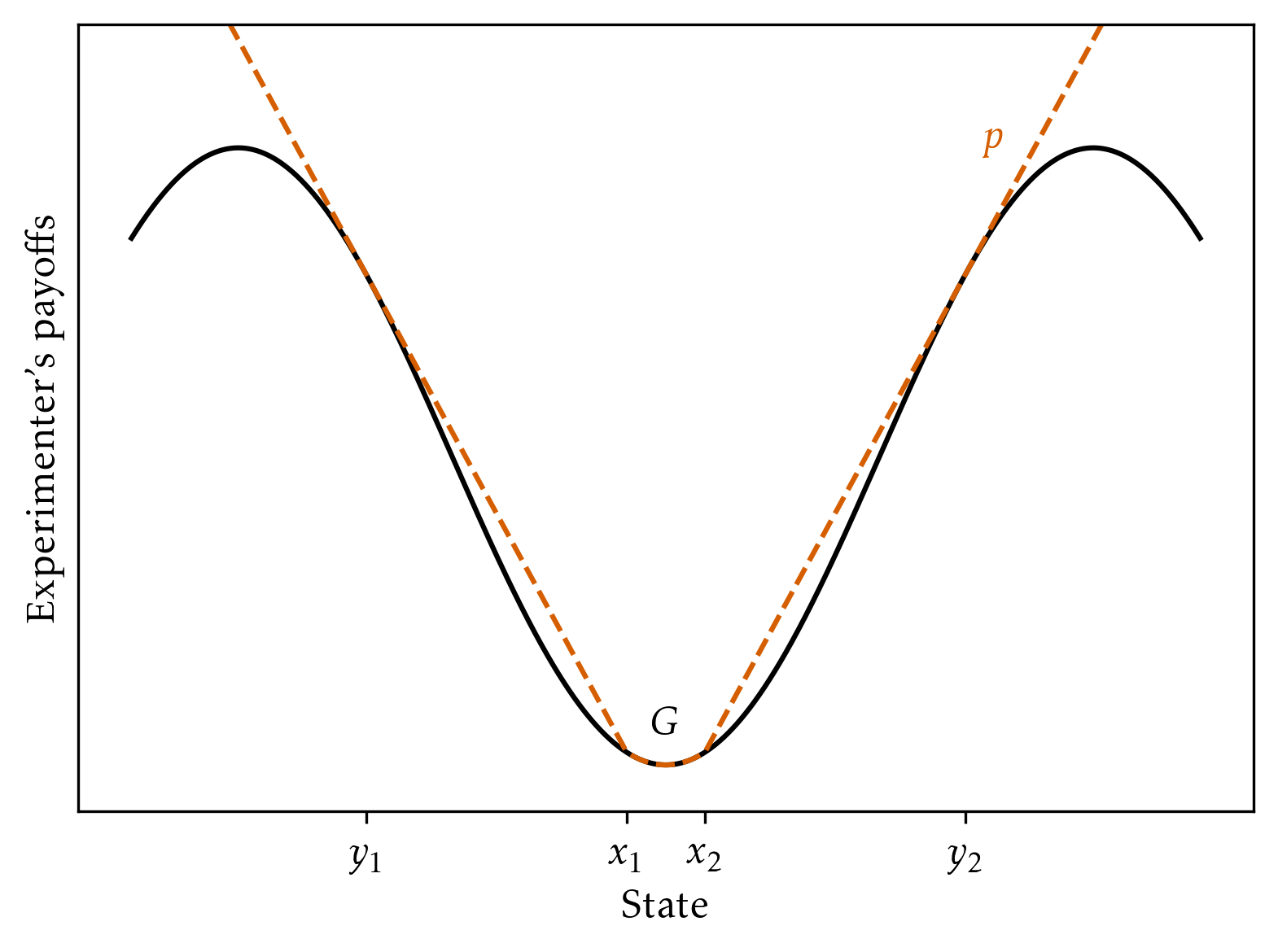}
        \caption{Full delegation is suboptimal.}
        \label{fig:M_shaped_prof_restriction}
    \end{subfigure}
    \caption{\emph{Examples of $M$-shaped delegation of information provision.}}
    \label{fig:M_shaped}
\end{figure}

The preceding examples suggest that profitable restrictions exist when the experimenter has locally strict preferences, and the full delegation outcome reveals some but not all states. 
We formalize these notions and the claim in \Cref{app:general}.

\section{Conclusion} \label{sec:conclusion}

We analyze delegated information provision when the party producing information, the experimenter, has persuasion incentives and can always garble any admissible experiment before it reaches a decision maker. The designer cannot impose experiments upon the experimenter, but exerts control by imposing a Blackwell-upper bound on informativeness.

The main conceptual contribution is to reduce optimal restriction design to the study of maximal incentive-compatible (MIC) experiments. An experiment is incentive compatible if it is self-enforcing: when the designer sets it as the upper bound, the experimenter optimally implements it rather than garbling. It is maximal if no strictly more informative experiment is also incentive compatible. We show that it is without loss to compare only MIC experiments.

In the baseline case of $S$-shaped experimenter preferences, MIC experiments are fully characterized as double censorship experiments. These reveal low states, pool an intermediate region to an interior atom, and pool high states to a top atom. Crucially, full delegation produces upper censorship (pooling all sufficiently high states), which is itself an extreme case of double censorship. Yet, we show that the designer always strictly prefers a nontrivial double-censorship restriction to full delegation.

The key takeaway is that the designer can protect the decision maker from excessive persuasion, while retaining the benefits of information provision. Optimal restrictions reallocate informativeness across states to discipline persuasion: sacrificing some precision in intermediate states relaxes the experimenter's incentives to censor high states, yielding a strict gain for the designer.

\appendix
\section{Proofs}\label{appendix}

\subsection{Proof of \headercref{Lemma}{{lemma:MIC_wlog}}}\label{appendix:mic_wlog}
Consider the problem of maximizing $\int \exp(m) \de\tilde{F}(m)$ across incentive-compatible experiments $\tilde{F}$ in $\MPC(H)\cap\MPS(F)$. Incentive-compatible experiments form a compact subset of $\MPC(H)$ when $G$ is upper semicontinuous (\Cref{appendix:optimal_existence}). In particular, the constraint set of the auxiliary problem is also compact. Hence, a solution exists. By construction, the solution is incentive-compatible and a mean-preserving spread of $F$. The solution is maximal since $\exp$ is strictly convex, and since any experiment that is a mean-preserving spread of $\tilde{F}$ is also a mean-preserving spread of $F$.
Since $\tilde F$ is incentive compatible and $F\in\MPC(\tilde F)$, the experimenter weakly prefers $\tilde F$ to $F$. Since $\gamma$ is convex, also $\int \gamma(m) \de (\tilde{F}(m) - F(m)) \geq 0$. Summarizing,
\begin{equation*}
    U_D(\tilde F)=\alpha\int\gamma\de\tilde F+(1-\alpha)\int G\de\tilde F
    \geq \alpha\int\gamma\de F+(1-\alpha)\int G\de F=U_D(F).
\end{equation*}
\qed

\subsection{Proof of \headercref{Lemma}{lemma:s_shaped_IC}}
    To prove one part of the equivalence, let $F\in\MPC(H)$ be IC and non-degenerate. 
    Let $p$ be as in \Cref{lemma:price_function_IC}.
    If $\max\supp F \leq r_{0}$, then, definitionally, $F$ assigns no mass to $(r_{0}, 1]$.
    The proof is complete for $x = y = r_{0}$, since $(r_{0}, r_{0}) \in P$.
    In what follows, let $\max\supp F > r_{0}$.
    Let $y = \max\supp F$, so that $p(y) = G(y)$.
    We find a point $x < r_{0}$ such that $(x, y)\in P$ and such that $F$ assigns no mass to $(x, y)\cup (y, 1]$.
    
    Recall that $G$ is strictly concave on $[r_{0}, 1]$.
    Since $p$ is convex and $p\geq G$ holds, there exists at most one point in $[r_{0}, 1]$ at which $p$ and $G$ coincide.
    Since $y > r_{0}$, this point is $y$.
    Thus, $\supp F\cap (r_{0}, 1]$ contains only $y$.
    Since $F$ is non-degenerate, the intersection $\supp F\cap [0, r_{0}]$ is non-empty.
    Let $\tilde{x} = \max(\supp F \cap [0, r_{0}])$.
    In particular, $p(\tilde{x}) = G(\tilde{x})$ and $\tilde{x} \leq r_{0}$.
   
    Now, recall $p\geq G$ and $p(y) = G(y)$, so that $p - G$ is minimized at $y$.
    Since $p$ is convex while $G$ is concave on a neighborhood of $y$, the subdifferential of $p$ at $y$ contains $g(y)$.
    Since also $p(\tilde{x}) = G(\tilde{x})$, we infer $G(\tilde{x}) + g(y)(y - \tilde{x}) = p(\tilde{x}) + g(y)(y - \tilde{x}) \geq p(y) = G(y)$, where the last inequality follows by definition of the subdifferential.
    We also know $G(r_{0}) + g(y)(y - r_{0}) < G(y)$ since $G$ is strictly concave on $[r_{0}, 1]$.
    Thus, $\tilde{x} < r_{0}$ and, by continuity, there exists $x\in [\tilde{x}, r_{0})$ such that $G(x) + g(y)(y - x) = G(y)$.
    That is, $(x, y) \in P$. 
    Since $\tilde{x} = \max(\supp F \cap [0, r_{0}]) \leq x$ and $\supp F\cap (r_{0}, 1] = \lbrace y\rbrace$, we conclude that $F$ assigns no mass to $(x, y)\cup (y, 1]$.

    To prove the other part of the equivalence, let $(x, y)\in P$ and $F \in \MPC(H)$ be such that $F$ is non-degenerate and assigns no mass to $(x, y)\cup (y, 1]$.
    In view of \Cref{lemma:price_function_IC}, we can show that $F$ is IC by finding a convex function $p$ such that $p\geq G$ and such that $p = G$ holds on $\supp F$.
    Let $p(m) = G(m)$ for all $m\in [0, x]$, and $p(m) = G(x) + g(y)(m - x)$ for all $m \in [x, 1]$.
    Note $p(y) = G(y)$ since $(x, y)\in P$.
    Clearly, $p = G$ holds on $\supp F$ since $F$ assigns no mass to  $(x, y)\cup (y, 1]$.
    We next argue $p\geq G$.
    On $[r_{0}, 1]$, the function $G$ is concave whereas $p$ is affine.
    Since $p(y) = G(y)$, it holds $p \geq G$ on $[r_{0}, 1]$.
    On $[x, r_{0}]$ (recall $x \leq r_{0}$ from the definition of $P$), the function $G$ is convex whereas $p$ is affine.
    Since $p(x) = G(x)$ but $p(r_{0}) \geq G(r_{0})$ (as already argued), we find $p\geq G$ on $[x, r_{0}]$.
    It remains to argue that $p$ is convex.
    Clearly, $p$ is convex on $[0, x]$.
    Since $p$ is affine on $[x, 1]$ with slope $g(y)$, global convexity follows if we can show $g(y) \geq g(x)$.
    This inequality, in turn, holds since $p\geq G$ and $p(x) = G(x)$ both hold.
    \qed

\subsection{Proof of \headercref{Lemma}{{lemma:p_order}}}
    For all $(x, y)$ such that $0\leq x \leq r_{0} \leq y\leq 1$, let $\rho(x, y) = G(x) + g(y)(y - x) - G(y)$.
    By definition of $P$, if $(x, y)\in P$, then $\rho(x, y) = 0$.
    We show that $\rho$ is strictly single-crossing from above in each argument (holding the other fixed).
    Therefore, if $(x, y)$ and $(x^{\prime}, y^{\prime})$ are both in $P$, then $[x, y]$ and $ [x^{\prime}, y^{\prime}]$ are ordered by set inclusion.

    Fix $x$ such that $0 \leq x\leq r_{0}$.
    For all $y_{1}, y_{2}$ such that $r_{0} \leq y_{1} < y_{2}\leq 1$, it holds $\rho(x, y_{2}) - \rho(x, y_{1}) = (g(y_{2}) - g(y_{1}))(y_{1} - x) + \int_{y_{1}}^{y_{2}} (g(y_{2}) - g(z)) \de z$. Recall that $g$ is strictly decreasing on $[r_{0}, 1]$. Thus, the first term is weakly negative and the second is strictly negative. Thus, $\rho(x, \cdot)$ is strictly decreasing on $[r_{0}, 1]$. In particular, it is strictly single-crossing from above.

    For all $(x, y)$ such that $0\leq x \leq r_{0} \leq y\leq 1$, it holds $\partial \rho(x, y) / \partial x = g(x) - g(y)$.
    We claim $g(x) < g(y)$ if $\rho(x, y) = 0$ and $x < y$, which suffices to establish that $\tilde{x}\mapsto \rho(\tilde{x}, y)$ (defined on $[0, r_{0}]$) is strictly single-crossing from above.
    Towards a contradiction, let $g(x) \geq g(y)$. Since $g$ is strictly quasiconcave, all $r\in(x, y)$ satisfy $g(r) > g(y)$, implying $G(y) - G(x) > g(y)(y - x)$, which contradicts $\rho(x, y) = 0$.
    \qed

\subsection{Proof of \headercref{Theorem}{{thm:MIC_characterization}}}
    For reference, we record the ICDF of a double censorship experiment $F$ with thresholds $(s, t)$ and atoms $(x, y)$:
    \begin{equation*}
        \forall m\in[0, 1],\quad
        I_{F}(m)
        =
        \begin{cases}
            I_{H}(m)\quad&\mbox{if }m\in [0, s],\\
            I_{H}(s) + H(s)(m - s)\quad&\mbox{if }m\in (s, x],\\
            I_{H}(t) + H(t)(m - t)\quad&\mbox{if }m\in (x, y],\\
            m - \mu,\quad&\mbox{if }m\in (y, 1].
        \end{cases}
    \end{equation*}

    \Cref{lemma:MIC_characterization:necessity,lemma:MIC_characterization:sufficiency} below together prove \Cref{thm:MIC_characterization}.
    
    \begin{lemma}\label{lemma:MIC_characterization:necessity}
        If $F$ is an MIC experiment, then $F$ is non-degenerate double censorship for a pair of thresholds $(s, t)$ and atoms $(x, y)$ such that $(x, y)\in P$ and
        $s \leq x^{\ast} \leq t$ and $0 < x < y < 1$.
    \end{lemma}
    
    \begin{proof}[Proof of \Cref{lemma:MIC_characterization:necessity}]
        The experiment must be non-degenerate since, otherwise, it is a strict MPC of $F^{\ast}$, contradicting maximality.
        Hence, we can appeal to \Cref{lemma:s_shaped_IC} to find $(x, y)\in P$ such that $F$ assigns no mass to $(x, y)\cup (y, 1]$.

        We next claim $[x^{\ast}, y^{\ast}] \subseteq [x, y]$.
        By \Cref{lemma:p_order}, it holds $[x^{\ast}, y^{\ast}] \subseteq [x, y]$ or $[x, y] \subseteq [x^{\ast}, y^{\ast}]$.
        In the first case, we are done, so let $[x, y] \subseteq [x^{\ast}, y^{\ast}]$.
        We show that $F$ is an MPC of $F^{\ast}$; by maximality, it follows $F = F^{\ast}$ and, hence, $[x, y] = [x^{\ast}, y^{\ast}]$.
        To show $F$ is an MPC of $F^{\ast}$, we show $I_{F}\leq I_{F^{\ast}}$ on $[0, 1]$.
        Clearly, $I_{F} \leq I_{H} = I_{F^{\ast}}$ on $[0, x^{\ast}]$.
        Next, we know $I_{F}$ coincides with $I_{\delta_{\mu}}$ on $[y, 1]$, whereas $I_{F^{\ast}}$ coincides with $I_{\delta_{\mu}}$ on $[y^{\ast}, 1]$.
        Since $y \leq y^{\ast}$, we infer $I_{F} = I_{F^{\ast}}$ on $[y^{\ast}, 1]$.
        In summary, $I_{F}\leq I_{F^{\ast}}$ on $[0, x^{\ast}]\cup [y^{\ast}, 1]$.
        Finally, since $F^{\ast}$ assigns no mass to $(x^{\ast}, y^{\ast})$, the integrated CDF $I_{F^{\ast}}$ is affine on $[x^{\ast}, y^{\ast}]$.
        Since $I_{F}$ is convex and lies below $I_{F^{\ast}}$ at the endpoints of this interval (as already argued), it follows $I_{F}\leq I_{F^{\ast}}$ also on $[x^{\ast}, y^{\ast}]$.
        Thus, $I_{F}\leq I_{F^{\ast}}$ on $[0, 1]$.
        
        We next show that $F$ is double censorship.
        The idea is to construct an MPS $\tilde{F}$ of $F$ that is double censorship and IC.
        Since $F$ is MIC, it will follow $\tilde{F} = F$.

        Our candidate double censorship will use $(x, y)$ as atoms.
        We thus show next that $(x, y)$ can arise as atoms for some choice of thresholds $(s, t)$.
        First, since $y^{\ast} = \mathbb{E}_{H}[\bm{\omega}\mid \bm{\omega} \in [x^{\ast}, 1]]$ and $y^{\ast} \leq y \leq 1$, there exists $t \in [x^{\ast}, y]$ such that $y = \mathbb{E}_{H}[\bm{\omega}\mid \bm{\omega} \in [t, 1]]$.
        Note, $t \in [x, y]$ since $x\leq x^{\ast}$.
        Second, we argue there exists $s \in [0, x^{\ast}]$ such that $x = \mathbb{E}_{H}[\bm{\omega}\mid \bm{\omega} \in [s, t]]$.
        By construction, it holds $y = \mathbb{E}_{H}[\bm{\omega}\mid \bm{\omega} \in [t, 1]]$, which rearranges to $I_{F}(y) = I_{H}(t) + H(t)(y - t)$.
        The affine map $m\mapsto I_{H}(t) + H(t)(m - t)$ lies above $I_{F}$ at $m = t$ (since $I_{F}(t) \leq I_{H}(t)$) and coincides with $I_{F}$ at $m = y$ (by construction of $t$).
        Since also $I_{F}$ is affine on $[x, y]$ (since $F$ assigns no mass to $(x, y)$), and since two affine maps on $\mathbb{R}$ either coincide or intersect at most once, we conclude $I_{F}(x) \leq I_{H}(t) + H(t)(x - t)$.
        Let $\xi = I_{H}(t) + H(t)(x - t)$.
        In particular, $\xi \geq I_{F}(x) \geq 0$.
        We also have $\xi \leq I_{H}(x)$ since the affine map $\tilde{t}\mapsto I_{H}(t) + H(t)(\tilde{t} - t)$ is tangent to the convex map $I_{H}$ at $t$.
        Finally, find $s \in [0, x]$ such that $I_{H}(s) + H(s)(x - s) = \xi$; this is possible since $\xi \in [0, I_{H}(x)]$.
        Thus, $I_{H}(s) + H(s)(x - s) = I_{H}(t) + H(t)(x - t)$, which rearranges to $x = \mathbb{E}_{H}[\bm{\omega}\mid \bm{\omega} \in [s, t]]$.

        Now let $\Phi$ be double censorship with thresholds $(s, t)$ and atoms $(x, y)$.
        Notice that $\Phi$ is IC; indeed, $\Phi$ assigns no mass to the set $(x, y)\cup (y, 1]$ (by definition of double censorship), and it holds $(x, y)\in P$ (as argued at the very beginning of the proof); now invoke \Cref{lemma:s_shaped_IC}.

        We now argue $\Phi$ is an MPS of $F$.
        Clearly, $I_{\Phi} = I_{H} \geq I_{F}$ on $[0, s]$, and $I_{\Phi} = I_{\delta_{\mu}} = I_{F}$ on $[y, 1]$.
        For $\tilde{t}\in [x, y]$ it holds $I_{\Phi}(\tilde{t}) = I_{H}(t) + H(t)(\tilde{t} - t) \geq I_{F}(\tilde{t})$, where the inequality follows from the argument that constructed $s$.
        Finally, it holds $I_{\Phi} \geq I_{F}$ on $[s, x]$ since $I_{\Phi}$ is affine on this interval, $I_{F}$ is convex, and since $I_{\Phi} \geq I_{F}$ holds at the endpoints of this interval (as already argued).
        Thus, $\Phi$ is an MPS of $F$.

        It remains to show $s \leq x^{\ast} \leq t$ and $0 < x < y < 1$.
        The inequalities $s \leq x^{\ast} \leq t$ follow from $[x^{\ast}, y^{\ast}]\subseteq [x, y]$ and the expressions for $x$, $y$, and $y^{\ast}$ as conditional expectations. From $[x^{\ast}, y^{\ast}]\subseteq [x, y]$ and $x^{\ast} < r_{0} < y^{\ast}$, we also obtain $x < y$.
        
        Next, towards a contradiction, let $y = 1$.
        Then also $t = 1$ (since $y = \mathbb{E}_{H}[\bm{\omega}\mid\bm{\omega}\in[t, 1]]$), meaning $x$ is the largest point in the support of $F$.
        From $y = 1 > y^{\ast}$ we also get $x < x^{\ast}$ (\Cref{lemma:p_order}).
        It now follows easily that $F$ is a proper MPC of $F^{\ast}$; indeed, $I_{F} = I_{\delta_{\mu}} \leq I_{F^{\ast}}$ on $[x, 1]$; this inequality is strict on $(x, x^{\ast})$ since $I_{F}$ is affine on this interval whereas $I_{F^{\ast}} $ is strictly convex (it coincides with $I_{H}$); finally, $I_{F} \leq I_{F^{\ast}} \leq I_{H}$ on $[0, x]$.
        Thus, $F$ is a proper MPC of $F^{\ast}$, contradicting maximality of $F$.
        Thus, $y < 1$.

        Finally, towards a contradiction, let $x = 0$. 
        This requires $s = t = 0$ since $x = \mathbb{E}_{H}[\bm{\omega}\mid \bm{\omega} \in [s, t]]$.
        It follows that $F$ is degenerate on $y$; contradiction.
    \end{proof}
    
    \begin{lemma}\label{lemma:MIC_characterization:sufficiency}
        If $F$ is double censorship for a pair of thresholds $(s, t)$ and atoms $(x, y)$ such that $(x, y)\in P$ and 
        $s \leq x^{\ast} \leq t$ and $0 < x < y < 1$,
        then $F$ is MIC.
    \end{lemma}
    \begin{proof}[Proof of \Cref{lemma:MIC_characterization:sufficiency}]
        \Cref{lemma:s_shaped_IC} implies that $F$ is IC.
        To prove that $F$ is maximal, let $\Phi$ be IC and such that $I_{F} \leq I_{\Phi}$.
        We show $I_{F} = I_{\Phi}$.
        Let $\hat{y} = \max\supp\Phi$.
    
        First, we show $\hat{y} < y$ cannot hold.
        Indeed, the function $I_{\Phi}$ coincides with $m\mapsto m - \mu$ on $[\hat{y}, 1]$.
        Meanwhile, $I_{F}$ coincides with $m\mapsto m - \mu$ on $[y, 1]$ but has slope of at most $H(t)$ on $[0, y)$.
        Since $t\leq y < 1$, we have $H(t) < 1$.
        Thus, if $\hat{y} < y$, then $I_{F}$ lies strictly above $I_{\Phi}$ on $(\hat{y}, y)$, contradicting $I_{F}\leq I_{\Phi}$.
        
        Thus, let $\hat{y}\geq y$.
        Since $\hat{y} = \max\supp\Phi$ and $y > r_{0}$, the IC characterization (\Cref{lemma:s_shaped_IC}) implies there exists $\hat{x} \leq r_{0}$ such that $(\hat{x}, \hat{y}) \in P$ and such that $\Phi$ assigns no mass to $(\hat{x}, \hat{y})\cup (\hat{y}, 1]$.
        Since $\hat{y}\geq y$, \Cref{lemma:p_order} implies $\hat{x} \leq x$.
        
        Since $\Phi$ assigns no mass to $(\hat{x}, \hat{y})$, and since $[x, y]\subseteq [\hat{x}, \hat{y}]$ and $t\in [x, y]$, we conclude that $I_{\Phi}$ is affine on $[\hat{x}, \hat{y}]$ and this interval contains $t$.
        Now, since $F$ is double censorship, $I_{F}$ is tangential to $I_{H}$ at $t$.
        Thus, also $I_{\Phi}(t) = I_{H}(t)$ since $I_{F}\leq I_{\Phi} \leq I_{H}$.
        Thus, $I_{\Phi}$ is affine with slope $H(t)$ on $[\hat{x}, \hat{y}]$.
        Since $I_{F}$ is affine on $[x, y]$ (a subset of $[\hat{x}, \hat{y}]$) and tangential to $I_{H}$ at $t$, it follows that $I_{\Phi}$ and $I_{F}$ coincide on $[x, y]$.
        From here and the inclusion $[x, y]\subseteq [\hat{x}, \hat{y}]$, it readily follows $I_{F} = I_{\Phi} = I_{\delta_{\mu}}$ on $[y, 1]$.

        We also know that $I_{F}$ and $I_{\Phi}$ coincide on $[0, s]$ since on this interval $I_{F} = I_{H} \geq I_{\Phi} \geq I_{F}$.
        It now follows $I_{F} = I_{\Phi}$ on $[s, x]$ since $I_{F}$ is affine on this interval, $I_{\Phi}$ is convex, and since the two coincide at the endpoints (as already argued).
        Thus, $I_{F} = I_{\Phi}$.
    \end{proof}

\subsection{Proof of \headercref{Corollary}{{cor:double_censorship_implementation}}}
    We use Theorem 1 of \citet{dworczak2019simple}.
    Let $p$ be the function that coincides with $G$ on $[0, x]$, and coincides with $m\mapsto G(x) + g(y)(m - x)$ for $m\in (x, 1]$.
    Then, $p$ is convex, weakly above $G$, and coincides with $G$ on $[0, x]$ and $y$; see the proof of \Cref{lemma:s_shaped_IC} for this argument.
    Moreover, $\supp F \subseteq \lbrace m\in [0, 1]\colon p(m) = G(m)\rbrace$ since $F$ is double censorship with parameters $(s, t, x, y)$.
    Finally, it holds $\int p \de F = \int p \de \bar{F}$.
    Indeed, we have $F = \bar{F} = H$ on $[0, s]$.
    Conditioned on the interval $[s, t]$, both $F$ and $\bar{F}$ are degenerate on $x$.
    Finally, on the interval $[t, 1]$, the function $p$ is affine and $F$ is obtained by pooling $\bar{F}$.
    Thus, $\int p \de F = \int p \de \bar{F}$.
    By Theorem 1 of \citet{dworczak2019simple}, the experiment $F$ is a best reply to $\bar{F}$.
    \qed
    
\subsection{Proof of \headercref{Theorem}{{thm:profitable_restriction}}}
    The idea is to consider MIC double censorship, parametrized by $(s, t, x, y)$, such that $y$ is close to the full delegation atom $y^{\ast}$; in this case, also $s$, $t$, and $x$ are all close to $x^{\ast}$.
    We evaluate the marginal gain from this perturbation.

    \Cref{lemma:s_shaped_IC} implies that a tuple $(s, t, x, y)$ with $0\leq s \leq x \leq t\leq y < 1$ and $x \leq r_{0} \leq y$ represents an IC double censorship if
    \begin{subequations}
    \begin{align}
        G(x) + g(y)(y - x) - G(y) &= 0,
        \label{eq:MIC_parametrization:1}
        \\ \int_{[t, 1]} (\omega - y) \de H(\omega) &= 0,
        \label{eq:MIC_parametrization:2}
        \\ \int_{[s, t]} (\omega - x) \de H(\omega) &= 0,
        \label{eq:MIC_parametrization:3}
    \end{align}
    \end{subequations}
    all hold.
    (In fact, in this case the experiment is MIC.)
    We know that one such tuple is $(s, t, x, y) = (x^{\ast}, x^{\ast}, x^{\ast}, y^{\ast})$, where $x^{\ast}$ is the threshold and $y^{\ast}$ is the atom under full delegation.
    We next argue that for all $y$ sufficiently close to $y^{\ast}$ there exist $(s, t, x)$ satisfying the above.
    
    The proof of \Cref{lemma:p_order} shows $g(x^{\ast}) < g(y^{\ast})$, which implies that the derivative of $x\mapsto G(x) + g(y^{\ast})(y^{\ast} - x) - G(y^{\ast})$ is strictly negative at $x^{\ast}$.
    By continuous differentiability of $g$, moreover, the function $(x, y)\mapsto G(x) + g(y)(y - x) - G(y)$ is continuously differentiable.
    Thus, by the Implicit Function Theorem, for all $y > y^{\ast}$ sufficiently close to $y^{\ast}$ there exists $\hat{x}(y)$ such that $G(\hat{x}(y)) + g(y)(y - \hat{x}(y)) = G(y)$, i.e. \eqref{eq:MIC_parametrization:1} holds.
    Moreover, $\hat{x}(y) \to x^{\ast}$ as $y\to y^{\ast}$ and $y\mapsto\hat{x}(y)$ is differentiable at $y^{\ast}$.
    Finally, $\hat{x}(y) \leq x^{\ast}$ for all $y$ (\Cref{lemma:p_order}).
    
    Next, using $y^{\ast} < 1$, it is easy to see (using, e.g., the Intermediate Value Theorem) that for all $y$ sufficiently close to $y^{\ast}$ there exists $\hat{t}(y) \in [0, y)$ solving \eqref{eq:MIC_parametrization:2}. Using the Implicit Function Theorem, on a neighborhood of $y^{\ast}$ the function $\hat{t}$ has the derivative
    \begin{equation*}
        \frac{\partial \hat{t}(y)}{\partial y} = \frac{1 - H(\hat{t}(y))}{(y - \hat{t}(y)) h(\hat{t}(y))} > 0.
    \end{equation*}
    For later reference, for $y > y^{\ast}$, equation \eqref{eq:MIC_parametrization:2} requires $\hat{t}(y) > x^{\ast}$.
    
    Finally, using also $0 < x^{\ast} < y^{\ast} < 1$, it is easy to see that for all $y$ sufficiently close to $y^{\ast}$ there exists $\hat{s}(y) \in (0, \hat{x}(y))$ such that $(\hat{s}(y), \hat{t}(y), \hat{x}(y), y)$ satisfies \eqref{eq:MIC_parametrization:3}.
    (We cannot use the Implicit Function Theorem to conclude that $\hat{s}$ is differentiable, which complicates bounding the integral \eqref{eq:perturbation:2} further ahead.)

    Thus, for some $\varepsilon > 0$ and all $y \in (y^{\ast}, y^{\ast} + \varepsilon]$, the tuple $(\hat{s}(y), \hat{t}(y), \hat{x}(y), y)$ defines an MIC double censorship experiment.
    We show that for $y$ sufficiently close to $y^{\ast}$ the designer is strictly better off than under unrestricted persuasion.
    Note, for $y\searrow y^{\ast}$, we have $(\hat{s}(y), \hat{t}(y), \hat{x}(y)) \to (x^{\ast}, x^{\ast}, x^{\ast})$ since, following the above arguments, there is a unique solution for \cref{eq:MIC_parametrization:1,eq:MIC_parametrization:2,eq:MIC_parametrization:3} that involves $y^{\ast}$, and $(x^{\ast}, x^{\ast}, x^{\ast}, y^{\ast})$ is a solution.

    The designer's utility from $(\hat{s}(y), \hat{t}(y), \hat{x}(y), y)$, denoted $U_{D}(y)$, is given by
    \begin{equation*}
        U_{D}(y) =
        \int_{[0, \hat{s}(y)]} u_{D}(\omega)\de H(\omega) + u_{D}(\hat{x}(y))(H(\hat{t}(y)) - H(\hat{s}(y))) + u_{D}(y)(1 - H(\hat{t}(y))).
    \end{equation*}
    The expected utility from $(x^{\ast}, x^{\ast}, x^{\ast}, y^{\ast})$ is given by
    \begin{equation*}
        U_{D}(y^{\ast}) =
        \int_{[0, x^{\ast}]} u_{D}(\omega) \de H(\omega) + u_{D}(y^{\ast})(1 - H(x^{\ast}))
        ,
    \end{equation*}
    which also obtains as the limit as $y\to y^{\ast}$ from above.
    Thus, 
    \begin{align}
        & U_{D}(y) - U_{D}(y^{\ast})
        \nonumber
        \\
        =
        &(u_{D}(y) - u_{D}(y^{\ast}))(1 - H(x^{\ast}))
        +
        (u_{D}(y) - u_{D}(\hat{x}(y)))(H(x^{\ast}) - H(\hat{t}(y)))
        \label{eq:perturbation:1}
        \\
        &+
        \int_{[\hat{s}(y), x^{\ast}]}(u_{D}(\hat{x}(y)) - u_{D}(\omega))\de H(\omega)
        .
        \label{eq:perturbation:2}
    \end{align}
    Dividing by $y - y^{\ast}$ and taking $y\to y^{\ast}$, the sum of the two terms in \eqref{eq:perturbation:1} converges to
    \begin{align*}
        & u_{D}^{\prime}(y^{\ast}) (1 - H(x^{\ast})) - (u_{D}(y^{\ast}) - u_{D}(x^{\ast})) h(x^{\ast}) \left.\frac{\partial \hat{t}(y)}{\partial y}\right\vert_{y = y^{\ast}}
        \\
        =
        & \alpha \left(\gamma^{\prime}(y^{\ast}) - \frac{\gamma(y^{\ast}) - \gamma(x^{\ast})}{y^{\ast} - x^{\ast}}\right)(1 - H(x^{\ast})) 
        \\
        & + (1 - \alpha) \left(g(y^{\ast}) - \frac{G(y^{\ast}) - G(x^{\ast})}{y^{\ast} - x^{\ast}}\right)(1 - H(x^{\ast})) 
        .
    \end{align*}
    This limit is strictly positive: the parenthesis involving $\gamma$ is strictly positive by convexity and since $y^{\ast} > x^{\ast}$, whereas the parenthesis involving $g$ equals $0$ since $(x^{\ast}, y^{\ast}) \in P$.
    Thus, to show $U_{D}(y) > U_{D}(y^{\ast})$ for $y$ sufficiently close to $y^{\ast}$, it suffices to show that
    the integral in \eqref{eq:perturbation:2} vanishes when normalized by $y - y^{\ast}$.
    We have
    \begin{align*}
        &\frac{1}{y - y^{\ast}}
        \left\vert\int_{[\hat{s}(y), x^{\ast}]}(u_{D}(\hat{x}(y)) - u_{D}(\omega))\de H(\omega)
        \right\vert
        \\
        \leq
        &
        \frac{1}{y - y^{\ast}}
        \left\vert
        (u_{D}(\hat{x}(y)) - u_{D}(x^{\ast})) (H(x^{\ast}) - H(\hat{s}(y)))
        \right\vert
        +
        \frac{1}{y - y^{\ast}}
        \left\vert
        \int_{[\hat{s}(y), x^{\ast}]}(u_{D}(x^{\ast}) - u_{D}(\omega))\de H(\omega)
        \right\vert
    \end{align*}
    Here, $\frac{1}{y - y^{\ast}}
        \left\vert
        (u_{D}(\hat{x}(y)) - u_{D}(x^{\ast})) (H(x^{\ast}) - H(\hat{s}(y)))
        \right\vert$ vanishes since $\hat{x}$ is differentiable at $y^{\ast}$, and since $\hat{s}(y) \to x^{\ast}$ as $y\to y^{\ast}$.
    Consider the integral. Find $L > 0$ such that $\vert u_{D}^{\prime}\vert \leq L$.
    Recall that $\hat{s}(y)$ is chosen such that
        $\int_{[\hat{s}(y), x^{\ast}]} (\hat{x}(y) - \omega)\de H(\omega) = \int_{[x^{\ast}, \hat{t}(y)]} (\omega - \hat{x}(y))\de H(\omega)$.
    We then have
    \begin{align*}
        &\frac{1}{y - y^{\ast}}
        \left\vert
        \int_{[\hat{s}(y), x^{\ast}]}(u_{D}(x^{\ast}) - u_{D}(\omega))\de H(\omega)
        \right\vert
        \\
        \leq
        &
        \frac{1}{y - y^{\ast}}
        \int_{[\hat{s}(y), x^{\ast}]} \left\vert \int_{[\omega, x^{\ast}]} u_{D}^{\prime}(r) \de r \right\vert \de H(\omega)
        \\
        \leq
        &
        \frac{L}{y - y^{\ast}}
        \int_{[\hat{s}(y), x^{\ast}]} (x^{\ast} - \omega) \de H(\omega)
        \\
        =
        &
        \frac{L}{y - y^{\ast}}
        (x^{\ast} - \hat{x}(y)) (H(x^{\ast}) - H(\hat{s}(y)))
        +
        \frac{L}{y - y^{\ast}}
        \int_{[\hat{s}(y), x^{\ast}]} (\hat{x}(y) - \omega)\de H(\omega)
        \\
        =
        &
        \frac{L}{y - y^{\ast}}
        (x^{\ast} - \hat{x}(y)) (H(x^{\ast}) - H(\hat{s}(y)))
        +
        \frac{L}{y - y^{\ast}}
        \int_{[x^{\ast}, \hat{t}(y)]} (\omega - \hat{x}(y))\de H(\omega)
        \\
        \leq
        &
        \frac{L}{y - y^{\ast}}
        (x^{\ast} - \hat{x}(y)) (H(x^{\ast}) - H(\hat{s}(y)))
        +
        \frac{L}{y - y^{\ast}}
        (\hat{t}(y) - \hat{x}(y)) (H(\hat{t}(y)) - H(x^{\ast}))
    \end{align*}
    Now, $\frac{1}{y - y^{\ast}} (x^{\ast} - \hat{x}(y)) (H(x^{\ast}) - H(\hat{s}(y)))$ vanishes as $y\to y^{\ast}$ since $\hat{x}$ is differentiable at $y^{\ast}$, and since $\hat{s}(y) \to x^{\ast}$ as $y\to y^{\ast}$.
    Moreover, $\frac{1}{y - y^{\ast}} (\hat{t}(y) - \hat{x}(y)) (H(\hat{t}(y)) - H(x^{\ast}))$ vanishes since $\hat{t}(y)$ is differentiable at $y^{\ast}$, and since $\hat{t}(y)$ and $\hat{x}(y)$ both converge to $x^{\ast}$ as $y\to y^{\ast}$.
    Thus, the integral in \eqref{eq:perturbation:2} vanishes when normalized by $y - y^{\ast}$.
    \qed

\subsection{Proof of \headercref{Corollary}{{cor:MIC_total_order}}}
    Let $(s, t, x, y)$ and $(s^{\prime}, t^{\prime}, x^{\prime}, y^{\prime})$ be the parameters of two MIC experiments $F$ and $F'$ such that $y^{\prime} < y$.
    Since $(x, y)$ and $(x^{\prime}, y^{\prime})$ are in $P$, \Cref{lemma:p_order} implies $x \leq x^{\prime}$; the proof of this lemma delivers $x < x^{\prime}$.
    Using the expressions for $x$, $x^{\prime}$, $y$, and $y^{\prime}$ as conditional expectations, it also follows $t > t'$ and $s < s'$.
    Thus, $F$ is more disciplined than $F'$.
    Finally, if $y' = y$, one may verify that also $x = x'$, $t = t'$, and $s = s'$, so that $F$ and $F'$ coincide. Thus, MIC experiments are totally ordered by discipline.

    Now let $F$ and $F'$ be two MIC experiments such that $F \succ F'$.
    The restriction $\bar{F}'$ that pools the prior on $[s', t']$ and that implements $F'$ (\Cref{cor:double_censorship_implementation}) is an MPS of $F$ since $s \leq s'$ and $t' \leq t$. Thus, under $\bar{F}'$ both $F$ and $F'$ are feasible, and the experimenter finds $F'$ optimal. Thus, the experimenter prefers $F'$ to $F$.
    \qed

\subsection{Comparative statics}

\begin{lemma}\label{lemma:MIC_SCP}
If $F$ and $F'$ are MIC double censorship experiments such that $F \succ F'$, then $I_{F} - I_{F'}$ is weakly single-crossing from below.
\end{lemma}
\begin{proof}[Proof of \Cref{lemma:MIC_SCP}]
    From \Cref{cor:MIC_total_order}, we know $t > t^{\prime}$ and $x < x^{\prime}$ and $s < s^{\prime}$.
    
    We proceed along a series of subintervals.
    First, clearly, $I_{F} \geq I_{\delta_{\mu}} = I_{F^{\prime}}$ on $[y^{\prime}, 1]$.
    
    Next, consider $[x^{\prime}, y^{\prime}]$. 
    Since $x < x^{\prime} \leq y^{\prime} \leq y$, both $I_{F}$ and $I_{F^{\prime}}$ are affine on $[x^{\prime}, y^{\prime}]$.
    As already established, $I_{F}(y^{\prime}) \geq I_{F^{\prime}}(y^{\prime})$.
    We also know $I_{F^{\prime}}(t^{\prime}) = I_{H}(t^{\prime}) \geq I_{F}(t^{\prime})$ and $t^{\prime} \in [x^{\prime}, y^{\prime}]$.
    Given these two inequalities, and since two affine functions intersect at most once or coincide, it follows that on $[x^{\prime}, y^{\prime}]$ the functions $I_{F}$ and $I_{F^{\prime}}$ intersect at a point $z_{0}$ in $[t^{\prime}, y^{\prime}]$.
    Using again that two affine functions intersect at most once, we also conclude $I_{F^{\prime}}(x^{\prime}) \geq I_{F}(x^{\prime})$.

    Now, consider the interval $[s^{\prime}, x^{\prime}]$.
    On this interval, $I_{F^{\prime}}$ is affine whereas $I_{F}$ is convex.
    We have already seen $I_{F^{\prime}}(x^{\prime}) \geq I_{F}(x')$.
    Moreover, clearly $I_{F^{\prime}}(s^{\prime}) = I_{H}(s^{\prime}) \geq I_{F}(s^{\prime})$.
    Thus, $I_{F^{\prime}} \geq I_{F}$ on $[s^{\prime}, x^{\prime}]$.

    Finally, on $[0, s^{\prime}]$ we have $I_{F^{\prime}} = I_{H} \geq I_{F}$.
    
    In summary, $I_{F} \leq I_{F'}$ on $[0, z_{0}]$ whereas $I_{F} \geq I_{F'}$ on $[z_{0}, 1]$.
\end{proof}

\subsubsection{Proof of \headercref{Proposition}{{prop:convex_comparative_statics}}}
Fix smooth $u_{D}$ and an experiment $F$.
For all $m \in [0, 1]$, it holds:
\begin{align*}
    u_{D}(m) - u_{D}(0)
    &=
    \int_{0}^{m} \left(u_{D}^{\prime}(0) + \int_{0}^{m^{\prime}} u_{D}^{\prime\prime}(m^{\prime\prime}) \de m^{\prime\prime}\right) \de m^{\prime}
    \\ &=
    m  u_{D}^{\prime}(0)
    +
    \int_{0}^{m} (m - m^{\prime\prime}) u_{D}^{\prime\prime}(m^{\prime\prime}) \de m^{\prime\prime}
    .
\end{align*}
Also, for all $m^{\prime\prime} \in [0, 1]$,
\begin{align*}
    \int_{m^{\prime\prime}}^{1} (m - m^{\prime\prime}) \de F(m)
    =
    \int_{m^{\prime\prime}}^{1} \int_{m^{\prime\prime}}^{m} 1 \de m^{\prime} \de F(m)
    =
    &\int_{m^{\prime\prime}}^{1} (1 - F(m^{\prime})) \de m^{\prime}
    \\ =
    &\mu - m^{\prime\prime} + I_{F}(m^{\prime\prime})
    .
\end{align*}
Thus, the designer's utility from the experiment $F$ under $u_{D}$ is given by
\begin{align*}
    U_{D}(F, u_{D}) = &
    \int_{0}^{1} u_{D}(m) \de F(m)
    \\
    =&
    u_{D}(0) + \mu u_{D}^{\prime}(0)
    +
    \int_{0}^{1} u_{D}^{\prime\prime}(m^{\prime\prime}) \int_{m^{\prime\prime}}^{1} (m - m^{\prime\prime})  \de F(m) \de m^{\prime\prime}
    \\
    =&
    u_{D}(0) + \mu u_{D}^{\prime}(0)
    +
    \int_{0}^{1} u_{D}^{\prime\prime}(m^{\prime\prime})
    \left(\mu - m^{\prime\prime} + I_{F}(m^{\prime\prime})\right)
    \de m^{\prime\prime}
    .
\end{align*}
Therefore, for two experiments $F$ and $F'$,
\begin{equation*}
    U_{D}(F, u_{D}) - U_{D}(F', u_{D}) 
    =
    \int_{0}^{1} \left(I_{F}(m) - I_{F'}(m) \right) u_{D}^{\prime\prime}(m) \de m
    .
\end{equation*}
Now let $F$ and $F'$ be MIC experiments such that $F \succ F'$.
Let $u_{D,1}$ and $u_{D, 2}$ be smooth, strictly convex, and such that $u_{D, 2}$ is more information-sensitive at high states.
According to \Cref{lemma:MIC_SCP}, $I_{F} - I_{F'}$ is weakly single-crossing from below.
Given $u_{D, i}^{\prime\prime} > 0$ for both $i$, higher information-sensitivity at higher states precisely means that $u_{D, i}^{\prime\prime}(m)$ is totally positive of order $2$ as a function of $(i, m)$.
By the Variation Diminishing Property (\citet{karlin1956theory}), the difference $i\mapsto  U_{D}(F, u_{D, i}) - U_{D}(F', u_{D, i})$ is single-crossing from below.\footnote{For the sake of completeness, let us spell out the argument. Let $z_{0}$ be a weak single-crossing point of $I_{F} - I_{F'}$, i.e. $I_{F} - I_{F'}$ is weakly negative on $[0, z_{0}]$ and weakly positive on $[z_{0}, 1]$. Then, all $m$ satisfy 
\begin{equation*}
    \left(I_{F}(m) - I_{F'}(m) \right) \left(u_{D, 2}^{\prime\prime}(m) - u_{D, 1}^{\prime\prime}(m) \frac{u_{D, 2}^{\prime\prime}(z_{0})}{u_{D, 1}^{\prime\prime}(z_{0})}\right) \geq 0.
\end{equation*}
Integrating over $m$ shows 
$U_{D}(F, u_{D,2}) - U_{D}(F', u_{D, 2}) \geq \frac{u_{D, 2}^{\prime\prime}(z_{0})}{u_{D, 1}^{\prime\prime}(z_{0})} (U_{D}(F, u_{D, 1}) - U_{D}(F', u_{D, 1}))$.
In particular,
$\sign (U_{D}(F, u_{D,2}) - U_{D}(F', u_{D, 2}))  \geq \sign (U_{D}(F, u_{D, 1}) - U_{D}(F', u_{D, 1}))$, which is the claimed single-crossing property.}
Invoking the Monotonicity Theorem of \citet{milgrom1994monotone} completes the proof.
\qed

\subsubsection{Proof of \headercref{Proposition}{{prop:welfare_comparative_statics}}}
Let $U_{D}(F, \alpha) = \int (\alpha \gamma(m) + (1- \alpha) G(m)) \de F(m)$ be the designer's utility from experiment $F$ given weight $1-\alpha$ on the experimenter's utility.
We show that if $F$ and $F'$ are MIC experiments with $F \succ F'$, then $\alpha\mapsto U_{D}(F, \alpha) - U_{D}(F', \alpha)$ is single-crossing from below.
\Cref{prop:welfare_comparative_statics} follows from this single-crossing property and the Monotonicity Theorem of \citet{milgrom1994monotone}.

Let $F\succ F'$.
Let $ U_{D}(F, \alpha') \geq U_{D}(F', \alpha')$ for some $\alpha'$. 
Equivalently, 
\begin{equation}\label{eq:proof:welfare_comparative_statics:ineq}
    \alpha' \int \gamma(m) \de (F(m) - F'(m)) \geq (1 - \alpha') \int G(m) \de (F'(m) - F(m)).
\end{equation}
According to \Cref{cor:MIC_total_order}, $F \succ F'$ implies that the experimenter prefers $F'$ to $F$, meaning $\int G(m) \de (F'(m) - F(m)) \geq 0$.
We conclude that the integrals on both sides of \eqref{eq:proof:welfare_comparative_statics:ineq} are weakly positive. Hence, raising $\alpha'$ preserves the inequality. Thus, for $\alpha > \alpha'$, also $ U_{D}(F, \alpha) \geq U_{D}(F', \alpha)$.
An analogous argument shows that $ U_{D}(F, \alpha') > U_{D}(F', \alpha')$ and $\alpha > \alpha'$ implies $ U_{D}(F, \alpha) > U_{D}(F', \alpha)$.
Thus, $\alpha\mapsto U_{D}(F, \alpha) - U_{D}(F', \alpha)$ is single-crossing from below.
\qed

\bibliographystyle{ecta-fullname} 
\bibliography{references}


\newpage
\clearpage

\begin{center}
    {\Large \textbf{Online Appendix: Delegated Information Provision}}  \\ \vspace{1cm}  Francesco Bilotta \hspace{1cm} Christoph Carnehl \hspace{1cm} Justus Preusser  \vspace{1cm}
\end{center}

\appendix
\vspace{1cm}
\setcounter{section}{0}
\renewcommand{\thesection}{O\Alph{section}}
\renewcommand{\theHsection}{O.\Alph{section}}

\section{Existence of an Optimal Restriction}\label{appendix:optimal_existence}

For the following lemma, let $G$ be continuous (and not necessarily $S$-shaped). In particular, we can apply the IC characterization  in \Cref{lemma:price_function_IC}.
Equip $\MPC(H)$ with the $L^{1}$-norm, making it compact.
\begin{lemma}\label{lemma:optimal_existence}
    The set of IC experiments is a compact subset of $\MPC(H)$.
    There exists an optimal IC experiment.
\end{lemma}
For continuous $G$, existence of an optimal IC experiment follows from Lemma 1 of \citet{lipnowski2020attention}, who consider settings with an abstract state.
We give a simpler argument for the mean-measurable setting.

\begin{proof}[Proof of \Cref{lemma:optimal_existence}]
    It suffices to show that the set of IC experiments is closed.
    Let $(F_{n})_{n\in\mathbb{N}}$ be a sequence of IC experiments that $L^{1}$-converges to an experiment $F$.
    Let $A = \min\supp F$ and $B=\max\supp F$.
    In view of \Cref{lemma:price_function_IC}, we show $F$ is IC by finding a convex $p\colon [A, B]\to\mathbb{R}$ such that $p\geq G$ and $\supp F\subseteq\lbrace m\colon p(m) = G(m)\rbrace$.
    Indeed, set $p = G$ on $\supp F$. For all $c\in [A, B]\setminus \supp F$, define $p(c)$ by linear interpolation, i.e. if $(a, b)$ is the largest open interval in $[A, B]\setminus \supp F$ that contains $c$, then $a, b\in \supp F$ and we let 
    \begin{equation*}
        p(c) = \frac{c-a}{b-a} G(b) + \frac{b-c}{b-a}G(a).
    \end{equation*}
    Clearly, $\supp F\subseteq\lbrace c\colon p(c) = G(c)\rbrace$.
    We have to show $p$ is convex and $p\geq G$ holds.
    
    In an auxiliary step, we show that if $a < c < b$ and $a, b\in \supp F$, then
    \begin{equation}\label{eq:existence_price_convexity}
        G(c) \leq \frac{c - a}{b - a} p(b) + \frac{b - c}{b - a} p(a)
        .
    \end{equation}
    For all $k\in\mathbb{N}$, let $B_{1/k}(a)$ and $B_{1/k}(b)$ be the $(1/k)$-balls around $a$ and $b$. Then, $F(B_{1/k}(a)) > 0$ and $F(B_{1/k}(b)) > 0$ by definition of the support.
    Since $(F_{n})_{n}$ $L^{1}$-converges to $F$, there is $n_{k}$ such that $F_{n_{k}}(B_{1/k}(a)) > 0$ and $F_{n_{k}}(B_{1/k}(b)) > 0$.
    Thus, there is $a_{k} \in B_{1/k}(a) \cap \supp F_{n_{k}}$ and $b_{k} \in B_{1/k}(b) \cap \supp F_{n_{k}}$.
    Since $F_{n_{k}}$ is IC, \Cref{lemma:price_function_IC} yields a convex $p_{k}$ (defined on the convex hull of the support of $F_{n_{k}}$) such that $p_{k}\geq G$ and $p_{k}(a_{k}) = G(a_{k})$ and $p_{k}(b_{k}) = G(b_{k})$.
    For large enough $k$, we have $a_{k} < c < b_{k}$. Thus,
    \begin{equation*}
        G(c) \leq p_{k}(c) \leq 
        \frac{c - a_{k}}{b_{k} - a_{k}} p_{k}(b_{k}) + \frac{b_{k} - c}{b_{k} - a_{k}} p_{k}(a_{k})
        =
        \frac{c - a_{k}}{b_{k} - a_{k}} G(b_{k}) + \frac{b_{k} - c}{b_{k} - a_{k}} G(a_{k}).
    \end{equation*}
    Since $G$ is upper semicontinuous and coincides with $p$ at $a$, we get $\limsup_{k\to\infty} G(a_{k}) \leq G(a) = p(a)$; similarly, $\limsup_{k\to\infty} G(b_{k}) \leq p(b)$.
    Thus, taking $k\to\infty$ yields \eqref{eq:existence_price_convexity}.
    
    To verify $p\geq G$ on $[A, B]$, recall $p = G$ on $\supp F$ by construction.
    For $c\in [A, B] \setminus \supp F$, let $(a, b)$ be the largest open interval in $[A, B]\setminus \supp F$ containing $c$. Then, $p(c)$ is defined as the right side of \eqref{eq:existence_price_convexity}, which is at least $G(c)$.
    
    Turning to convexity of $p$, we evaluate \eqref{eq:existence_price_convexity} at points $a < c < b$ that are all in $\supp F$, obtaining
    \begin{equation*}
        p(c) = G(c) \leq \frac{c - a}{b - a} p(b) + \frac{b - c}{b - a} p(a)
        .
    \end{equation*}
    Since $p$ is defined by linear interpolation on $[A, B]\setminus \supp F$, it is a standard fact that the same inequality also holds for arbitrary $a < c < b$ in $[A, B]$.
    Thus, $p$ is convex.
\end{proof}

\section{Ordered Experiments and Garblings}\label{appendix:ordered_experiments}
\begin{lemma}\label{lemma:ordered_experiments}
    If $F\in \MPC(H)$ and $F^{\prime}\in \MPC(F)$, then there exist a signal $(\mathcal{S}, \pi)$ inducing $F$, and a signal $(\mathcal{S}^{\prime}, \pi^{\prime})$ inducing $F^{\prime}$ such that $(\mathcal{S}^{\prime}, \pi^{\prime})$ is a garbling of $(\mathcal{S}, \pi)$.
\end{lemma}
\begin{proof}[Proof of \Cref{lemma:ordered_experiments}]
Since $F\in \MPC(H)$, there exists $K\in \Delta([0, 1]^{2})$ whose marginals are $F$ and $H$, and such that $m = \int \omega \de K(\omega\vert m)$ for $F$-almost all $m$. Similarly, since $F^{\prime}\in \MPC(F)$, there exists $K^{\prime}\in\Delta([0, 1]^{2})$ whose marginals are $F^{\prime}$ and $F$, and such that $m^{\prime} = \int m \de K^{\prime}(m\vert m^{\prime})$ for $F^{\prime}$-almost all $m^{\prime}$. Now define a joint CDF $K^{\ast}$ on $[0, 1]^{3}$ for all $m^{\prime}, m, \omega$ by $K^{\ast}(m^{\prime}, m, \omega) = \int_{[0, 1]^{2}} \bm{1}_{\tilde{m} \leq m} \bm{1}_{\tilde{\omega}\leq \omega} K^{\prime}(m^{\prime}\vert \tilde{m}) \de K(\tilde{m}, \tilde{\omega})$.
The joint $K^{\ast}$ yields random variables $\bm{\omega}$, $\bm{m}$, $\bm{m}^{\prime}$ such that (i) $\bm{\omega}\sim H$, (ii) conditional on $\bm{\omega} = \omega$, the CDF of $\bm{m}$ is $K(\cdot\vert \omega)$, and (iii), conditional on $\bm{\omega} = \omega$ and $\bm{m} = m$ the CDF of $\bm{m}^{\prime}$ is $K^{\prime}(\cdot\vert m)$ (note the conditional independence). The random variable $\bm{m}$ defines a signal whose distributions of posterior means is $F$ since $m = \int \omega \de K(\omega\vert m)$ holds for $F$-almost all $m$. Similarly, $\bm{m}^{\prime}$ defines a signal whose distribution of posterior means is $F^{\prime}$. By the conditional independence in (iii), the signal induced by $\bm{m}^{\prime}$ is a garbling of the signal induced by $\bm{m}$.
\end{proof}

\section{MIC Experiments as Extreme Points}\label{appendix:extreme_points}
In this appendix, we establish a link between extreme points of $\MPC(H)$ and MIC experiments for general, continuous experimenter payoffs $G$.

The definitions of incentive-compatibility and maximality are as before. Assuming that, as in the baseline model, the designer's payoff is given by $u_{D} = \alpha \gamma + (1 - \alpha) G$ with $\gamma$ convex and $\alpha \in [0, 1]$, it holds that MIC experiments are without loss, which can be seen by inspecting the proof of \Cref{lemma:MIC_wlog}.
Our price-theoretic characterization of incentive-compatibility (\Cref{lemma:price_function_IC}) remains applicable given that $G$ is continuous.
In contrast, the characterization of MIC experiments as double censorship uses $S$-shapedness of $G$ and does not apply.

At first sight, the extreme points of $\MPC(H)$ may seem unrelated to MIC experiments. 
The extreme point characterization reflects the feasibility constraints of being a mean-preserving contraction of $H$; in particular, the characterization is independent of preferences. 
In contrast, MIC experiments are precisely characterized by the experimenter's incentives to garble information.
Nevertheless:

\begin{theorem}\label{thm:general_extreme_points}
    Every MIC experiment is an extreme point of $\MPC(H)$.
\end{theorem}
Since MIC experiments are extreme points of the \emph{entire} set $\MPC(H)$, the literature's characterization of the extreme points as bi-pooling experiments (\citet{kleiner2021extreme,arieli2023optimal}) applies under the novel MIC constraint. 

Typically, MIC experiments form a proper subset of all extreme points. For example, the degenerate experiment $\delta_{\mu}$ is incentive-compatible, but it is not maximal except in uninteresting instances of the model. Similarly, the fully revealing experiment $H$ would be maximal if it were incentive-compatible, but it is incentive-compatible only if the experimenter would fully reveal the state when unrestricted.

Our proof leverages the characterization of extreme points as bi-pooling experiments.
For a high-level intuition, note that any experiment $F\in\MPC(H)$ is weakly increasing, and that the ICDF of $F$ lies weakly below the ICDF of the prior $H$. Roughly, the extreme point characterization states that at least one of these two constraints must bind locally. When the monotonicity constraint binds on an interval, $F$ assigns no mass to this interval; when the ICDF constraint binds, $F$ is maximally informative on this interval. Likewise, incentive-compatibility and maximality are characterized by two constraints on the respective experiment. Incentive-compatibility requires that $F$ assigns no mass to certain intervals (\Cref{lemma:price_function_IC}); this is akin to the binding monotonicity constraint. Maximality requires that $F$ is maximally informative otherwise; this is akin to the binding informativeness constraint.

In some settings, the designer's preferences may not justify focusing on MIC experiments. For these cases, we observe that IC experiments can nevertheless be represented by IC extreme points. 
Thus, IC extreme points are without loss generally.
This fact follows from results of \citet[see their Online Appendix]{lipnowski2020attention} for settings with an abstract state and continuous experimenter payoffs.
For the sake of completeness, here we prove this representation result for our special case with posterior-mean preferences via a shorter argument.

Equip $\MPC(H)$ with the $L^{1}$-norm.
For a Borel probability measure $\nu$ on $\MPC(H)$ and an experiment $F\in\MPC(H)$, say $\nu$ \emph{represents} $F$ if $V(F) = \int V(\Phi) \de\nu(\Phi)$ holds for all continuous linear functionals $V\colon\MPC(H)\to \mathbb{R}$.
Say $\nu$ is \emph{supported on} a measurable set $\mathcal{F}\subseteq \MPC(H)$ if $\nu(\mathcal{F}) = 1$.

\begin{theorem}\label{thm:ic-extreme-points}
    For all IC experiments $F$ there exists a Borel probability measure $\nu$ that represents $F$ and that is supported on the set of incentive-compatible extreme points of $\MPC(H)$.\footnote{The set of incentive-compatible extreme points of $\MPC(H)$ is measurable since IC experiments form a closed set (see \Cref{appendix:optimal_existence}) while the extreme points form a $G_{\delta}$-set (\citet[Lemma 7.63]{aliprantis2006infinite}).}
\end{theorem}
The content of \Cref{thm:ic-extreme-points} is that the extreme points can be chosen to be IC, and that they are extreme points of the entire set $\MPC(H)$.

\Cref{thm:ic-extreme-points} does \emph{not} imply that the set of IC experiments is convex. 
Indeed, it is not generally convex: for $S$-shaped preferences of the experimenter, IC experiments can have at most one atom in the concave part of $G$ (\Cref{lemma:s_shaped_IC}), implying that a proper mixture of two distinct double censorship experiments is not IC.

\subsection{Proof of \headercref{Theorem}{{thm:general_extreme_points}}}
    Let $F$ be MIC.
    For later reference, by \Cref{lemma:price_function_IC}, there exists a convex function $p\colon [\min\supp F, \max\supp F]\to\mathbb{R}$ such that $p\geq G$ and $\supp F\subseteq \lbrace m\colon p(m) = G(m)\rbrace$.

    By \citet[Theorem 2]{kleiner2021extreme}, to show that $F$ is an extreme point of $\MPC(H)$, we have to show that, for every open subinterval $(a, b)$ such that $I_{F}(m) < I_{H}(m)$ for all $m\in (a, b)$, the support of $F$ has at most two points in $(a, b)$.

    Let $[a, b]$ be a subinterval of $[0, 1]$ such that $I_{F}(m) < I_{H}(m)$ for all $m\in (a, b)$.
    
    Let $m\in (\supp F) \cap (a, b)$.
    We claim that there is a neighborhood $O$ of $m$ such that $(\supp F) \cap O = \lbrace m\rbrace$.
    Towards a contradiction, suppose not.
    Then, there is a sequence $(m_{n})_{n\in\mathbb{N}}$ in $((\supp F) \setminus \lbrace m\rbrace) \cap (a, b)$ converging to $m$.
    By possibly passing to a subsequence, let the sequence be strictly monotone.
    Assume the sequence is strictly increasing, the case where it is strictly decreasing being similar.
    Consider the line segment between $(m_{n}, I_{F}(m_{n}))$ and $(m, I_{F}(m))$; it is given by $t\mapsto I_{F}(m_{n}) + (I_{F}(m) - I_{F}(m_{n}))(t - m_{n}) / (m - m_{n})$ for $t\in [m_{n}, m]$.
    Find $n$ sufficiently large such that this line segment does not intersect $I_{H}$ on $[m_{n}, m]$; such $n$ exists since $I_{F}$ lies strictly below $I_{H}$ on $(a, b)$ and since ICDFs are $1$-Lipschitz.
    Consider the function that coincides with $I_{F}$ below $m_{n}$ and above $m$, and that coincides with the chosen line segment on $[m_{n}, m]$.
    This function is weakly increasing and convex, coincides with $I_{H}$ at $0$ and $1$, and lies below $I_{H}$.
    Thus, its derivative $\tilde{F}$ is an experiment.
    Since $m_{n}$ and $m$ are in the support of $F$, and since $I_{\tilde{F}}$ and $I_{F}$ coincide outside $[m_{n}, m]$, it holds $\supp \tilde{F}\subseteq \supp F \subseteq \lbrace m^{\prime}\colon p(m^{\prime}) = G(m^{\prime})\rbrace$.
    In particular, $\tilde{F}$ is IC (\Cref{lemma:price_function_IC}).
    Finally, $\tilde{F}$ is an MPS of $F$ since $I_{F}$ is convex on $[m_{n}, m]$ while $I_{\tilde{F}}$ is affine on $[m_{n}, m]$ and coincides with $I_{F}$ at $m_{n}$ and $m$.
    Since $F$ is maximal, $F = \tilde{F}$.
    In particular, $I_{F}$ is affine on $[m_{n}, m]$, meaning $F$ assigns no mass to the open interval $(m_{n}, m)$.
    However, by strict monotonicity, $m_{n+1} \in (m_{n}, m)$, contradicting $m_{n+1} \in \supp F$.

    We now claim that for all closed subintervals $[c, d]$ of $(a, b)$ the intersection $(\supp F) \cap [c, d]$ is finite. Indeed, for all $m\in (\supp F) \cap [c, d]$, let $O_{m}$ be as in the previous paragraph, i.e. $(\supp F) \cap O_{m} = \lbrace m\rbrace$. This yields an open cover of the compact set $(\supp F) \cap [c, d]$. Now find a finite subcover to conclude that $(\supp F) \cap [c, d]$ is finite.

    Let $[c, d]\subseteq (a, b)$ be arbitrary.
    We claim $(\supp F) \cap [c, d]$ contains at most two points.
    Towards a contradiction, suppose not.
    Since $(\supp F) \cap [c, d]$ is finite, there exist three consecutive points $x, y, z$ in $(\supp F) \cap [c, d]$; by ``consecutive'' we mean that there is no other point in $\supp F$ between $x$ and $y$, nor between $y$ and $z$.
    In particular, $I_{F}$ is affine between $x$ and $y$; call the slope $\alpha$.
    Also, $I_{F}$ is affine between $y$ and $z$; call the slope $\beta$.
    Since $y$ is in the support of $F$, it holds $\alpha < \beta$.
    Consider the function $\tilde{J}$ obtained from $I_{F}$ by raising $\alpha$ by $\varepsilon_{\alpha} > 0$ while decreasing $\beta$ by $\varepsilon_{\beta} > 0$, holding fixed $I_{F}$ at $x$ and at $z$.
    (Thus, $I_{F}$ is shifted up at $y$.)
    Finally, let $\tilde{J}$ coincide with $I_{F}$ outside $[x, z]$.
    For $\varepsilon_{\alpha}$ and $\varepsilon_{\beta}$ sufficiently small, the function $\tilde{J}$ is convex since $\alpha < \beta$.
    Moreover, since $I_{F}$ lies strictly below $I_{H}$ on $(a, b)$, and since $I_{F}$ is affine on each of $[x, y]$ and $[y, z]$, for $\varepsilon_{\alpha}$ and $\varepsilon_{\beta}$ sufficiently small also $\tilde{J}$ lies below $I_{H}$ on $[x, z]$ (and thus on all of $[0, 1]$).\footnote{To see this, note $I_{H} - I_{F}$ is positive and bounded away from $0$ across $[x, z]$, by continuity. As $\varepsilon_{\alpha} \to 0$ and $\varepsilon_{\beta} \to 0$, the function $\tilde{J}$ converges uniformly to $I_{F}$ on $[x, z]$ since $I_{F}$ is piecewise affine.}
    Thus, the derivative of $\tilde{J}$, call it $\tilde{F}$, is an experiment.
    By construction $\tilde{F}$ is a proper MPS of $F$.
    Finally, since $\supp\tilde{F}\subseteq \supp F \subseteq \lbrace m\colon p(m) = G(m)\rbrace$, \Cref{lemma:price_function_IC} implies that $\tilde{F}$ is IC. Contradiction to maximality of $F$.

    Since in the previous paragraph $[c, d] \subseteq (a, b)$ was arbitrary, the intersection $(\supp F) \cap (a, b)$ contains at most two points.
    Since $[a, b]$ was an arbitrary interval such that $I_{F}< I_{H}$ on $(a, b)$, Theorem 2 of \citet{kleiner2021extreme} implies that $F$ is an extreme point of $\MPC(H)$. 
    \qed

\subsection{Proof of \headercref{Theorem}{{thm:ic-extreme-points}}}
    Let $F$ be IC.
    By \Cref{lemma:price_function_IC}, there exists a convex $p\colon [\min\supp F, \max\supp F]\to\mathbb{R}$ such that $p\geq G$ and $\supp F\subseteq \lbrace m\colon p(m) = G(m)\rbrace$.
    By Choquet's Representation Theorem, there exists a measure $\nu$ that represents $F$ and is supported on the extreme points of $\MPC(H)$.
    We argue that $\nu$-almost all $\Phi$ are IC.
    Let $\gamma$ be a continuous function such that $\gamma(m) > 0$ for all $m\notin \supp F$ and $\gamma(m) = 0$ for all $m\in \supp F$; e.g., let $\gamma(m) = \min_{m^{\prime}\in \supp F}\vert m - m^{\prime}\vert$, where the minimum is well-defined and continuous by closedness of $\supp F$ and Berge's Maximum Theorem.
    Thus, the linear functional $\Phi\mapsto \int \gamma(m)\de\Phi(m)$ is continuous.
    Since $\gamma$ vanishes on $\supp F$, and since $\nu$ represent $F$, we have $0 = \int \gamma(m)\de F(m) = \int \int_{[0, 1]\setminus \supp F} \gamma(m)\de\Phi(m)\de\nu(\Phi)$.
    Since $\gamma$ is strictly positive outside $\supp F$, it follows that $\nu$-almost all $\Phi$ assign mass $0$ to the set $[0, 1]\setminus \supp F$. 
    Since $\supp F$ is closed, it follows that $\nu$-almost all $\Phi$ satisfy $\supp \Phi \subseteq \supp F \subseteq \lbrace m\colon p(m) = G(m)\rbrace$.
    Invoking \Cref{lemma:price_function_IC} with the price function $p$, it follows that $\nu$-almost all $\Phi$ are IC.
    \qed

\section{Profitable Restrictions Beyond \textit{S}-shaped Persuasion}\label{app:general}
We provide the definition of locally strict curvature preferences and incomplete full revelation that we alluded to in \Cref{sec:discussion}.
Then, we prove that these conditions are sufficient for the existence of profitable restrictions when the designer has strict preferences for information, and under smoothness assumptions on $H$, $G$, and $u_{D}$.

\subsection{Definitions}

\subsubsection{Locally strict curvature preferences}
We first introduce an auxiliary definition.
Let $T = (x_{1}, y_{1}, y_{2}, x_{2}) \in (0, 1)^{4} $ be such that $x_{1} < y_{1} \leq y_{2} < x_{2}$.
Let $p\colon [x_{1}, x_{2}]\to \mathbb{R}$ be a function.
Say $T$ is \emph{supported by} $p$ if $p$ is continuous and convex, lies weakly above $G$ on $[x_{1}, x_{2}]$, and coincides with $G$ at each point in $\lbrace x_{1}, y_{1}, y_{2}, x_{2}\rbrace$.
The tuple $T$ \emph{can be supported by a price function} if $T$ is supported by at least one function $p$.
In this definition, $y_{1}$ and $y_{2}$ should be thought of the atoms in a bi-pooling interval $[x_{1}, x_{2}]$ of an extreme point (possibly, $y_{1} = y_{2}$).

\begin{definition}\label{assumption:smooth_support}
    The experimenter has \emph{locally strict curvature preferences} if the following hold.

    (I) The state space $[0, 1]$ can be partitioned into a finite number of intervals such that, on each interval, $G$ is either strictly convex or strictly concave. 
    Let $\Cvx$ and $\Ccv$, respectively, be the union of intervals on which $G$ is strictly convex and strictly concave, respectively.

    (II) For all $T = (x_{1}, y_{1}, y_{2}, x_{2})$ such that $x_{1}, x_{2} \in \Cvx$ and $y_{1}, y_{2} \in \Ccv$ and $0 < x_{1} < y_{1} \leq y_{2} < x_{2} < 1$, if $T$ can be supported by an affine function,
    then there exists a continuous function $\tilde{T} = (\tilde{x}_{1}, \tilde{y}_{1}, \tilde{y}_{2}, \tilde{x}_{2})\colon [0, 1]\to (0, 1)^{4}$ such that:
    \begin{itemize}
        \item it holds $\tilde{T}(0) = T$ and $\tilde{x}_{1}(r) \leq x_{1}$ and $\tilde{x}_{2}(r) \geq x_{2}$ for all $r\in (0, 1]$;
        \item $\tilde{y}_{1}$ and $\tilde{y}_{2}$ are continuously differentiable, and at $r = 0$ their derivatives are non-zero; further, $\tilde{x}_{1}$ and $\tilde{x}_{2}$ are differentiable at $r = 0$;
        \item for all $r\in (0, 1]$, the tuple $\tilde{T}(r)$ can be supported by a price function.
    \end{itemize}
\end{definition}

In plain words, \Cref{assumption:smooth_support} says that if some tuple can be supported by an affine price function, then there are nearby tuples that ``smoothly approximate'' the given one and that can also be supported by price functions.
\Cref{fig:smooth_support:bipool,fig:smooth_support:pool} depict such approximations schematically.
In both figures, there is an affine function $p$ (dashed blue) that coincides with $G$ at $x_{1}$, $y_{1}$, $y_{2}$, and $x_{2}$, and otherwise lies above $G$. Further, $p$ and $G$ are tangential at $y_{1}$ and $y_{2}$, where $G$ is concave.
The points $y_{1}$ and $ y_{2}$ are distinct in \Cref{fig:smooth_support:bipool}; they coincide in \Cref{fig:smooth_support:pool}.
As alluded to earlier, $y_{1}$ and $y_{2}$ should be thought of the atoms in a bi-pooling interval $[x_{1}, x_{2}]$ of an extreme point (possibly, $y_{1} = y_{2}$).
The affine function $p$ is a price function (restricted to the interval $[x_{1}, x_{2}]$) verifying that this extreme point is a best reply of the experimenter to the prior.

The dotted orange line depicts a price function $\tilde{p}$ that touches $G$ at nearby but distinct points $\tilde{x}_{1}$, $\tilde{y}_{1}$, $\tilde{y}_{2}$, and $\tilde{x}_{2}$.
The points $\tilde{y}_{1}$ and $\tilde{y}_{2}$ are distinct in \Cref{fig:smooth_support:bipool}; they coincide in \Cref{fig:smooth_support:pool}. 
A substantive assumption in \Cref{fig:smooth_support:bipool} is that the kinked part of $\tilde{p}$ between $\tilde{y}_{1}$ and $\tilde{y}_{2}$ lies above $G$; likewise, in \Cref{fig:smooth_support:pool} the kinked part of $\tilde{p}$ between $\tilde{y}_{1}$ and $\tilde{x}_{2}$ lies above $G$.
As a further substantive assumption, one can choose $(\tilde{x}_{1}, \tilde{y}_{1}, \tilde{y}_{2}, \tilde{x}_{2})$ arbitrarily close to $(x_{1}, y_{1}, y_{2}, x_{2})$ to meet the differentiability assumptions of \Cref{assumption:smooth_support}.
We shall use the price function $\tilde{p}$ to construct other IC experiments whose support contains $\tilde{x}_{1}$, $\tilde{y}_{1}$, $\tilde{y}_{2}$, and $\tilde{x}_{2}$; the fact that $\tilde{p}$ coincides with $G$ at these points allows to verify IC.

\begin{figure}[t!]
\centering
    \includegraphics[width=0.75\textwidth]{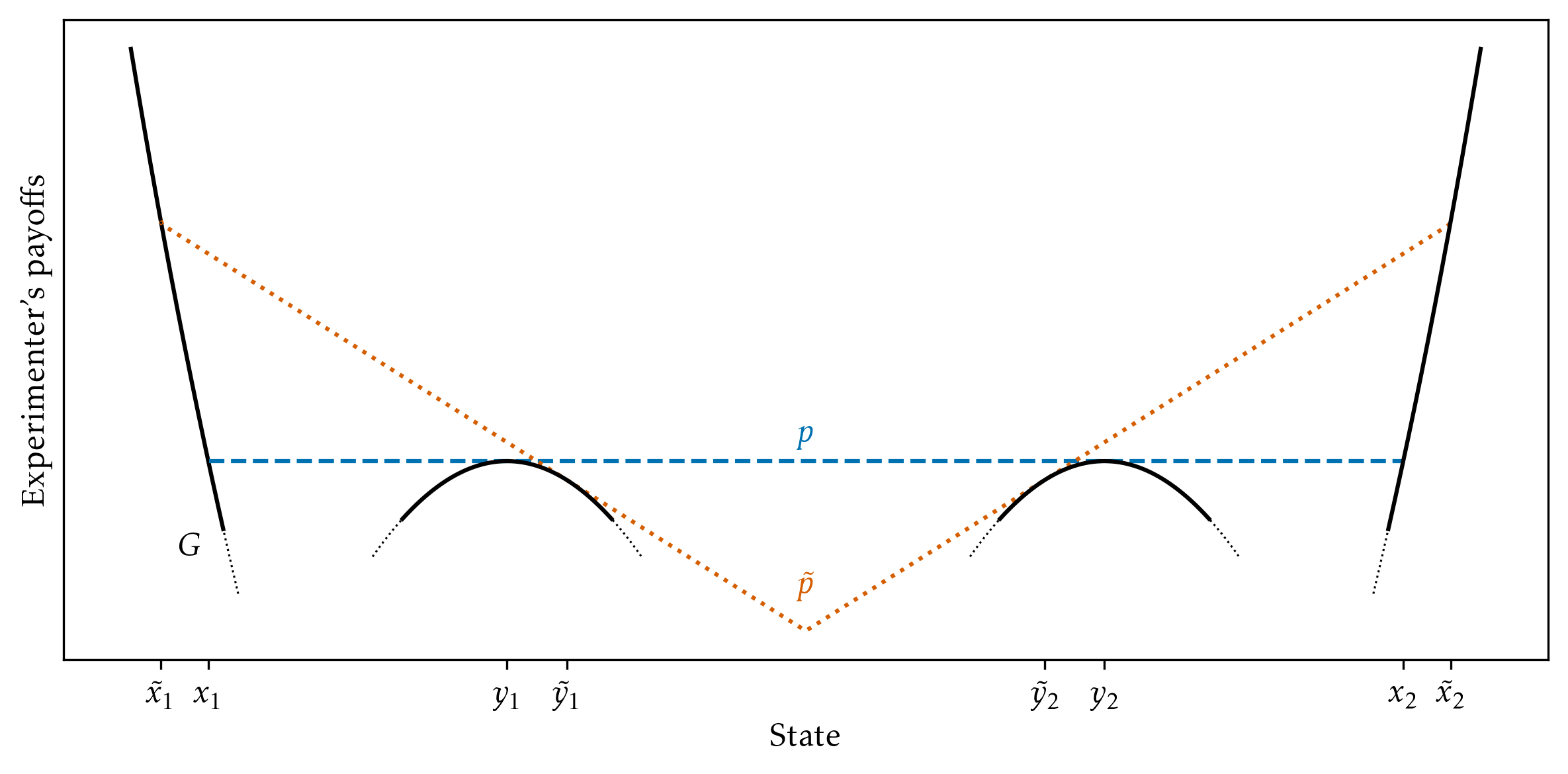}
    \caption{\emph{Two tuples and supporting price functions (colored dashed and dotted). The experimenter's payoff $G$ is shown in black. The points $y_{1}$ and $y_{2}$ lie on two distinct intervals on which $G$ is concave.}}
    \label{fig:smooth_support:bipool}
\end{figure}

\begin{figure}[t!]
    \centering
    \includegraphics[width=0.75\textwidth]{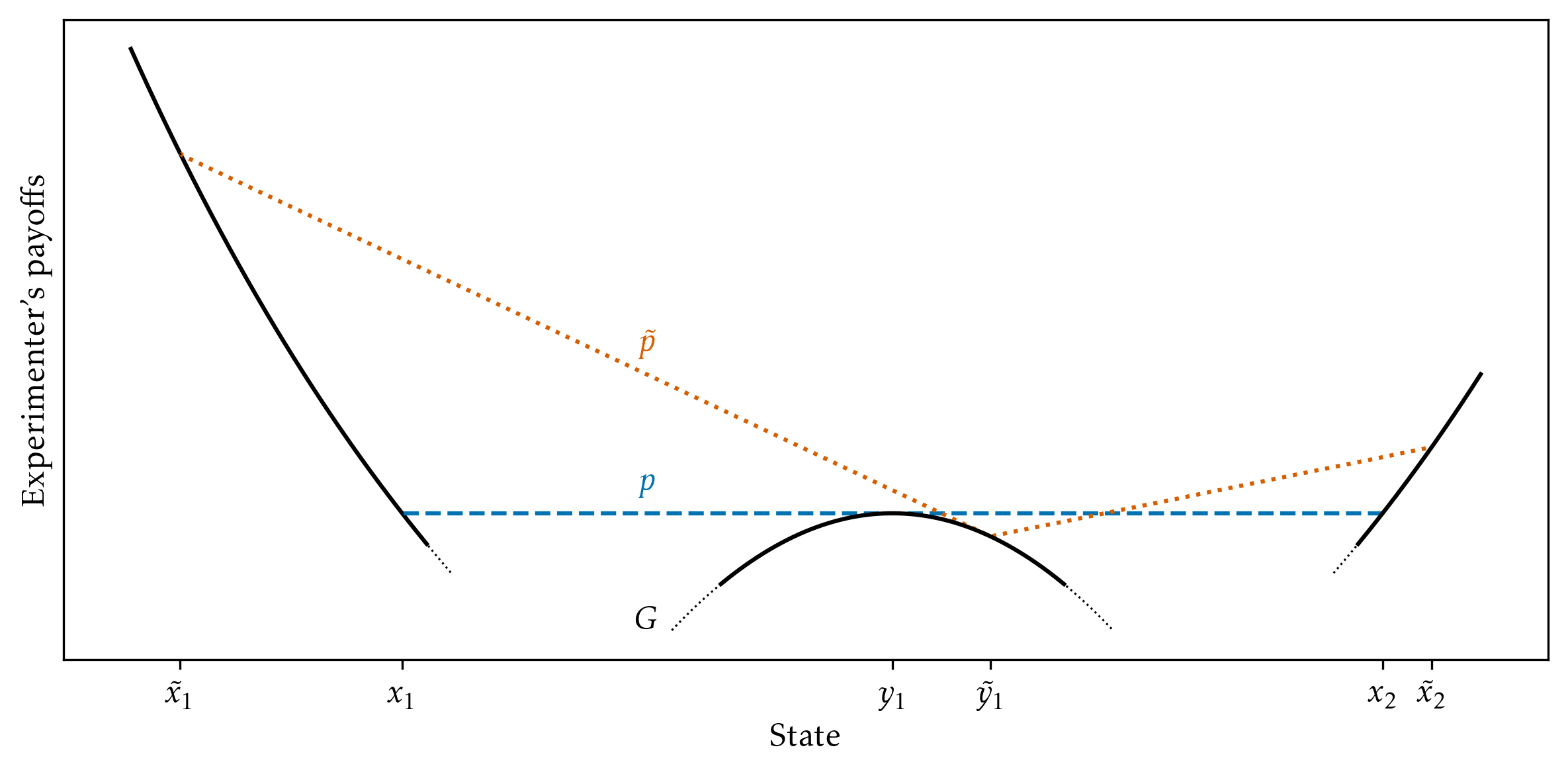}
    \caption{\emph{Two tuples and supporting price functions (colored dashed and dotted). The experimenter's payoff $G$ is shown in black. Both tuples admit a single point that lies in an interval where $G$ is concave.}}
    \label{fig:smooth_support:pool}
\end{figure}

\subsubsection{Incomplete full revelation}
Incomplete full revelation concerns the extreme points of $\MPC(H)$.
Recall, for each extreme point $F$, there are countably many, disjoint \emph{bi-pooling intervals} $[x_{1}, x_{2}]$, such that $F$ pools $[x_{1}, x_{2}]$ to at most two atoms and $I_{F}$ lies strictly below $I_{H}$ on the interior $(x_{1}, x_{2})$.
Outside the union of the bi-pooling intervals, $I_{F}$ coincides with $I_{H}$.
A non-trivial interval $[a, b]$ is \emph{fully revealing} if $I_{F}$ coincides with $I_{H}$ on $[a, b]$.

Note carefully that full revelation intervals are non-trivial by definition.
An extreme point need not admit any full revelation intervals.

\begin{definition}\label{def:incomplete_full_revelation}
    An extreme point $F$ of $\MPC(H)$ admits \emph{incomplete full revelation} if there is a non-degenerate bi-pooling interval $[x_{1}, x_{2}]$ of $F$ and $\delta > 0$ satisfying at least one of the following:
    \begin{enumerate}
        \item $F$ admits two atoms on $(x_{1}, x_{2})$, both $[x_{1} - \delta, x_{1}]$ and $[x_{2}, x_{2} + \delta]$ are full revelation intervals, and $0 < x_{1} < x_{2} < 1$ holds.
        \item $F$ admits exactly one atom on $(x_{1}, x_{2})$, and at least one of the following holds:
        \begin{enumerate}
            \item Both $[x_{1} - \delta, x_{1}]$ and $[x_{2}, x_{2} + \delta]$ are full revelation intervals, and $0 < x_{1} < x_{2} < 1$ holds. 
            \item $[x_{1} - \delta, x_{1}]$ is a full revelation interval, and $0 < x_{1} < x_{2} = 1$ holds.
            \item $[x_{2}, x_{2} + \delta]$ is a full revelation interval, and $0 = x_{1} < x_{2} < 1$ holds.
        \end{enumerate}
    \end{enumerate}
\end{definition}
In words, $F$ fully reveals the state on non-degenerate intervals to the left and right of a bi-pooling interval (cases (1) and (2a)).
Cases (2b) and (2c) are as (2a), except that one neighboring interval is the boundary of the state space.

\subsection{Profitable Restrictions}
Say full delegation is not optimal if there does not exist optimal $F\in\BR(H)$.

There exists an extreme point of $\MPC(H)$ that is designer-preferred among $\BR(H)$.
We refer to such an extreme point as a \emph{DEF} (\textbf{d}esigner-optimal \textbf{e}xtreme point under \textbf{f}ull delegation). 
We provide sufficient conditions such that a given DEF is not an optimal IC experiment; accordingly, full delegation is not optimal.

\begin{theorem}\label{thm:general_perturbation}
    Let the designer's payoff $u_{D}$ be continuously differentiable and strictly convex, let the experimenter's payoff $G$ be continuously differentiable, and let the prior $H$ admit a strictly positive density.
    If the experimenter has locally strict curvature preferences and there exists a DEF with incomplete full revelation, then full delegation is not optimal.
\end{theorem}

\begin{proof}[Proof of \Cref{thm:general_perturbation}]
    Let $F$ be a DEF with incomplete full revelation.
    Let $\delta > 0$ and $[x_{1}, x_{2}]$ be as in the definition of incomplete full revelation. Let $y_{1}$ and $y_{2}$ be the atoms in $[x_{1}, x_{2}]$ (possibly $y_{1} = y_{2}$).
    
    Find a price function $p$ as in Corollary 1 of \citet{dworczak2019simple}, i.e. $p$ is convex and continuous, and it holds $p \geq G$, $\supp F\subseteq \lbrace m\in [0, 1]\colon p(m) = G(m)\rbrace$ and $\int p\de F = \int p \de H$.
    Since $[x_{1}, x_{2}]$ is bi-pooling, Corollary 3 of \citet{arieli2023optimal} tells us that $p$ is affine on $[x_{1}, x_{2}]$.

    We distinguish the cases of \Cref{def:incomplete_full_revelation} concerning the form of $[x_{1}, x_{2}]$.
    In all cases, the proof strategy is to use \Cref{assumption:smooth_support} to find a price function $\tilde{p}$ that touches $G$ close to points at which $p$ equals $G$.
    Using the price-theoretic characterization of IC, we obtain a particular IC experiment $\tilde{F}$ whose support is close to the support of $F$.
    For the particularly constructed experiment, we show the designer is strictly better off than under $F$, which completes the proof.

    \textbf{Case 1.} \emph{$F$ admits two atoms, $y_{1} < y_{2}$, on $(x_{1}, x_{2})$, both $[x_{1} - \delta, x_{1}]$ and $[x_{2}, x_{2} + \delta]$ are full revelation intervals, and $0 < x_{1} < x_{2} < 1$ holds.}

    Since $p$ is affine on $[x_{1}, x_{2}]$ and $y_{1}, y_{2} \in (x_{1}, x_{2})$, both $y_{1}$ and $y_{2}$ must be in the interior of $\Ccv$ (in view of \Cref{assumption:smooth_support}).
    Also, since $I_{F} = I_{H}$ on $[x_{1} - \delta, x_{1}]\cup [x_{2}, x_{2} + \delta]$, the functions $p$ and $G$ agree on this set; in particular, $G$ is convex on this set.
    Thus, $[x_{1} - \delta, x_{1}]\cup [x_{2}, x_{2} + \delta] \subseteq \Cvx$.
    Consequently, the tuple $(x_{1}, y_{1}, y_{2}, x_{2})$ meets the hypothesis of \Cref{assumption:smooth_support}.
    Let $\tilde{T} = (\tilde{x}_{1}, \tilde{y}_{1}, \tilde{y}_{2}, \tilde{x}_{2})$ be as in \Cref{assumption:smooth_support}; for all $r \in (0, 1]$, let $\tilde{p}(\cdot, r)$ be a supporting price function.

    Throughout, we take $r$ sufficiently close to $0$ such that $\tilde{x}_{1}(r) \in [x_{1} - \delta, x_{1}]$ and $\tilde{x}_{2}(r) \in [x_{2}, x_{2} + \delta]$, which is possible by continuity of $\tilde{T}$.

    The function $\tilde{p}(\cdot, r)$ obtained from \Cref{assumption:smooth_support} has domain $[\tilde{x}_{1}(r), \tilde{x}_{2}(r)]$.
    We extend to $[0, 1]$ by setting $\tilde{p}(m, r) = p(m)$ for all $m\notin [\tilde{x}_{1}(r), \tilde{x}_{2}(r)]$.
    The resulting function clearly lies above $G$.
    One may verify that $\tilde{p}(\cdot, r)$ is convex by observing that $\tilde{p}(\cdot, r)$ lies above $G$ on $[\tilde{x}_{1}(r), \tilde{x}_{2}(r)]$, while $\tilde{x}_{1}(r)$ and $\tilde{x}_{2}(r)$, respectively, lie in the intervals $[x_{1} - \delta, x_{1}]$ and $[x_{2}, x_{2} + \delta]$, respectively, on which $p$ coincides with $G$.

    We claim $\tilde{y}_{1}^{\prime}(0) > 0$ and $\tilde{y}_{2}^{\prime}(0) < 0$.
    Indeed, for small $r$ we have $\tilde{y}_{1}(r), \tilde{y}_{2}(r) \in \Ccv$ since $y_{1}$ and $y_{2}$ are in the interior of $\Ccv$.
    The points $(y_{1}, G(y_{1}))$ and $(y_{2}, G(y_{2}))$ in $\mathbb{R}^{2}$ lie on the graph of an affine function with slope $g(y_{1}) = g(y_{2})$.
    Thus, for $\tilde{T}(r)$ to be supported by a convex price function, concavity of $G$ around $y_{1}$ and $y_{2}$ requires $\tilde{y}_{1}(r) \geq y_{1}$ and $\tilde{y}_{2}(r) \leq y_{2}$ for sufficiently small $r$.
    Since $\tilde{y}_{1}$ and $\tilde{y}_{2}$ are assumed to be continuously differentiable with non-zero derivatives at $r= 0$, we conclude $\tilde{y}_{1}^{\prime}(0) > 0$ and $\tilde{y}_{2}^{\prime}(0) < 0$.

    \Cref{fig:smooth_support:bipool} schematically depicts the assumed and derived properties of $p$ and $\tilde{p}(\cdot, r)$ for some fixed $r$.

    Let $q$ be the constant slope of $I_{F}$ on $(y_{1}, y_{2})$, i.e. $q = F(y_{1})$.
    Since $[x_{1}, x_{2}]$ is bi-pooling with atoms $y_{1}$ and $y_{2}$, both $i\in\lbrace 1, 2\rbrace$ satisfy
    \begin{equation*}
        I_{H}(x_{i}) + H(x_{i}) (y_{i} - x_{i}) = I_{F}(y_{i}) + q (y_{i} - y_{i}),
    \end{equation*}
    where the right side of course simply equals $I_{F}(y_{i})$.
    The function $t\mapsto I_{H}(t) + H(t) (y_{i} - t)$ has a non-zero derivative at $t = x_{i}$.
    Therefore, for sufficiently small $r$ the Inverse Function Theorem yields $\tilde{t}_{i}(r)$ such that
    \begin{equation}\label{eq:general_perturbation:bipool:tildet_def}
        I_{H}(\tilde{t}_{i}(r)) + H(\tilde{t}_{i}(r)) (\tilde{y}_{i}(r) - \tilde{t}_{i}) = I_{F}(y_{i}) + q (\tilde{y}_{i}(r) - y_{i}),
    \end{equation}
    Moreover, $\tilde{t}_{i}(r)\to x_{i}$ as $r\to 0$, and $\tilde{t}_{i}(r)$ is differentiable with respect to $r$, with
    \begin{equation*}
        \tilde{t}_{i}^{\prime}(0)
        =
        \tilde{y}_{i}^{\prime}(0) \frac{q - H(x_{i})}{h(x_{i})(y_{i} - x_{i})}
        .
    \end{equation*}
    We have $\tilde{t}_{1}^{\prime}(0) > 0$ and $\tilde{t}_{2}^{\prime}(0) < 0$ since $\tilde{y}_{1}^{\prime}(0) > 0$ and $\tilde{y}_{2}^{\prime}(0) < 0$ and $x_{1} < y_{1}$ and $y_{2} < x_{2}$.\footnote{Also, $H(x_{1}) < q = F(y_{1}) < H(x_{2})$ since the slope of $I_{F}$ is $F(y_{1})$ on $(y_{1}, y_{2})$, while it is $H(x_{i})$ at $x_{i}$, for both $i$.}
    Finally, note that since $I_{F}$ is affine on $[y_{1}, y_{2}]$, the right side of \eqref{eq:general_perturbation:bipool:tildet_def} simply equals $I_{F}(\tilde{y}_{i}(r))$, for all $r$.

    Using $(\tilde{t}_{1}(r), \tilde{x}_{1}(r), \tilde{t}_{2}(r), \tilde{x}_{2}(r))\to (x_{1}, x_{1}, x_{2}, x_{2})$ as $r\to 0$, an application of the Intermediate Value Theorem shows that for $r$ sufficiently small there exist $\tilde{s}_{1}(r) \leq \tilde{x}_{1}(r)$ and $\tilde{s}_{2}(r) \geq \tilde{x}_{2}(r)$ such that
    \begin{equation}\label{eq:general_perturbation:bipool:tildes_def}
        \int_{[\tilde{s}_{1}(r), \tilde{t}_{1}(r)]} (\tilde{x}_{1}(r) - \omega) \de H(\omega)
        =
        \int_{[\tilde{t}_{2}(r), \tilde{s}_{2}(r)]} (\tilde{x}_{2}(r) - \omega) \de H(\omega) = 0
        ,
    \end{equation}
    i.e. $\tilde{x}_{1}(r)$ is the expectation with respect to $H$ conditional on $[\tilde{s}_{1}, \tilde{t}_{1}(r)]$, and analogously for $\tilde{x}_{2}(r)$.
    We have $\tilde{s}_{1}(r)\to x_{1}$ and $\tilde{s}_{2}(r)\to x_{2}$ as $r\to 0$.
    The equations \eqref{eq:general_perturbation:bipool:tildes_def} imply
    \begin{equation*}
        \forall i\in\lbrace 1, 2\rbrace,\quad
        I_{H}(\tilde{s}_{i}(r)) + H(\tilde{s}_{i}(r)) (\tilde{x}_{i}(r) - \tilde{s}_{i}(r)) = I_{H}(\tilde{t}_{i}(r)) + H(\tilde{t}_{i}(r)) (\tilde{x}_{i}(r) - \tilde{t}_{i}(r))
        .
    \end{equation*}
    
    In what follows, fix $r$ sufficiently small such that $\tilde{s}_{1}(r), \tilde{t}_{1}(r), \tilde{t}_{2}(r)$, and $\tilde{s}_{2}(r)$ are well-defined, and $x_{1} - \delta < \tilde{s}_{1}(r) < \tilde{x}_{1}(r) < \tilde{t}_{1}(r) < \tilde{y}_{1}(r) < \tilde{y}_{2}(r) < \tilde{t}_{2}(r) < \tilde{x}_{2}(r) < \tilde{s}_{2}(r) < x_{2} + \delta$.

    We define an experiment $\tilde{F}_{r}$ via its integrated CDF, $I_{\tilde{F}_{r}}$.
    For all $m\in [0, 1]$, let
    \begin{equation*}
        I_{\tilde{F}_{r}}(m) = 
        \begin{cases}
            I_{F}(m) \quad&\mbox{if } m \notin [\tilde{s}_{1}(r), \tilde{s}_{2}(r)],\\
            I_{H}(\tilde{s}_{1}(r)) + H(\tilde{s}_{1}(r))(m - \tilde{s}_{1}(r)) \quad&\mbox{if } m \in [\tilde{s}_{1}(r), \tilde{x}_{1}(r)],\\
            I_{H}(\tilde{t}_{1}(r)) + H(\tilde{t}_{1}(r))(m - \tilde{t}_{1}(r)) \quad&\mbox{if } m \in (\tilde{x}_{1}(r), \tilde{y}_{1}(r)],\\
            I_{F}(m) \quad&\mbox{if } m \in (\tilde{y}_{1}(r), \tilde{y}_{2}(r)],\\
            I_{H}(\tilde{t}_{2}(r)) + H(\tilde{t}_{2}(r))(m - \tilde{t}_{2}(r)) \quad&\mbox{if } m \in (\tilde{y}_{2}(r), \tilde{x}_{2}(r)],\\
            I_{H}(\tilde{s}_{2}(r)) + H(\tilde{s}_{2}(r))(m - \tilde{s}_{2}(r)) \quad&\mbox{if } m \in (\tilde{x}_{2}(r), \tilde{s}_{2}(r)].
        \end{cases}
    \end{equation*}
    The choice of $(\tilde{s}_{1}(r), \tilde{t}_{1}(r), \tilde{t}_{2}(r), \tilde{s}_{2}(r))$ in \eqref{eq:general_perturbation:bipool:tildet_def} and \eqref{eq:general_perturbation:bipool:tildes_def} ensures that $\tilde{F}_{r}$ is an experiment, where we recall that $I_{F}$ is affine on $(y_{1}, y_{2})$ and given by the right side of \eqref{eq:general_perturbation:bipool:tildet_def}.
    \Cref{lemma:price_function_IC} implies that $\tilde{F}_{r}$ is IC since $\supp\tilde{F}_{r}\subseteq\lbrace m\in [0, 1]\colon \tilde{p}(m, r) = G(m)\rbrace$.

    We show that the designer strictly prefers $\tilde{F}_{r}$ to $F$ for $r$ sufficiently small.
    The designer's utility from $\tilde{F}_{r}$ is given by
    \begin{align}
        \label{eq:general_perturbation:bipool:Fprimeutil:1}
        &\int_{x_{1} - \delta}^{\tilde{s}_{1}(r)} u_{D}(\omega)\de H(\omega) +
        (H(\tilde{t}_{1}(r)) - H(\tilde{s}_{1}(r))) u_{D}(\tilde{x}_{1}(r)) + (q - H(\tilde{t}_{1}(r))) u_{D}(\tilde{y}_{1}(r))
        \\
        \label{eq:general_perturbation:bipool:Fprimeutil:2}
        &+ (H(\tilde{t}_{2}(r)) - q) u_{D}(\tilde{y}_{2}(r))
        + (H(\tilde{s}_{2}(r)) - H(\tilde{t}_{2}(r))) u_{D}(\tilde{x}_{2}(r))
        + \int_{\tilde{s}_{2}(r)}^{x_{2} + \delta} u_{D}(\omega)\de H(\omega)
        \\
        \nonumber
        &+
        \int_{[0, 1]\setminus [x_{1} - \delta, x_{2} + \delta]} u_{D}(m) \de F(m)
        .
    \end{align}
    The designer's utility from $F$ is given by
    \begin{align}
        \label{eq:general_perturbation:bipool:Futil:1}
        &\int_{x_{1} - \delta}^{x_{1}} u_{D}(\omega)\de H(\omega)
         + (q - H(x_{1})) u_{D}(y_{1})
        \\
        \label{eq:general_perturbation:bipool:Futil:2}
        &+ (H(x_{2}) - q) u_{D}(y_{2})
        + \int_{x_{2}}^{x_{2} + \delta} u_{D}(\omega)\de H(\omega)
        \\
        \nonumber
        &+
        \int_{[0, 1]\setminus [x_{1} - \delta, x_{2} + \delta]} u_{D}(m) \de F(m)
        .
    \end{align}
    We claim that for $r$ sufficiently small, the sum in \eqref{eq:general_perturbation:bipool:Fprimeutil:1} is strictly larger than  the sum in \eqref{eq:general_perturbation:bipool:Futil:1}, and the sum in \eqref{eq:general_perturbation:bipool:Fprimeutil:2} is strictly larger than  the sum in \eqref{eq:general_perturbation:bipool:Futil:2}.
    The arguments are similar, so we show only the former.
    Dividing by $r$, the difference between \eqref{eq:general_perturbation:bipool:Fprimeutil:1} and \eqref{eq:general_perturbation:bipool:Futil:1} is
    \begin{align}
        & \frac{(u_{D}(\tilde{y}_{1}(r)) - u_{D}(y_{1}))}{r}(q - H(x_{1}))
        + 
        (u_{D}(\tilde{y}_{1}(r)) - u_{D}(\tilde{x}_{1}(r))) \frac{(H(x_{1}) - H(\tilde{t}_{1}(r)))}{r}
        \label{eq:general_perturbation:bipool:comparison:1}
        \\
        & +
        \frac{1}{r}
        \int_{[\tilde{s}_{1}(r), x_{1}]}(u_{D}(\tilde{x}_{1}(r)) - u_{D}(\omega))\de H(\omega)
        \label{eq:general_perturbation:bipool:comparison:2}
    \end{align}
    As $r\to 0$, the sum of the two terms in \eqref{eq:general_perturbation:bipool:comparison:1} converges to
    \begin{align*}
        &\tilde{y}_{1}^{\prime}(0)
        u_{D}^{\prime}(y_{1}) (q - H(x_{1})) - (u_{D}(y_{1}) - u_{D}(x_{1})) h(x_{1}) \tilde{t}_{1}^{\prime}(0)
        \\
        =
        &\tilde{y}_{1}^{\prime}(0)
        \left(
        u_{D}^{\prime}(y_{1})  - \frac{u_{D}(y_{1}) - u_{D}(x_{1})}{y_{1} - x_{1}}
        \right)
        (q - H(x_{1}))
        .
    \end{align*}
    This limit is strictly positive since $u_{D}$ is strictly convex and since $\tilde{y}_{1}^{\prime}(0) > 0$ holds.
    Turning to the integral in \eqref{eq:general_perturbation:bipool:comparison:2}, we argue that this integral vanishes.
    We have
    \begin{align*}
        &
        \frac{1}{r}
        \left\vert
        \int_{[\tilde{s}_{1}(r), x_{1}]}(u_{D}(\tilde{x}_{1}(r)) - u_{D}(\omega))\de H(\omega)
        \right\vert
        \\
        \leq &
        \frac{1}{r} \left\vert \left(u_{D}(\tilde{x}_{1}(r)) - u_{D}(x_{1})\right) \left(H(x_{1}) - H(\tilde{s}_{1}(r))\right) \right\vert
        +
        \frac{1}{r}
        \left\vert
        \int_{[\tilde{s}_{1}(r), x_{1}]}(u_{D}(x_{1}) - u_{D}(\omega))\de H(\omega)
        \right\vert.
    \end{align*}
    The first term in this upper bound, $\frac{1}{r} \left(u_{D}(\tilde{x}_{1}(r)) - u_{D}(x_{1})\right) \left(H(x_{1}) - H(\tilde{s}_{1}(r))\right)$, vanishes as $r\to 0$ since $\tilde{x}_{1}(r)$ is differentiable at $r = 0$, and since $\tilde{s}_{1}(r) \to x_{1}$ as $r\to 0$.
    Consider the second term.    
    Find $L > 0$ such that $\vert u_{D}^{\prime}\vert \leq L$.
    We then have
    \begin{align*}
        &\frac{1}{r}
        \left\vert
        \int_{[\tilde{s}_{1}(r), x_{1}]}(u_{D}(x_{1}) - u_{D}(\omega))\de H(\omega)
        \right\vert
        \\
        \leq
        &
        \frac{1}{r}
        \int_{[\tilde{s}_{1}(r), x_{1}]} \left\vert (u_{D}(x_{1}) - u_{D}(\omega)) \right\vert \de H(\omega)
        \\ 
        =
        &
        \frac{1}{r}
        \int_{[\tilde{s}_{1}(r), x_{1}]} \left\vert \int_{[\omega, x_{1}]} u_{D}^{\prime}(r) \de r \right\vert \de H(\omega)
        \\
        \leq
        &
        \frac{L}{r}
        \int_{[\tilde{s}_{1}(r), x_{1}]} (x_{1} - \omega) \de H(\omega)
        \\
        =
        &
        \frac{L}{r}
        (x_{1} - \tilde{x}_{1}(r)) (H(x_{1}) - H(\tilde{s}_{1}(r)))
        +
        \frac{L}{r}
        \int_{[\tilde{s}_{1}(r), x_{1}]} (\tilde{x}_{1}(r) - \omega)\de H(\omega)
        \\
        =
        &
        \frac{L}{r}
        (x_{1} - \tilde{x}_{1}(r)) (H(x_{1}) - H(\tilde{s}_{1}(r)))
        +
        \frac{L}{r}
        \int_{[x_{1}, \tilde{t}_{1}(r)]} (\omega - \tilde{x}_{1}(r))\de H(\omega)
        \qquad\mbox{(by \eqref{eq:general_perturbation:bipool:tildes_def})}
        \\
        \leq
        &
        \frac{L}{r}
        (x_{1} - \tilde{x}_{1}(r)) (H(x_{1}) - H(\tilde{s}_{1}(r)))
        +
        \frac{L}{r}
        (\tilde{t}_{1}(r) - \tilde{x}_{1}(r)) (H(\tilde{t}_{1}(r)) - H(x_{1}))
    \end{align*}
    Now, $\frac{1}{r} (x_{1} - \tilde{x}_{1}(r)) (H(x_{1}) - H(\tilde{s}_{1}(r)))$ vanishes as $r\to 0$ since $\tilde{x}_{1}$ is differentiable at $r=0$, and since $\tilde{s}_{1}(r) \to x_{1}$ as $r\to 0$.
    Moreover, $\frac{1}{r} (\tilde{t}_{1}(r) - \tilde{x}_{1}(r)) (H(\tilde{t}_{1}(r)) - H(x_{1}))$ vanishes since $\tilde{t}_{1}(r)$ is differentiable at $r=0$, and since $\tilde{t}_{1}(r)$ and $\tilde{x}_{1}(r)$ both converge to $x_{1}$ as $r\to 0$.
    Thus, \eqref{eq:general_perturbation:bipool:comparison:2} vanishes.
    Thus, the sum of \eqref{eq:general_perturbation:bipool:comparison:1} and \eqref{eq:general_perturbation:bipool:comparison:2} is strictly positive for $r$ sufficiently close to $0$.

    \bigskip

    \textbf{Case 2.} \emph{$F$ admits exactly one atom $y_{1}$ on $(x_{1}, x_{2})$, both $[x_{1} - \delta, x_{1}]$ and $[x_{2}, x_{2} + \delta]$ are full revelation intervals, and $0 < x_{1} < x_{2} < 1$ holds.}

    Similarly to the previous case, we have $y_{1} \in \interior \Ccv$, and $[x_{1} - \delta , x_{1}] \cup [x_{2}, x_{2} + \delta] \subseteq\Cvx$.
    Thus, the tuple $(x_{1}, y_{1}, y_{1}, x_{2})$ meets the hypothesis of \Cref{assumption:smooth_support}.
    Let $\tilde{T} = (\tilde{x}_{1}, \tilde{y}_{1}, \tilde{y}_{2}, \tilde{x}_{2})$ be as postulated by \Cref{assumption:smooth_support}.
    For all $r \in (0, 1]$, let $\tilde{p}(\cdot, r)$ be a supporting price function.
    The component $\tilde{y}_{2}(r)$ shall be irrelevant for our argument.
    
    Recall that $\tilde{y}_{1}$ is continuously differentiable with non-zero derivative at $r = 0$.
    Without loss, let $\tilde{y}_{1}^{\prime}(0) > 0$, the other case being similar.

    Throughout, take $r$ sufficiently small such that $\tilde{y}_{1}(r) > y_{1}$ and $\tilde{x}_{1}(r) \in [x_{1} - \delta, x_{1}]$ and $\tilde{x}_{2}(r) \in [x_{2}, x_{2} + \delta]$, where we recall that $\tilde{x}_{1}(r) \leq x_{1}$ and $\tilde{x}_{2}(r) \geq x_{2}$ follow directly from \Cref{assumption:smooth_support}.

    \Cref{fig:smooth_support:pool} schematically depicts the assumed and derived properties of $p$ and $\tilde{p}(\cdot, r)$ for some fixed $r$.
    
    As in the proof for the previous case, we extend $\tilde{p}(\cdot, r)$ to $[0, 1]$ by setting $\tilde{p}(m, r) = p(m)$ for all $m\notin [\tilde{x}_{1}(r), \tilde{x}_{2}(r)]$.
    The resulting function is continuous, convex, and lies above $G$.

    Similar to the previous case, we next construct three thresholds $\tilde{s}_{1}(r), \tilde{t}_{1}(r), \tilde{s}_{2}(r)$. 
    (The reader is invited to think of the fourth threshold $\tilde{t}_{2}(r)$ from the previous case as being fixed to $x_{2}$.)
    Note, it holds
    \begin{equation*}
        I_{H}(x_{1}) + H(x_{1}) (y_{1} - x_{1}) = I_{H}(x_{2}) + H(x_{2}) (y_{1} - x_{2}),
    \end{equation*}
    and the function $t\mapsto I_{H}(t) + H(t) (y_{1} - t)$ has a strictly positive derivative at $x_{1}$.
    Therefore, an application of the Inverse Function Theorem shows that for $r$ sufficiently small there exists $\tilde{t}_{1}(r)$ such that
    \begin{align*}
        I_{H}(\tilde{t}_{1}(r)) + H(\tilde{t}_{1}(r)) (\tilde{y}_{1}(r) - \tilde{t}_{1}(r)) &= I_{H}(x_{2}) + H(x_{2}) (\tilde{y}_{1}(r) - x_{2})
        \\
        \Leftrightarrow\quad
        \int_{[\tilde{t}_{1}(r), x_{2}]} (\tilde{y}_{1}(r) - \omega) \de H(\omega) &= 0.
    \end{align*}
    Moreover, $\tilde{t}_{1}(r)\to x_{1}$ as $r\to 0$, and $\tilde{t}_{1}(r)$ is differentiable with respect to $r$, with
    \begin{equation*}
        \tilde{t}_{1}^{\prime}(0)
        =
        \tilde{y}_{1}^{\prime}(0) \frac{H(x_{2}) - H(x_{1})}{h(x_{1})(y_{1} - x_{1})}
        .
    \end{equation*}
    Note $\tilde{t}_{1}^{\prime}(0) > 0$ since $\tilde{y}_{1}^{\prime}(0) > 0$ and $x_{1} < y_{1} < x_{2}$.

    Using $(\tilde{x}_{1}(r), \tilde{t}_{1}(r), \tilde{x}_{2}(r))\to (x_{1}, x_{1}, x_{2})$ as $r\to 0$ and $\tilde{x}_{1}(r) \leq x_{1} < \tilde{t}_{1}(r)$ and $x_{2} \leq \tilde{x}_{2}(r)$, an application of the Intermediate Value Theorem shows that for $r$ sufficiently small there exist $\tilde{s}_{1}(r)$ and $\tilde{s}_{2}(r)$ such that
    \begin{equation}\label{eq:general_perturbation:pool:tildes_def}
        \int_{[\tilde{s}_{1}(r), \tilde{t}_{1}(r)]} (\tilde{x}_{1}(r) - \omega) \de H(\omega) 
        =
        \int_{[x_{2}, \tilde{s}_{2}(r)]} (\tilde{x}_{2}(r) - \omega) \de H(\omega) = 0
        .
    \end{equation}
    We have $\tilde{s}_{1}(r) \leq \tilde{x}_{1}(r)$ and $\tilde{x}_{2}(r) \leq \tilde{s}_{2}(r)$.
    Moreover, $\tilde{s}_{1}(r)\to x_{1}$ and $\tilde{s}_{2}(r)\to x_{2}$ as $r\to 0$.
    The equations \eqref{eq:general_perturbation:pool:tildes_def} are equivalent to
    \begin{align*}
        I_{H}(\tilde{s}_{1}(r)) + H(\tilde{s}_{1}(r)) (\tilde{x}_{1}(r) - \tilde{s}_{1}(r)) &= I_{H}(\tilde{t}_{1}(r)) + H(\tilde{t}_{1}(r)) (\tilde{x}_{1}(r) - \tilde{t}_{1}(r))
        ,
        \\
        I_{H}(x_{2}) + H(x_{2}) (\tilde{x}_{2}(r) - x_{2}) &= I_{H}(\tilde{s}_{2}(r)) + H(\tilde{s}_{2}(r)) (\tilde{x}_{2}(r) - \tilde{s}_{2}(r))
    \end{align*}

    Define the experiment $\tilde{F}_{r}$ via its integrated CDF as follows:
    \begin{equation*}
    \forall m\in [0, 1]\quad
    I_{\tilde{F}_{r}}(m) = 
    \begin{cases}
        I_{F}(m) \quad&\mbox{if } m \notin [\tilde{s}_{1}(r), \tilde{s}_{2}(r)],\\
        I_{H}(\tilde{s}_{1}(r)) + H(\tilde{s}_{1}(r))(m - \tilde{s}_{1}(r)) \quad&\mbox{if } m \in [\tilde{s}_{1}(r), \tilde{x}_{1}(r)],\\
        I_{H}(\tilde{t}_{1}(r)) + H(\tilde{t}_{1}(r))(m - \tilde{t}_{1}(r)) \quad&\mbox{if } m \in (\tilde{x}_{1}(r), \tilde{y}_{1}(r)],\\
        I_{H}(x_{2}) + H(x_{2}) (m - x_{2}) \quad&\mbox{if } m \in (\tilde{y}_{1}(r), \tilde{x}_{2}(r)],\\
        I_{H}(\tilde{s}_{2}(r)) + H(\tilde{s}_{2}(r))(m - \tilde{s}_{2}(r)) \quad&\mbox{if } m \in (\tilde{x}_{2}(r), \tilde{s}_{2}(r)].
    \end{cases} 
    \end{equation*}
    The choice of $\tilde{t}_{1}(r)$, $\tilde{s}_{1}(r)$, and $\tilde{s}_{2}(r)$ ensure that $\tilde{F}_{r}$ is an experiment.
    \Cref{lemma:price_function_IC} implies that $\tilde{F}_{r}$ is IC since $\supp \tilde{F}_{r} \subseteq\lbrace m\in [0, 1]\colon \tilde{p}(m, r) = G(m)\rbrace$.
    We show the designer is strictly better off under $\tilde{F}_{r}$ than $F$ for $r$ sufficiently small.

The designer's utility from $\tilde{F}_{r}$, denoted $U_{D}(\tilde{F}_{r})$, is given by
\begin{align*}
    U_{D}(\tilde{F}_{r})
    =
    (H(\tilde{t}_{1}(r)) - H(\tilde{s}_{1}(r))) u_{D}(\tilde{x}_{1}(r))
    +
    (H(x_{2}) - H(\tilde{t}_{1}(r))) u_{D}(\tilde{y}_{1}(r))
    \\
    +
    (H(\tilde{s}_{2}(r)) - H(x_{2})) u_{D}(\tilde{x}_{2}(r))
    +
    \int_{[0, 1]\setminus [\tilde{s}_{1}(r), \tilde{s}_{2}(r)]} u_{D}(m)\de F(m)
\end{align*}
The designer's utility from $F$, denoted $U_{D}(F)$, is given by
\begin{align*}
    U_{D}(F)=
    (H(x_{2}) - H(x_{1})) u_{D}(y_{1})
    +
    \int_{[\tilde{s}_{1}(r), x_{1}]} u_{D}(\omega) \de H(\omega)
    \\
    +
    \int_{[x_{2}, \tilde{s}_{2}(r)]} u_{D}(\omega) \de H(\omega)
    +
    \int_{[0, 1]\setminus [\tilde{s}_{1}(r), \tilde{s}_{2}(r)]} u_{D}(m) \de F(m).
\end{align*}
The difference $\frac{1}{r}(U_{D}(\tilde{F}_{r}) - U_{D}(F))$ equals
\begin{align}
    \label{eq:general_perturbation:pool:difference:1}
    &(H(x_{2}) - H(x_{1}))\frac{u_{D}(\tilde{y}_{1}(r)) - u_{D}(y_{1})}{r}
    + \frac{(H(x_{1}) - H(\tilde{t}_{1}(r)))}{r} (u_{D}(\tilde{y}_{1}(r)) - u_{D}(\tilde{x}_{1}(r)))
    \\
    \label{eq:general_perturbation:pool:difference:2}
    &+
    \frac{1}{r}
    \int_{[\tilde{s}_{1}(r), x_{1}]} (u_{D}(\tilde{x}_{1}(r)) - u_{D}(\omega)) \de H(\omega)
    \\
    \label{eq:general_perturbation:pool:difference:3}
    &+ \frac{1}{r}
    \int_{[x_{2}, \tilde{s}_{2}(r)]} (u_{D}(\tilde{x}_{2}(r)) - u_{D}(\omega)) \de H(\omega)
    .
\end{align}
We argue that, as $r\to 0$, the sum in \eqref{eq:general_perturbation:pool:difference:1} converges to a strictly positive term, whereas both \eqref{eq:general_perturbation:pool:difference:2} and \eqref{eq:general_perturbation:pool:difference:3} vanish.

The limit of \eqref{eq:general_perturbation:pool:difference:1} as $r\to 0$ equals
\begin{equation*}
    \left(H(x_{2}) - H(x_{1})\right) \tilde{y}_{1}^{\prime}(0) \left(u_{D}^{\prime}(y_{1}) - \frac{u_{D}(y_{1}) - u_{D}(x_{1})}{y_{1} - x_{1}}\right).
\end{equation*}
This limit is strictly positive since $u_{D}$ is strictly convex and $\tilde{y}_{1}^{\prime}(0) > 0$.

Next, the integral in \eqref{eq:general_perturbation:pool:difference:2} vanishes; the argument is analogous to the argument that in the previous \textbf{Case 1} the integral \eqref{eq:general_perturbation:bipool:comparison:2} vanishes.
We have
    \begin{align*}
        &
        \frac{1}{r}
        \left\vert
        \int_{[\tilde{s}_{1}(r), x_{1}]}(u_{D}(\tilde{x}_{1}(r)) - u_{D}(\omega))\de H(\omega)
        \right\vert
        \\
        \leq &
        \frac{1}{r} \left\vert \left(u_{D}(\tilde{x}_{1}(r)) - u_{D}(x_{1})\right) \left(H(x_{1}) - H(\tilde{s}_{1}(r))\right) \right\vert
        +
        \frac{1}{r}
        \left\vert
        \int_{[\tilde{s}_{1}(r), x_{1}]}(u_{D}(x_{1}) - u_{D}(\omega))\de H(\omega)
        \right\vert.
    \end{align*}
    The first term in this upper bound, $\frac{1}{r} \left(u_{D}(\tilde{x}_{1}(r)) - u_{D}(x_{1})\right) \left(H(x_{1}) - H(\tilde{s}_{1}(r))\right)$, vanishes as $r\to 0$ since $\tilde{x}_{1}(r)$ is differentiable at $r = 0$, and since $\tilde{s}_{1}(r) \to x_{1}$ as $r\to 0$.
    Consider the second term.    
    Find $L > 0$ such that $\vert u_{D}^{\prime}\vert \leq L$.
    We have
    \begin{align*}
        &\frac{1}{r}
        \left\vert
        \int_{[\tilde{s}_{1}(r), x_{1}]}(u_{D}(x_{1}) - u_{D}(\omega))\de H(\omega)
        \right\vert
        \\
        \leq
        &
        \frac{1}{r}
        \int_{[\tilde{s}_{1}(r), x_{1}]} \left\vert u_{D}(x_{1}) - u_{D}(\omega)) \right\vert \de H(\omega)
        \\ 
        =
        &
        \frac{1}{r}
        \int_{[\tilde{s}_{1}(r), x_{1}]} \left\vert \int_{[\omega, x_{1}]} u_{D}^{\prime}(r) \de r \right\vert \de H(\omega)
        \\
        \leq
        &
        \frac{L}{r}
        \int_{[\tilde{s}_{1}(r), x_{1}]} (x_{1} - \omega) \de H(\omega)
        \\
        =
        &
        \frac{L}{r}
        (x_{1} - \tilde{x}_{1}(r)) (H(x_{1}) - H(\tilde{s}_{1}(r)))
        +
        \frac{L}{r}
        \int_{[\tilde{s}_{1}(r), x_{1}]} (\tilde{x}_{1}(r) - \omega)\de H(\omega)
        \\
        =
        &
        \frac{L}{r}
        (x_{1} - \tilde{x}_{1}(r)) (H(x_{1}) - H(\tilde{s}_{1}(r)))
        +
        \frac{L}{r}
        \int_{[x_{1}, \tilde{t}_{1}(r)]} (\omega - \tilde{x}_{1}(r))\de H(\omega)
        \qquad\mbox{(by \eqref{eq:general_perturbation:pool:tildes_def})}
        \\
        \leq
        &
        \frac{L}{r}
        (x_{1} - \tilde{x}_{1}(r)) (H(x_{1}) - H(\tilde{s}_{1}(r)))
        +
        \frac{L}{r}
        (\tilde{t}_{1}(r) - \tilde{x}_{1}(r)) (H(\tilde{t}_{1}(r)) - H(x_{1}))
    \end{align*}
    Now, $\frac{1}{r} (x_{1} - \tilde{x}_{1}(r)) (H(x_{1}) - H(\tilde{s}_{1}(r)))$ vanishes as $r\to 0$ since $\tilde{x}_{1}$ is differentiable at $r=0$, and since $\tilde{s}_{1}(r) \to x_{1}$ as $r\to 0$.
    Moreover, $\frac{1}{r} (\tilde{t}_{1}(r) - \tilde{x}_{1}(r)) (H(\tilde{t}_{1}(r)) - H(x_{1}))$ vanishes since $\tilde{t}_{1}(r)$ is differentiable at $r=0$, and since $\tilde{t}_{1}(r)$ and $\tilde{x}_{1}(r)$ both converge to $x_{1}$ as $r\to 0$.
    Thus, \eqref{eq:general_perturbation:pool:difference:2} vanishes.

Finally, the integral \eqref{eq:general_perturbation:pool:difference:3} vanishes; the argument is simpler than but analogous to the argument that \eqref{eq:general_perturbation:pool:difference:2} vanishes.
We have
\begin{align*}
    \left\vert \eqref{eq:general_perturbation:pool:difference:3} \right\vert
    =
    &
    \frac{1}{r}
    \left\vert
    \int_{[x_{2}, \tilde{s}_{2}(r)]} (u_{D}(\tilde{x}_{2}(r)) - u_{D}(\omega)) \de H(\omega)
    \right\vert
    \\
    \leq &
    \frac{1}{r} \left\vert\left(u_{D}(\tilde{x}_{2}(r)) - u_{D}(x_{2})\right) \left(H(\tilde{s}_{2}(r)) - H(x_{2})\right) \right\vert
    +
    \frac{1}{r}
    \left\vert
    \int_{[x_{2}, \tilde{s}_{2}(r)]} (u_{D}(x_{2}) - u_{D}(\omega)) \de H(\omega)
    \right\vert.
\end{align*}
The first term in this upper bound, $\frac{1}{r} \left(u_{D}(\tilde{x}_{2}(r)) - u_{D}(x_{2})\right) \left(H(\tilde{s}_{2}(r)) - H(x_{2})\right)$, vanishes as $r\to 0$ since $\tilde{x}_{2}(r)$ is differentiable at $r = 0$, and since $\tilde{s}_{2}(r) \to x_{2}$ as $r\to 0$.
Consider the second term.    
Find $L > 0$ such that $\vert u_{D}^{\prime}\vert \leq L$.
We have
\begin{align*}
    \frac{1}{r}
    \left\vert
    \int_{[x_{2}, \tilde{s}_{2}(r)]} (u_{D}(x_{2}) - u_{D}(\omega)) \de H(\omega)
    \right\vert
    \leq
    &
    \frac{1}{r}
    \int_{[x_{2}, \tilde{s}_{2}(r)]} \left\vert (u_{D}(x_{2}) - u_{D}(\omega)) \right\vert \de H(\omega)
    \\ 
    =
    &
    \frac{1}{r}
    \int_{[x_{2}, \tilde{s}_{2}(r)]} \left\vert - \int_{[x_{2}, \omega]} u_{D}^{\prime}(r) \de r\right\vert \de H(\omega)
    \\
    \leq
    &
    \frac{L}{r}
    \int_{[x_{2}, \tilde{s}_{2}(r)]} (\omega - x_{2}) \de H(\omega)
    \\
    =
    &
    \frac{L}{r}
    \int_{[x_{2}, \tilde{s}_{2}(r)]} (\tilde{x}_{2}(r) - x_{2}) \de H(\omega)
    \qquad\mbox{(by \eqref{eq:general_perturbation:pool:tildes_def})}
    \\
    =
    &
    \frac{L}{r}
    (\tilde{x}_{2}(r) - x_{2}) (H(\tilde{s}_{2}(r)) - H(x_{2})),
\end{align*}
Now, $\frac{L}{r}
    (\tilde{x}_{2}(r) - x_{2}) (H(\tilde{s}_{2}(r)) - H(x_{2}))$ vanishes as $r\to 0$ since $\tilde{x}_{2}$ is differentiable at $r = 0$ and since $\tilde{s}_{2}(r) \to x_{2}$ as $r \to 0$.
Thus, \eqref{eq:general_perturbation:pool:difference:3} vanishes.

\bigskip

\textbf{Case 3.} \emph{$F$ admits exactly one atom $y_{1}$ on $(x_{1}, x_{2})$, the interval $[x_{1} - \delta, x_{1}]$ is a full revelation interval, and $0 < x_{1} < x_{2} = 1$ holds.}

The proof is similar to the proof of \Cref{thm:profitable_restriction}, and hence we shall be brief.
Abbreviate $y = y_{1}$ and $x = x_{1}$.
As before, $y \in \interior\Ccv$, and $[x - \delta, x] \subseteq \Cvx$.
In this case, the price $p$ is affine on $[x, 1]$; in particular, $G(x) + g(y)(y - x) = G(y)$.
For $\tilde{y}$ converging to $y$ from above, consider the affine function $t\mapsto G(\tilde{y}) + g(\tilde{y})(t - \tilde{y})$.
For $\tilde{y}$ close to $y$, this affine function intersects $G$ at a point $\tilde{x}$ in $[x - \delta, x]$ since $y \in \interior\Ccv$, and $[x - \delta, x] \subseteq \Cvx$.
Consider the function $\tilde{p}$ that coincides with $p$ up to $\tilde{x}$, then equals the affine function up to $\tilde{y}$, and finally continues with a sufficiently steep slope such that $\tilde{p}$ is convex and lies above $G$ on $[\tilde{y}, 1]$.
The function $\tilde{p}$ is convex and lies above $G$ on all of $[0, 1]$.
Now find $\tilde{t}$ such that $\tilde{y} = \mathbb{E}_{H}[\bm{\omega}\mid\bm{\omega}\in [\tilde{t}, 1]]$ and $\tilde{x} = \mathbb{E}_{H}[\bm{\omega}\mid \bm{\omega}\in [\tilde{s}, \tilde{t}]]$, which is possible for $\tilde{y}$ sufficiently close to $y$.
Finally, let $\tilde{F}$ be the experiment that coincides with $F$ up to $\tilde{s}$, then pools $H$ on $[\tilde{s}, \tilde{t}]$ to $\tilde{x}$, and pools $H$ on $[\tilde{t}, 1]$ to $\tilde{y}$.
For $\tilde{y}$ sufficiently close to $y$, the point $\tilde{s}$ is in $[x - \delta, x]$, implying that $\tilde{F}$ is supported on points $m$ such that $\tilde{p}(m) = G(m)$.
In particular $\tilde{F}$ is IC.
A calculation similar to the one in the proof of \Cref{thm:profitable_restriction} shows that the designer is strictly better off under $\tilde{F}$ than under $F$.

\textbf{Case 4.} \emph{$F$ admits exactly one atom $y_{1}$ on $(x_{1}, x_{2})$, the interval $[x_{2}, x_{2} + \delta]$ is a full revelation interval, and $0 = x_{1} < x_{2} < 1$ holds.} Analogous to the previous case.
\end{proof}

\end{document}